\documentclass[12pt]{article}

\usepackage{amsmath}
\usepackage{amssymb}
\usepackage{amsthm}
\usepackage{bbm} 
\usepackage{fdsymbol} 

\usepackage[misc]{ifsym}
\usepackage[T1]{fontenc}
\usepackage{lipsum}

\usepackage{tikz}
\usetikzlibrary{decorations.pathmorphing}
\usetikzlibrary{calc}

\usepackage{enumitem}
\setlist{wide=0pt,labelsep=2pt}
\usepackage{graphicx}
\usepackage{pifont}
\usepackage{multirow}
\usepackage{parskip}

\usepackage{changepage}
\usepackage[a4paper,margin=3cm,top=3cm,bottom=3cm]{geometry}

\usepackage{fancyhdr}
\usepackage{calligra}
\newcommand{\handwriting}{\calligra}

\fancypagestyle{plain}{
  \fancyhf{}
  \fancyfoot[C]{\Large\bf\handwriting \thepage}

  \setlength{\footskip}{1.5cm}
}

\usepackage{xcolor}
\definecolor{colK}{RGB}{221,160,221}
\definecolor{colKS}{RGB}{128,0,128}
\definecolor{colAMG}{RGB}{0,0,205}
\definecolor{colAMGo}{RGB}{30, 144, 255}
\usepackage{hyperref}
\hypersetup{
colorlinks=true,
breaklinks=true, 
allcolors= red
}
\usepackage{cleveref}

\DeclareMathOperator*{\argmax}{arg\,max}
\DeclareMathOperator*{\argmin}{arg\,min}

\makeatletter
\@tfor\letter:=ABCDEFGHJKLMNPQRSTUVWXYZ\do{ 
  \expandafter\edef\csname\letter\endcsname{\noexpand\mathbb{\letter}}%
}
\makeatother%

\makeatletter
\@tfor\letter:=ABCDEFGHIJKLMNOPQRSTUVWXYZ\do{%
  \expandafter\edef\csname cal\letter\endcsname{\noexpand\mathcal{\letter}}%
}
\makeatother%

\newcommand{\var}{\mathbb{V}\textrm{\normalfont ar}}
\newcommand{\ud}{\mathrm{d}}
\newcommand{\LL}{\mathbb{L}}
\newcommand{\I}{\mathbbm{1}}

\newcommand{\todist}{\overset{\textrm{\normalfont d}}{\longrightarrow}}
\newcommand{\toas}{\overset{\textrm{\normalfont a.s.}}{\longrightarrow}}
\newcommand{\toprob}{\overset{\P}{\longrightarrow}}

\newcommand{\invCT}{\mu^{\text{\normalfont\tiny CT}}}
\newcommand{\lip}[1]{[#1]_{\text{\normalfont\tiny Lip}}}
\newcommand{\tvnorm}[1]{\|#1\|_{\text{\normalfont\tiny TV}}}

\newcommand{\lambdahat}[1]{\widehat{\lambda}_{#1}} 
\newcommand{\lambdak}[1]{\widehat{\lambda}_{#1}^{\textcolor{colK}{\text{\normalfont\footnotesize $\clubsuit$}}}}%
\newcommand{\lambdaks}[1]{\widehat{\lambda}_{#1}^{\textcolor{colKS}{\text{\normalfont\footnotesize $\vardiamondsuit$}}}}%
\newcommand{\lambdaamgo}[1]{\mathring{\lambda}_{#1}^{\textcolor{colAMGo}{\text{\normalfont\footnotesize $\spadesuit$}}}}%
\newcommand{\lambdaamg}[1]{\widehat{\lambda}_{#1}^{\textcolor{colAMG}{\text{\normalfont\footnotesize $\spadesuit$}}}}%

\newcommand{\sigmak}[1]{\sigma_{\textcolor{colK}{\text{\normalfont\footnotesize $\clubsuit$}}}^2(#1)}
\newcommand{\sigmaks}[1]{\sigma_{\textcolor{colKS}{\text{\normalfont\footnotesize $\vardiamondsuit$}}}^2(#1)}
\newcommand{\sigmaamg}[1]{\sigma_{\textcolor{colAMGo}{\text{\normalfont\footnotesize $\spadesuit$}}}^2(#1)}

\newcommand{\hk}[1]{h_{#1}^{\textcolor{colK}{\text{\normalfont\footnotesize $\clubsuit$}}}}
\newcommand{\hks}[1]{h_{#1}^{\textcolor{colKS}{\text{\normalfont\footnotesize $\vardiamondsuit$}}}}
\newcommand{\hamgs}[1]{h_{#1}^{\textcolor{colAMGo}{\text{\normalfont\footnotesize $\spadesuit$}}\textrm{\normalfont\tiny s}}}
\newcommand{\hamgt}[1]{h_{#1}^{\textcolor{colAMGo}{\text{\normalfont\footnotesize $\spadesuit$}}\textrm{\normalfont\tiny t}}}

\newcommand{\Dk}[1]{D^{\text{\normalfont\tiny K}}(#1)}
\newcommand{\Dks}[1]{D^{\text{\normalfont\tiny KS}}(#1)}

\newcommand{\circledletter}[3]{
    \tikz[baseline=(char.base)]{
        \node[draw=#1, inner sep=1pt, thick, minimum width=5pt, minimum height=10pt,align=center] (char) {\textcolor{#2}{\normalfont\scriptsize #3}};
    }
}

\newtheorem{assumptions}{Assumptions}[section]
\newtheorem{assumption}[assumptions]{Assumption}
\newtheorem{theorem}{Theorem}[section]

\newtheorem{lemma}{Lemma}[section]
\newtheorem{remark}{Remark}[section]

\begin{document}

\title{
\vspace{-2em}%
\begin{adjustwidth}{-0.5cm}{-0.5cm}\centering\large\textsc{Asymptotic Analysis and Practical Evaluation of Jump Rate Estimators in Piecewise-Deterministic Markov Processes}
\end{adjustwidth}%
}

\author{Romain Aza\"{\i}s and Solune Denis}
\date{}
\maketitle

{\footnotesize
\noindent\makebox[\linewidth]{\rule{\textwidth}{0.4pt}}\\[5pt]
{\normalsize\handwriting\textbf{Abstract}}\\
Piecewise-deterministic Markov processes (PDMPs) offer a powerful stochastic modeling framework that combines deterministic trajectories with random perturbations at random times.  Estimating their local characteristics (particularly the jump rate) is an important yet challenging task. In recent years, non-parametric methods for jump rate inference have been developed, but these approaches often rely on distinct theoretical frameworks, complicating direct comparisons.
In this paper, we propose a unified framework to standardize and consolidate state-of-the-art approaches. We establish new results on consistency and asymptotic normality within this framework, enabling rigorous theoretical comparisons of convergence rates and asymptotic variances. Notably, we demonstrate that no single method uniformly outperforms the others, even within the same model. These theoretical insights are validated through numerical simulations using a representative PDMP application: the TCP model. Furthermore, we extend the comparison to real-world data, focusing on cell growth and division dynamics in \textit{Escherichia coli}.
This work enhances the theoretical understanding of PDMP inference while offering practical insights into the relative strengths and limitations of existing methods.

\vspace{0.5em}

{\normalsize\handwriting\textbf{Keywords}}\\
Piecewise-deterministic Markov process; Growth-fragmentation model; Jump rate; Non-parametric estimation; Consistency; Central limit theorem; Vector-valued martingale

\vspace{0.5em}

{\normalsize\handwriting\textbf{Author affiliations}}\\
Romain Aza\"{i}s, Inria Lyon, France \\
Solune Denis, Univ Angers, CNRS, LAREMA, SFR MATHSTIC, F-49000 Angers, France

\vspace{0.5em}

{\normalsize\handwriting\textbf{Acknowledgment}}\\
This work has been supported by the Inria Action Exploratoire ALAMO. Solune Denis also thanks the Chair Stress Test, RISK Management and Financial Steering of the Foundation \'Ecole Polytechnique, and the France 2030 program Centre Henri Lebesgue ANR-11-LABX-0020-01. \\
 \makebox[\linewidth]{\rule{\textwidth}{0.4pt}}
}

\section{Introduction}

\subsection{Piecewise-deterministic Markov processes}

Piecewise-deterministic Markov processes (commonly abbreviated as PDMPs) were introduced by Davis in \cite{D84,D93} as a broad category of continuous-time stochastic models that exclude diffusion. These processes are well-suited for modeling deterministic dynamics where randomness manifests through discrete events. The motion of a PDMP $X_t$ (with jump times $T_n$) on $\R^d$ (endowed with the Borel algebra $\calB(\R^d)$) is defined from its local characteristics:
\begin{itemize}[wide=4.5pt,labelsep=4.5pt]
\item $\Phi:\R\times\R^d\to\R^d$ is the deterministic flow, which satisfies the semi-group property,
$$\forall\,x\in\R^d,~\forall\,t,\,s\in\R,~\Phi(t+s|x) = \Phi(s|\Phi(t|x)).$$
In some contexts, $\Phi(\cdot|x)$ will also be denoted $\Phi_x$.
\item $\lambda:\R^d\to\R_+$ is the jump rate, related to the flow by the following condition,
$$\forall\,x\in E,~\exists\,\varepsilon>0,~\int_0^\varepsilon \lambda(\Phi(t|x))\ud t <\infty .$$
\item $Q:\calB(\R^d)\times\R^d\to[0,1]$ is the transition kernel such that,
$$\forall\,x\in\R^d,~Q(\R^d\setminus\{x\}|x) = 1.$$
\end{itemize}

The jump mechanism, that is, the discrete part of the dynamics, can be expressed as follows. Starting from $X_{0}$ at $T_0=0$, for any integer $n$, for any bounded measurable test function $\varphi$,
\begin{eqnarray*}
\P(T_{n+1}-T_n>t\,|\,X_{T_n},T_n) &=& \exp\left(-\int_0^t\lambda\left(\Phi(s|X_{T_n})\right)\ud s\right),\\
\E\left[\varphi(X_{T_{n+1}})\,|\,\Phi(T_{n+1}-T_n|X_{T_n})\right] &=& \int_{\R^d} \varphi(z)\, Q\left(\ud z|\Phi(T_{n+1}-T_n|X_{T_n})\right).
\end{eqnarray*}
These equations iteratively define the jump times $T_n$ and the post-jump locations $Z_n = X_{T_n}$. Between two consecutive jump times $T_n$ and $T_{n+1}$, the motion is expressed as,
$$\forall\,T_n\leq t<T_{n+1},~X_{t} = \Phi(t-T_n|X_{T_n}).$$
Beyond jump times and post-jump locations, one may also consider the inter-jumping times $S_n=T_n-T_{n-1}$ and the pre-jump locations $Z_n^- = \Phi(S_n|Z_{n-1})$. An illustrative PDMP trajectory is depicted in Figure \ref{fig:pdmp:traj}. If the process evolves in continuous time, all randomness is contained in the discrete-time characteristics $Z_n$, $Z_n^-$, $S_n$, and $T_n$. For example, when the flow is known, it is sufficient to know the pair $(Z_n, S_{n+1})$ (or $(Z_n, Z_{n+1}^-)$), which notably form Markov chains, to reconstruct the entire trajectory.

\begin{figure}[ht]
    \centering
    \begin{tikzpicture}[scale=2]

\draw[black, ultra thick]
  plot[smooth] coordinates {
    (0.5,0.5) (0.9,0.9) (1.1,0.9) (1.3,0.7)(1.5,0.7) (1.7,1)
  };

\coordinate (P) at (1.65,0.95);
\coordinate (Q) at (1.7,1);

\pgfmathsetmacro{\angle}{atan2(0.05, 0.05)}

\begin{scope}[shift={(P)}, rotate=\angle] 
    \draw[black, line width=1.5pt] 
        (4pt,2pt) arc[start angle=90, end angle=270, radius=2pt] ;
\end{scope}

\draw[gray, ultra thick, dashed]
  (1.7,1) -- (2.5,0.8);

\draw[black, ultra thick]
  plot[smooth] coordinates {
    (2.5,0.8) (2.8,0.7) (3.0,0.7)(3.21,0.8) (3.2,1.0) (3.0,1.3)  (3.0,1.5) (3.3,1.6)
  };
  
\coordinate (P1) at (3.25,1.55);
\coordinate (Q1) at (3.3,1.6);

\pgfmathsetmacro{\angle}{atan2(0.05, 0.05)}

\begin{scope}[shift={(P1)}, rotate=\angle] 
    \draw[black, line width=1.5pt] 
        (4pt,2pt) arc[start angle=90, end angle=270, radius=2pt] ;
\end{scope}

\draw[gray, ultra thick, dashed]
  (3.3,1.6) -- (2.0,1.9);

  \draw[black, ultra thick]
  plot[smooth] coordinates {
    (2.0,1.9) (1.6,1.7) (1.4,1.7)(1.3,1.9)
  };
\draw[black, ultra thick, dashed]
  plot[smooth] coordinates {
    (1.3,1.9) (1.1,2.1)(0.6,1.9)
  };

\fill[black] (0.5,0.5) circle (1pt) node[below] {$Z_{n-1}$};
\fill[black] (2.5,0.8) circle (1pt) node[below] {$Z_{n}$};
\fill[black] (2.0,1.9) circle (1pt) node[above] {$Z_{n+1}$};

\node at (1.0,1.05) {$S_n$};
\node at (3.45,0.8) {$S_{n+1}$};
\node at (1.7,1.3) {$Z_n^-$};
\node at (3.65,1.6) {$Z_{n+1}^-$};

\end{tikzpicture}
    \caption{A typical trajectory of a PDMP starting at $Z_{n-1}=X_{T_{n-1}}$, evolving deterministically along the flow $\Phi$ for the (random) duration $S_n$, then jumping from $Z_n^-=\Phi(S_n|Z_{n-1})$ to the (random) position $Z_n = X_{T_n}$ before resuming its deterministic dynamics.}
    \label{fig:pdmp:traj}
\end{figure}

\subsection{A canonical example: the TCP model}
The TCP process, formally known as the one-dimensional PDMP with linear flow $\Phi(t|x) = x+t$, linear jump rate $\lambda(x)=x$, and deterministic fragmentation $Q(\cdot|x)=\delta_{\{\kappa x\}}$, $0<\kappa<1$, is named after its role in modeling the well-known Transmission Control Protocol, a key mechanism for data transmission over the Internet (see \cite{CMP10} and the references therein). This growth–fragmentation model is a typical one-dimensional application of PDMPs, and has been thoroughly studied from a probabilistic perspective \cite{BardetChristenGuillinMalrieuZitt2013,CMP10,DGR02}. For any $0<\kappa<1$, both the continuous-time process $X_t$ and the embedded chain $Z_n$ admit an explicit invariant distribution \cite{DGR02},
\begin{eqnarray}
\invCT(x)&=&\frac{\sqrt{2/\pi}}{\prod\limits_{n=0}^{\infty}(1-\kappa^{2n+1})}\sum\limits_{n=0}^{\infty}\frac{\kappa^{-2n}}{\prod\limits_{k=1}^{n}(1-\kappa^{-2k})}\exp(-\kappa^{-2n}x^2/2), \nonumber\\
\mu(x)&=&\frac{1}{\prod\limits_{n=1}^{\infty}(1-\kappa^{2n})}\sum\limits_{n=1}^{\infty}\frac{\kappa^{-2n}}{\prod\limits_{k=1}^{n-1}(1-\kappa^{-2k})}x\exp(-\kappa^{-2n}x^2/2). \label{eq:mu:tcp}
\end{eqnarray}
Figure~\ref{fig:simu:data} shows a sampled trajectory along with its distribution.

\begin{figure}[ht]
    \centering
    \includegraphics[height=4.4cm]{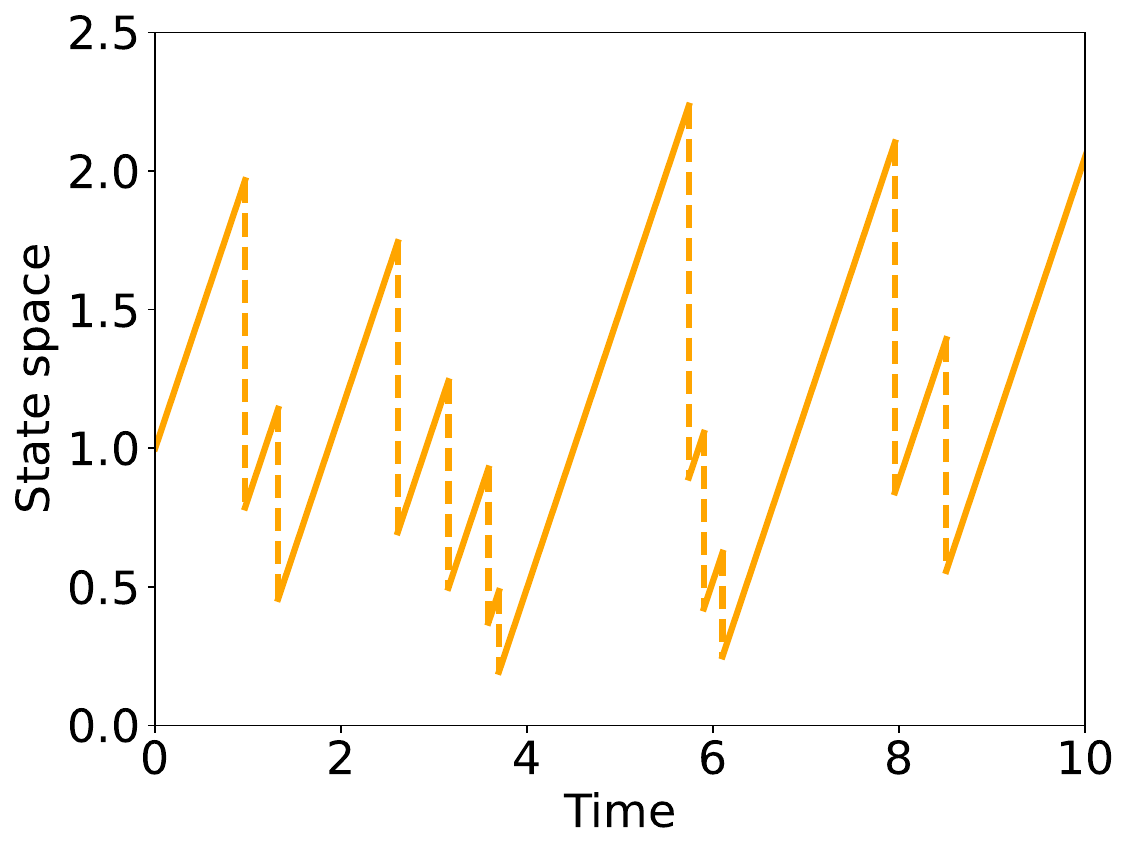}\hspace{0.75cm}
    \includegraphics[height=4.4cm]{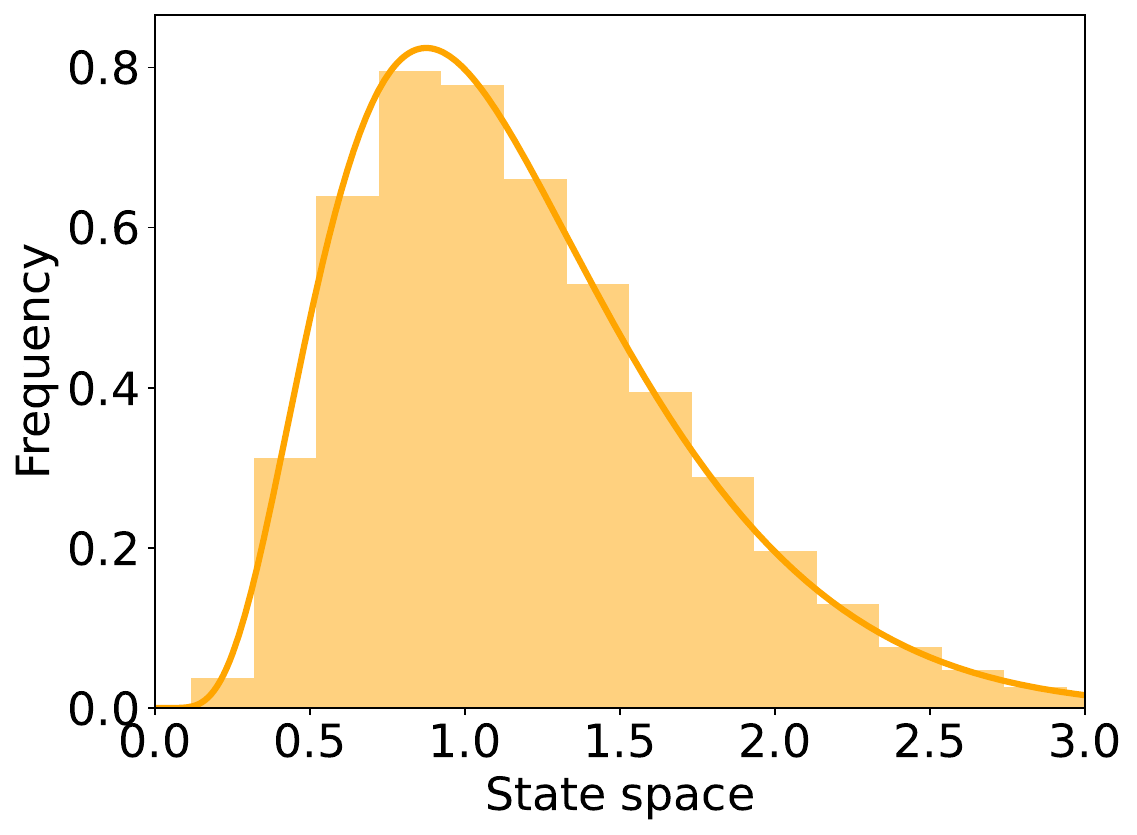}
    \caption{Simulated trajectory from the TCP model with parameter $\kappa=0.4$ until time $10$ (left) and its empirical distribution (in continuous time) evaluated from the first $10\,000$ jumps together with its theoretical version $\invCT$ overlaid (right).}
    \label{fig:simu:data}
\end{figure}

\subsection{Goal: non-parametric estimation of the jump rate}\label{ss:intro:tcp}

Given the wide range of applications for the class of PDMPs (in reliability, biology, insurance, see for instance \cite{CDGMMY17,SDZE12,KP11} and references therein), the question of estimating local characteristics is crucial. In particular, the jump rate is a key quantity of interest in this type of model. Depending on the context, it may govern the onset of congestion (in the TCP model \cite{CMP10}), a cell division (in a growth–fragmentation model \cite{DHKR15,K24,RHKARD14}), a change in the propagation regime (in reliability \cite{doi:10.1177/1748006X16651170,SDZE12}), or even a relapse (in a cancer model \cite{Cleynen2025}). Several statistical frameworks can be considered for this purpose. In this paper, we adopt a particularly general framework, aiming not only to develop statistical methods that can be applied to a large variety of problems, but also to better understand what distinguishes PDMPs. Specifically, we focus on the non-parametric estimation of the jump rate $\lambda$.

Various observation schemes for this type of process can be considered, the most challenging probably being noisy observations on a time grid, which generally do not allow direct measurement of the jumps \cite{doi:10.1177/1748006X16651170,Cleynen2025}. In this article, we instead assume that the embedded chain is perfectly observed over a long time horizon, i.e. the first $n$ jumps of a single trajectory are observed. Assuming that the flow $\Phi$ is known, this scheme, notably adopted in \cite{A14,ADGP14,AM16,F13,K16,KS21}, is equivalent to fully observing the continuous-time trajectory. As shown in recent literature \cite{ADGP14,AM16,F13,K16,KS21}, even under this ideal observation framework, the non-parametric estimation of the jump rate remains a non-trivial problem.

We consider general processes without imposing specific assumptions on the flow, the transition kernel, or the jump rate. However, to ensure convergence properties of the estimators (as $n$ tends to infinity), it is typically required to assume or state an ergodicity condition on the embedded chain $Z_n$. Interestingly, there is an equivalence on stability, recurrence and ergodicity between the continuous-time process $X_t$ and a discrete-time embedded Markov chain, which is fully defined from the local characteristics \cite{CD08}. In this paper, the PDMP of interest $X_t$ is supposed to admit a unique invariant measure $\invCT$. Additionally, $\mu$ represents the invariant distribution of post-jump locations $Z_n$, while $\mu^-$ denotes the invariant distribution of pre-jump locations $Z_n^-$.

To capture the jump rate of interest, one approach is to observe that the jump intensity of the process starting from $\xi$ is in fact the composition $\lambda(\Phi(t|\xi))$, which can be written as the ratio of the conditional density function $f(t|\xi)$ of $S_{n+1}$ given $Z_n=\xi$ over the related conditional survival function $G(t|\xi)$, with
\begin{eqnarray}
f(t|\xi) &=& \lambda(\Phi(t|\xi))\exp\left(-\int_0^t\lambda(\Phi(s|\xi))\ud s\right),\nonumber\\
G(t|\xi) &=& \exp\left(-\int_0^t\lambda(\Phi(s|\xi))\ud s\right).\label{eq:def:condG}
\end{eqnarray}
Multiplying the numerator and the denominator by the invariant measure $\mu$ of $Z_n$, one obtains the formula
\begin{equation}\label{eq:lambdacircphi:intro}
\lambda(\Phi(t|\xi)) = \frac{\mu(\xi) f(t|\xi)}{\mu(\xi)G(t|\xi)},
\end{equation}
where the numerator can be interpreted as the invariant density of $(Z_n,S_{n+1})$ at $(\xi,t)$ and the denominator can be seen as the measure of $(Z_n,\I_{\{S_{n+1}>t\}})$ at $(\xi,1)$ under the invariant distribution. By estimating these two invariant measures using a non-parametric method, one obtains a quotient estimator of $\lambda(\Phi(t|\xi))$ in any dimension. Any pair $(\xi,t)$ such that $\Phi(t|\xi)=x$ then provides an estimator of $\lambda(\Phi(t|\xi)) = \lambda(x)$.

In dimension $1$, if the flow defines an increasing change of variable, it is straightforward to see from \eqref{eq:lambdacircphi:intro} that
\begin{equation}\label{eq:lambda:intro}
\lambda(x) = \frac{\Delta(x) \mu^{-}(x)}{\P_\mu\left(Z_0\leq x<Z_1^-\right)},
\end{equation}
where $\Delta(x) = \partial_t\Phi(0|x)$. This provides another strategy for estimating the jump rate of interest: both the numerator and the denominator are related to the invariant measure of the embedded chain, which makes them amenable to non-parametric estimation.

Consequently, at this stage we have two quotient-type formulas for non-parame\-tri\-cally estimating the jump rate of interest in dimension $1$. In a sense, they differ according to when the flow is used as a change of variable: either before or after the estimation step. The main objective of this paper is to compare these two approaches for one-dimensional PDMPs when the invariant distributions are estimated using non-recursive kernel methods. It should be noted that the aforementioned applications \cite{doi:10.1177/1748006X16651170,CMP10,Cleynen2025,CDGMMY17,SDZE12,DHKR15,KP11,K24,RHKARD14} involve one-dimensional models, which justify the development and study of estimation techniques tailored to this case.

\subsection{A commented state of the art}

We review here, in chronological order, the non-parametric methods developed in the literature and their link with formulas \eqref{eq:lambdacircphi:intro} and \eqref{eq:lambda:intro} above. Some connections and similarities in the construction of the estimators, highlighted in Remarks~\ref{rem:f}, \ref{rem:k} and \ref{rem:ks}, do not appear to have been reported in the existing bibliography.

\paragraph{Fujii (2013)} The method designed in \cite{F13} to estimate the jump rate of a PDMP is based on Rice's formula derived for one-dimensional processes in \cite{BL08}, which highlights the link between local time and stationary distribution. The author assumes that the flow of the PDMP under consideration $X_t$ is of the form,
\begin{equation}
    \label{eq:def:flow:fujii}
\Phi(t|x) = x + \int_0^t \Delta(\Phi(s|x)) \ud s ,
\end{equation}
where $\Delta$ is a real-valued function such that $\inf\Delta>0$. By definition, $\Delta$ coincides with the time-derivative of the flow appearing in \eqref{eq:lambda:intro}. The local time of $X_t$ is then defined as
$$r(x) = \frac{\#\{t\in[0,T]\,:\,X_t = x\}}{\Delta(x)}.$$
After noting that the normalized local time $r(x)/T$ estimates the invariant density (of the continuous-time process) $\invCT(x)$, the author turns to the estimation of the jump rate and proposes to estimate $\lambda(x)$ by $\lambdahat{T}(x) = T A(x)/ r(x)$, with
$$A(x) = \frac{1}{T} \int_0^T \int_{\R_+} K_{h_T}\left(X_{t^-} - x\right)\ud \chi(t,z),$$
where $K_{h_T}$ is a kernel with bandwidth $h_T$ and $\chi$ is the counting measure defined by
$$\chi(t,A) = \#\{n\in\N\,:\,T_n\leq t,\,Z_n\in A\}.$$
Under the assumption that jumps are only additive and downwards, i.e. $Z_n = Z_n^- - J_n$ with $J_n>0$ almost surely, the uniform convergence (when $T$ goes to infinity) in probability of this estimate is notably established. It should be noted that the observation scheme deviates from the one presented earlier in this paper: the observation window is in the time of the process $X_t$ and not in the discrete time of the embedded chain.

\begin{remark}\label{rem:f}
Denoting $n_T$ the (random) number of jumps before $T$ and using the monotonicity of the flow, we rewrite Fujii's estimator in the more conventional form,
\begin{equation}\label{eq:fujii:lambdahat} \lambdahat{T}(x) =\Delta(x)\times\frac{\frac{1}{n_T}\sum_{i=1}^{n_T} K_{h_T}(Z_i^--x)}{\frac{1}{n_T}\sum_{i=0}^{n_T-1}\I_{\{Z_i\leq x<Z_{i+1}^-\}}}.
\end{equation}
Interestingly, the numerator is a kernel estimator of the invariant distribution $\mu^-$ while the denominator is the empirical version of $\P_\mu(Z_0\leq x<Z_1^-)$. Therefore, Fujii's estimator implicitly uses \eqref{eq:lambda:intro} to capture the jump rate of interest.
\end{remark}

\paragraph{Aza\"{\i}s, Dufour and G\'egout-Petit (2014)} The approach developed in \cite{ADGP14} does not rely on \eqref{eq:lambdacircphi:intro} or \eqref{eq:lambda:intro} to capture the jump rate but on the multiplicative intensity model developed by Aalen \cite{A78}. The multiplicative intensity model assumes the observation of a continuous-time counting process $N(t)$ whose stochastic intensity is of the form
$\lambda(t)Y(t)$ for some predictable process $Y(t)$. Equivalently, the process $M(t) = N(t) - \int_0^t\lambda(s)Y(s)\ud s$ is a continuous-time martingale. Within this framework, the Nelson–Aalen estimator provides a consistent estimator of the cumulative jump rate $\Lambda(t) = \int_0^t\lambda(s)\ud s$. Kernel-based methods can then be developed to estimate directly the quantity of interest $\lambda(t)$ \cite{RamlauHansen1983}.

The multiplicative intensity assumption holds in a wide range of applications \cite{ABGK93}, notably for non-homogeneous marked renewal processes whose dynamics resemble those of PDMPs \cite{ADGP13}. However, the jump mechanism of PDMPs differs from that of marked renewal processes, so that the multiplicative intensity model is generally not satisfied. The approach proposed in \cite{ADGP14} circumvents this difficulty by identifying a suitable transformation of the process of interest that satisfies the multiplicative intensity assumption.

The return to the jump rate of interest is also studied, leading to an estimate of the conditional density $f$ of inter-jumping times. 
The main drawback of this technique is that, in order to return to the quantity of interest, it requires integration over the whole space, even for local estimation. Nevertheless, convergence in probability, uniform on any compact set, of the conditional density estimator is demonstrated for processes on very general spaces (more general than $\R^d$) and involving forced jumps at the boundary.

\paragraph{Aza\"{\i}s and Muller-Gueudin (2016)} PDMPs investigated in \cite{AM16} are defined on $\R^d$, may involve forced jumps at the boundary, and possess a transition kernel that is continuous with respect to the Lebesgue measure. The key is the quotient formula \eqref{eq:lambdacircphi:intro} for $\lambda(\Phi(t|\xi))$, where both the numerator and the denominator are directly related to the invariant distribution of the embedded chain. Both are estimated using recursive kernel methods with $h_n^\textrm{\tiny s}$ and $h_n^\textrm{\tiny t}$ denoting respectively space and time bandwidths. Under mainly an ergodicity assumption and a Lipschitz mixing condition, pointwise almost sure convergence and asymptotic normality with expected rate $\sqrt{n(h_n^\textrm{\tiny s})^d h_n^\textrm{\tiny t}}$ of this estimator have been established in \cite{AM16}. As aforementioned, the authors notice that any couple $(\xi,\tau_x(\xi))$ such that $\Phi(\tau_x(\xi)|\xi)=x$ provides a consistent estimator of $\lambda(x)$. Indeed,
$$ \widehat{\lambda\circ\Phi}_n(\tau_x(\xi)|\xi) \toas \lambda(\Phi(\tau_x(\xi)|\xi)) = \lambda(x).$$
It is then proposed to select the one with minimal asymptotic variance, which yields to maximize $\mu(\cdot)G(\tau_x(\cdot)|\cdot)$ on $\calC_x = \{\xi\in\R^d\,:\,\exists\,t\geq0,\,x=\Phi(t|\xi)\}$. Of course, this quantity is unknown but well-estimated by the denominator of $\widehat{\lambda\circ\Phi}_n(\tau_x(\xi)|\xi)$. The estimator obtained with optimal argument selection from the estimated criterion is not theoretically investigated but numerical simulations show its good properties.

\paragraph{Krell (2016)} PDMPs under consideration in \cite{K16} are one-dimensional and with a specific transition kernel of the form $Q(\cdot|x) = \delta_{\{h(x)\}}$, i.e. jumps are deterministic with $Z_n = h(Z_n^-)$. The statistical approach relies on the assumption that both the flow $\Phi(\cdot|x)$ (which will also be denoted $\Phi_x$ to shorten some equations) and the fragmentation function $h$ define a change of variable on $\R_+$, which makes this strategy specific to dimension $1$. In that context, it is shown that the targeted jump rate can be written as $\lambda(x) = \mu(h(x)) / \Dk{x}$, with
\begin{equation}\label{eq:def:Dk}
\Dk{x} = \E_{\mu}\left[g_{Z_0}(h(x)) \I_{\{h(Z_0)\leq h(x)\}} \I_{\{Z_1\geq h(x)\}}\right],
\end{equation}
where
\begin{equation*}\label{eq:def:g}
g_{x}(y)=\displaystyle\frac{1}{(h\circ \Phi_x)'((h\circ \Phi_x)^{-1}(y))} .
\end{equation*}
The numerator is evaluated as a kernel estimator of the invariant distribution $\mu$ composed with $h$, while the denominator is estimated by its empirical version. The main result of \cite{K16} states the convergence in $\LL^1$-norm to $0$ of the square error with a rate depending on the regularity of $\lambda$. In contrast with previous approaches \cite{ADGP14,AM16,F13}, this paper does not assume the ergodicity of the process; rather, it establishes it as an intermediate outcome derived from regularity assumptions on its local characteristics. It should be noted that this estimator requires the knowledge of the transition kernel through the fragmentation function $h$. However, this part does not involve any randomness, which makes it easy to be estimated.

\begin{remark}\label{rem:k}
In fact, $g_x(y)$ does not depend on x. Indeed, using the flow property, a direct calculus shows that
\begin{equation*}
g_x(y) = \frac{(h^{-1})'(y)}{\Delta(h^{-1}(y))},
\end{equation*}
where $\Delta(x) = \Phi_x'(0)$ appears in \eqref{eq:lambda:intro} and \eqref{eq:def:flow:fujii}. This allows us to simplify the formula \eqref{eq:def:Dk} of $\Dk{x}$,
\begin{eqnarray*}
    \Dk{x} &=& \frac{\P_\mu\left(Z_0\leq x,Z_1\geq h(x)\right)}{h'(x)\Delta(x)}, \nonumber \\
    &=&\frac{\P_\mu\left(Z_0\leq x<Z_1^-\right)}{h'(x)\Delta(x)}, \label{eq:Dk:2}
\end{eqnarray*}
showing
\begin{equation}\label{eq:lambda:intro:k}
\lambda(x) = h'(x)\Delta(x)\frac{\mu(h(x))}{\P_\mu\left(Z_0\leq x<Z_1^-\right)}.\end{equation}
Together with $\mu^-(x) = \mu(h(x))h'(x)$, this makes it possible to point out that Krell's approach is actually based on \eqref{eq:lambda:intro} to capture the jump rate of interest, but tackles the numerator via $\mu$ (and not $\mu^-$) using the knowledge on the deterministic transition kernel $Q$.
\end{remark}

\paragraph{Krell and Schmisser (2021)} The last paper \cite{KS21} of this review generalizes the approach of \cite{K16} still to one-dimensional processes but with general transition kernel (in particular not assumed to be continuous with respect to the Lebesgue measure as in \cite{AM16}). Assuming again that the flow defines a change of variable on $\R_+$, it is shown that
$\lambda(x) = \mu^-(x)/\Dks{x}$, with
\begin{equation*}\label{eq:def:Dks}
\Dks{x} = \E_\mu\left[ (\Phi_{Z_0}^{-1})'(x)\I_{\{Z_0\leq x<Z_1^-\}}\right].
\end{equation*}
As in \cite{K16}, the denominator is estimated by its empirical version, while the numerator is evaluated by an adaptive projection estimator of $\mu^-$. The main result of \cite{KS21} states that, under regularity assumptions, the estimator under consideration is nearly minimax (up to a $\log^2n$ factor) in $\LL^2$-norm. In contrast with \cite{ADGP14,AM16,F13} and as in \cite{K16}, the ergodicity of the embedded chain is not assumed but is instead established on the basis of the properties of the process’s local characteristics.

\begin{remark}\label{rem:ks}
Reusing the rationale of Remark~\ref{rem:k}, it is easy to see that
$$\Dks{x} = \frac{\P_\mu\left(Z_0\leq x<Z_1^-\right)}{\Delta(x)} ,$$
proving again \eqref{eq:lambda:intro}. This highlights the fact that, like Fujii’s approach, the procedure of Krell and Schmisser estimates the jump rate of interest via the quotient \eqref{eq:lambda:intro}. Besides the observation window (in the time of the PDMP or in the time of the embedded chain), the main difference in terms of estimator construction between the two approaches lies in the technique used to handle the numerator: kernel estimator in Fujii's paper \cite{F13} vs. adaptive projection estimator in Krell and Schmisser's article \cite{KS21}. It is important to point out, however, that the theoretical results demonstrated in these papers are different and complementary: uniform convergence in probability in \cite{F13} vs. minimax rate in $\LL^2$-norm in \cite{KS21}.
\end{remark}

\subsection{Kernel estimators of the jump rate}

The first conclusion to be drawn from this review of the literature is that the results obtained, albeit within the same framework, i.e. general PDMPs observed over a long period of time, are difficult to compare, notably:
\begin{itemize}[wide=4.5pt,labelsep=4.5pt]
\item \cite{F13} focuses on the asymptotics in the (continuous) true time of the PDMP, while the other publications deal with the number of jumps observed.
\item \cite{ADGP14} does not estimate the jump rate but the associated conditional density.
\item \cite{AM16} implements a recursive kernel method, \cite{ADGP14,F13,K16} assume a non-recursive kernel estimator, while \cite{KS21} uses an adaptive projection estimator.
\item \cite{F13,K16,KS21} deal with one-dimensional processes, while the state space at hand in \cite{AM16} is $\R^d$ and \cite{ADGP14} considers general metric spaces.
\item Transitions at jump times are deterministic in \cite{K16}, continuous with respect to the Lebesgue measure in \cite{AM16}, and of a general type in \cite{ADGP14,F13,KS21}.
\item The convergence results obtained do not involve the same measures of estimation error.
\item Several strategies are employed to capture the jump rate of interest: \cite{F13,KS21} rely on the representation \eqref{eq:lambda:intro}, \cite{K16} uses \eqref{eq:lambda:intro:k}, while \cite{AM16} is based on \eqref{eq:lambdacircphi:intro}.
\end{itemize}

The main goal of this paper is to make the first rigorous, both theoretical and numerical, comparison of the strategies used to capture the jump rate of interest. To achieve this, we assume a unique statistical framework and apply the same estimation method to the different strategies. That is why, in the present paper, we only consider one-dimensional processes observed via their embedded Markov chain, which entails asymptotic analyses in the number of jumps observed. In addition, in order to work with estimators built on a single method, we restrict ourselves to non-recursive kernel estimators.

Here and throughout the paper, $K^d$ denotes a kernel function in dimension $d$, i.e. a density on $\R^d$. To shorten equations, $K_h^d = K^d(\cdot/h)/h^d$ for any $h>0$. In addition, we omit the exponent when $d=1$, i.e. $K^1 = K$. The estimators selected for the rest of the study are as follows:
\begin{itemize}
    \item $\lambdak{n}(x)$ is the estimator based on \eqref{eq:lambda:intro:k},
\begin{equation}\label{eq:def:k}
\lambdak{n}(x) = h'(x)\Delta(x)\times\frac{\frac{1}{n}\sum_{i=0}^{n-1} K_{h_n}(Z_i-h(x))}{\frac{1}{n}\sum_{i=0}^{n-1}\I_{\{Z_i\leq x\}}\I_{\{Z_{i+1}\geq h(x)\}}},
\end{equation}
where $h$ is the (deterministic) fragmentation function (assumed to be known), i.e. $Q(\cdot|x) = \delta_{\{h(x)\}}$.

    \item $\lambdaks{n}(x)$ is the estimator based on \eqref{eq:lambda:intro},
    \begin{equation}\label{eq:def:ks}
\lambdaks{n}(x) = \Delta(x)\times\frac{\frac{1}{n}\sum_{i=1}^{n} K_{h_n}(Z_i^--x)}{\frac{1}{n}\sum_{i=0}^{n-1}\I_{\{Z_i\leq x<Z_{i+1}^-\} }}.
\end{equation}
It is a generalization of \eqref{eq:def:k} when the transition kernel is unknown and not assumed to be deterministic.

    \item $\widehat{\lambda\circ\Phi}_n(t|\xi)$ is the estimator of $\lambda(\Phi(t|\xi))$ based on \eqref{eq:lambdacircphi:intro},
    \begin{equation*}
\widehat{\lambda\circ\Phi}_n(t|\xi) = \frac{\frac{1}{n}\sum_{i=0}^{n-1} K_{h_n^\textrm{\tiny s}}(Z_i-\xi) K_{h_n^\textrm{\tiny t}}(S_{i+1}-t)}{\frac{1}{n}\sum_{i=0}^{n-1} K_{h_n^\textrm{\tiny s}}(Z_i-\xi)\I_{\{S_{i+1}>t\}}} .
\end{equation*}
Following the optimal argument selection approach developed in \cite{AM16}, this yields the (oracle) estimator of the jump rate,
\begin{equation}\label{eq:def:amgo} \lambdaamgo{n}(x) = \widehat{\lambda\circ\Phi}_n(\tau_x(\xi)|\xi) \quad\text{with}\quad \xi = \argmax_{\calC_x} \mu(\cdot)G(\tau_x(\cdot)|\cdot),
\end{equation}
where $\tau_x(\xi)$ is the unique solution of $x=\Phi(\cdot|\xi)$ and $G$ is the conditional survival function of inter-jumping times given in \eqref{eq:def:condG}. $\lambdaamgo{n}(x)$ can not be evaluated in real world application scenarios since both $\mu$ and $G$ are unknown. However, with an estimated argument selection, one gets
\begin{equation}\label{eq:def:amg} \lambdaamg{n}(x) = \widehat{\lambda\circ\Phi}_n(\tau_x(\xi)|\xi) \quad\text{with}\quad \xi = \argmax_{\calC_x}\sum_{i=0}^{n-1} K_{h_n^\textrm{\tiny s}}(Z_i-\cdot)\I_{\{S_{i+1}>\tau_x(\cdot)\}}.
\end{equation}
\end{itemize}

\subsection{Contribution}

The main goal of this paper is to compare the estimators mentioned above, both theoretically and numerically. In particular, we aim to evaluate their rates of convergence and, if the rates are equal, their asymptotic variances. For this purpose, we shall rely on the following limit result.

\begin{theorem}{\normalfont\cite[Corollary 3.10]{AM16}}\label{theo:lambdaamg}
Under mainly an ergodicity assumption and a Lipschitz mixing condition, when $n$ goes to infinity,
$$\sqrt{nh_n^\textrm{\normalfont\tiny s}h_n^\textrm{\normalfont\tiny t}}\left[\lambdaamgo{n}(x)-\lambda(x)\right] \todist\calN(0,\sigmaamg{x}),$$
for bandwidth sequences of the form $h_n^\textrm{\normalfont\tiny s}\propto n^{-\alpha}$ and $h_n^\textrm{\normalfont\tiny t}\propto n^{-\beta}$ (with $\alpha,\beta>0$ and $1-2\min(\alpha,\beta)<\alpha+\beta<1$), with $\tau^2 = \int_\R K^2(x)\ud x$ and
\begin{equation*}
\sigmaamg{x}=\frac{\tau^4\lambda(x)}{\max_{\xi\in\calC_x}\mu(\xi)G(\tau_x(\xi)|\xi)},
\end{equation*}
where $\tau_x(\xi)$ is the unique solution of $x=\Phi(\cdot|\xi)$ and $G$ is the conditional survival function \eqref{eq:def:condG} associated with jump intensity $\lambda\circ\Phi$.
\end{theorem}
\begin{remark}
Theorem~\ref{theo:lambdaamg} above is stated for the non-recursive estimator \eqref{eq:def:amgo}, wher\-eas the reference result \cite[Corollary 3.10]{AM16} is established for its recursive version. However, a careful review of the proof in the reference paper reveals that the two asymptotic variances are identical up to a factor $\alpha+\beta+1$ (in dimension 1).
\end{remark}

\begin{remark}\label{rem:lambdaamgo:vs:lambdaamg}
According to \cite[Corollary~3.10]{AM16}, the estimator $\widehat{\lambda\circ\Phi}_n(t|\xi)$ is consistent for any fixed argument. However, in the definition \eqref{eq:def:amg} of $\lambdaamg{n}(x)$ (in contrast to $\lambdaamgo{n}(x)$), the argument $(\xi,t)$ of $\widehat{\lambda\circ\Phi}_n(t|\xi)$ is random, since it is defined as the maximizer of a random function. This substantially complicates the theoretical analysis of the estimator, which is not developed in \cite{AM16}. Nevertheless, the random function being maximized converges almost surely, at rate $\sqrt{n h_n^\textrm{\normalfont\tiny s}}$, to the deterministic function maximized in the definition \eqref{eq:def:amgo} of $\lambdaamgo{n}(x)$ \cite[Theorem~3.3]{AM16}. It is therefore expected that the estimators $\lambdaamgo{n}(x)$ and $\lambdaamg{n}(x)$ exhibit similar behavior.
\end{remark}

We aim to establish a result analogous to the above theorem for $\lambdak{n}$ and $\lambdaks{n}$. This will make it possible to compare both convergence rates and limit variances, and thus establish the first quantitative comparison of the estimators chosen in this study. Better still, we may expect this comparison to remain somehow valid for evaluating the strategies used to capture the jump rate of interest, whatever the estimation technique employed. The theoretical study is presented in Section~\ref{s:clt}, the model assumptions are stated in Subsection~\ref{ss:modelass}, the results are provided in Subsection~\ref{ss:res} and their proof is outlined in Subsection~\ref{ss:sketch}. In particular, the consistency of the estimators is established in Theorem~\ref{main_ps}, while the asymptotic normality is demonstrated in Theorem~\ref{main_tcl}. The detailed proofs are deferred to Appendix~\ref{app:proof}.

The theoretical results that we will obtain with Theorems~\ref{main_ps} and \ref{main_tcl} indicate that $\lambdak{n}$ and $\lambdaks{n}$ can not generally be ordered, at least based on the criterion of asymptotic variance, even within the same model (see Remark~\ref{rem:comp:var:k:ks}). To delve deeper and include $\lambdaamgo{n}$ and $\lambdaamg{n}$ in this comparison, Section~\ref{s:tcp} focuses on the TCP model, for which the stationary distribution is explicit (see Subsection~\ref{ss:tcp:tcp}). After carefully selecting the smoothing parameters in Subsection~\ref{ss:tcp:bandwidth}, we establish a rigorous comparison of the variances of the estimators under consideration, taking convergence rates into account, in Subsection~\ref{ss:tcp:var}. Furthermore, we demonstrate through intensive numerical simulations that these variances allow for a qualitative prediction of the ordering of pointwise estimation errors (at least for $\lambdaks{n}$ and $\lambdaamg{n}$), making the link between theory and practice. In Subsection~\ref{ss:adaptive:simus}, we introduce adaptive projection–based variants of the estimators $\lambdak{n}$ and $\lambdaks{n}$. Numerical results show that the pointwise error of these estimators follows the trend identified both theoretically and numerically for kernel estimators, with comparable orders of magnitude. This finding appears to confirm the importance of carefully choosing the strategy used to estimate the jump rate, regardless of the method employed to estimate the underlying invariant distribution.

Applying estimators to real-world data allows us to test them outside the idealized framework of numerical simulations. This is the goal of Section~\ref{s:data}, where we focus on a model of cell growth and division, a typical application of one-dimensional PDMPs in biology \cite{CDGMMY17,DHKR15,K24}. Specifically, we analyze single-cell data from \textit{Escherichia coli} \cite{TPPHBY17} through the lens of piecewise-deterministic models. In this context, the jump rate governs the timing of cell division based on cell size. Our objective is to estimate this quantity using the methods described earlier. The biological context is detailed in Subsection~\ref{ss:data:context}, while the model fitting is carried out in Subsection~\ref{ss:data:37}. The computation of the estimators (including an adaptive projection-based version of $\lambdaks{n}$) is presented in Subsection~\ref{ss:data:jr}, with an a posteriori validation developed in Subsection~\ref{ss:data:valid}. This approach allows us to complement the theoretical and numerical investigations with an evaluation of the estimators in a real-world application.

The conclusion of our study, encompassing the theoretical, numerical, and real-world data application aspects of the estimators under consideration, is presented in Section~\ref{s:conclusion}.

\section{Consistency and asymptotic normality}
\label{s:clt}

\subsection{Model assumptions}
\label{ss:modelass}

The PDMP $X_t$ under consideration in this section is defined on $\R_+$. The set of sufficient conditions to establish the (pointwise) almost sure convergence and asympotic normality of the estimators is given below, starting with the main assumptions (on the general dynamics of the process in Assumptions~\ref{ass:local} and on its asymptotic behavior in Assumption~\ref{ass:ergodicity}).

\begin{assumptions}\label{ass:local}The local parameters $\Phi$ and $Q$ satisfy the following conditions.
\begin{itemize}
    \item[\!\circledletter{colK}{colK}{$\clubsuit$}\!\!\circledletter{colKS}{colKS}{$\vardiamondsuit$}] For any $x\in\R_+$, the flow $\Phi_x:\R_+\to[x,\infty)$ is of class $C^1$ and strictly increasing.
    \item[\!\circledletter{colK}{colK}{$\clubsuit$}]The transition kernel writes $Q(\cdot|x) = \delta_{\{h(x)\}}$, where $h:\R_+\to\R_+$ is of class $C^1$, strictly increasing, and sublinear ($h(x)\leq x$).
    \item[\!\circledletter{colKS}{colKS}{$\vardiamondsuit$}]The transition kernel $Q(\cdot|x)$ admits a density with respect to the Lebesgue measure and is such that $Q([0,x)|x)=1$.
\end{itemize}
\end{assumptions}

\begin{assumption}\label{ass:ergodicity}
For any initial distribution $\mu_0=\delta_{\{x\}}$, $x\in\R_+$, when $n$ goes to infinity,
$$\tvnorm{\mu_n-\mu}\to0,$$
where $\mu_n$ is the distribution of $Z_n$ and $\tvnorm{\cdot}$ stands for the total variation norm.
\end{assumption}

This ergodicity assumption can be directly linked to the transition kernel of the Markov chain $Z_n$, for example, through the existence of a Foster-Lyapunov function or Doeblin’s condition. For further details on such connections, we refer the interested reader to \cite{MT09} and references therein. Following the arguments provided in \cite[3.2~Ergodicity]{ADGP14}, this guarantees that both $Z_n$ and $Z_n^-$ are irreducible, positive Harris-recurrent \cite[9.1.2~Harris recurrent chains]{MT09}, and aperiodic, each with a unique invariant probability measure ($\mu$ for $Z_n$ and $\mu^-$ for $Z_n^-$).
Consequently, the almost sure ergodic theorem \cite[Theorem~17.1.7]{MT09} and the central limit theorem for Markov chains \cite[Theorem~17.5.3]{MT09} apply to $Z_n$ and $Z_n^-$.

The next (technical) conditions operate on the conditional distribution $R$ of $Z_{n+1}^-$ given $Z_n$ and the transition kernel $P$ of the Markov chain $Z_n$. Specifically, with a change of variable, straightforward calculations show that $R$ admits a density with respect to the Lebesgue measure, expressed as
$$
R(z|x) = \frac{\lambda(z)}{\Delta(z)}\exp\left(-\int_x^z\frac{\lambda(u)}{\Delta(u)}\ud u\right)\I_{[x,\infty)}(z).
$$
If $Q$ possesses a density, then $P$ is given by
$$P(y|x) = \int_x^\infty Q(y|z)R(z|x)\ud z.$$
In the case of deterministic transitions where $Q(\mathrm{d}y|z) = \delta_{\{h(z)\}}(\mathrm{d}y)$, the transition kernel simplifies to
$$P(y|x) = (h^{-1})'(y) R(h^{-1}(y)|x).$$
Interestingly, $P$ admits a density in both scenarios. It is worth noting that this implies that $\mu$ and $\mu^-$ possess a density regardless of the transition type considered.

\begin{assumptions}\label{ass:technical}The regularity class $\mathfrak{C}$ of conditional densities $\rho$ on $\R_+$ is defined by the following conditions.
\begin{itemize}[wide=4.5pt,labelsep=4.5pt]
    \item $\rho$ is bounded.
    \item $\rho$ is uniformly Lipschitz in the first variable, i.e. there exists $\lip{\rho}>0$ such that,
    $$\forall\,x,y,z\in \R_+,\ |\rho(y|x)-\rho(z|x)|\leq \lip{\rho}|y-z|. $$
\end{itemize}
\begin{itemize}
    \item[\!\circledletter{colK}{colK}{$\clubsuit$}] The transition density $P$ of the Markov chain $Z_n$ belongs to $\mathfrak{C}$.
    \item[\!\circledletter{colKS}{colKS}{$\vardiamondsuit$}] The conditional density $R$ of $Z_{n+1}^-$ given $Z_n$, for any $n\in\N$, belongs to $\mathfrak{C}$.
\end{itemize}
\end{assumptions}

\begin{assumptions}\label{ass:clt}
The transition density $P$ of $Z_n$ and the conditional density $R$ of $Z_{n+1}^-$ given $Z_n$, for any $n\in\N$, satisfy the following conditions.
\begin{itemize}
    \item[\!\circledletter{colK}{colK}{$\clubsuit$}\!\!\circledletter{colKS}{colKS}{$\vardiamondsuit$}] There exist $a_1\geq 1$ and $a_2<1$ such that,
    \begin{equation}\label{eq:condition:lipschitzmixing}
    	\forall\,x,y\in\R_+,~\int_{\R_+\times \R_+}|u-v|^{a_1}P(u|x)P(v|y)\mathrm{d}u\,\mathrm{d}v\leq a_2|x-y|^{a_1}.
    \end{equation}
\end{itemize}    
The regularity class $\emph{\normalfont Li}(r_1,r_2)$ \cite[6.3.2~Lipschitz mixing]{D97} is the family of transition densities $A$ such that
\begin{align*}
&|A(y_1|x_1) - A(y_2|x_2)| \\
&= O\left(\|(x_1,y_1) - (x_2,y_2)\|^{r_2}\left( \|(x_1,y_1)\|^{r_1} + \|(x_2,y_2)\|^{r_2} + 1 \right)\right).
\end{align*}
\vspace{-0.75cm}
\begin{itemize}
    \item[\!\circledletter{colK}{colK}{$\clubsuit$}] $P$ belongs to the regularity class $\emph{\normalfont Li}(r_1,r_2)$ for some positive numbers $r_1$ and $r_2$, satisfying $2(r_1+r_2)\leq a_1$.
    \item[\!\circledletter{colKS}{colKS}{$\vardiamondsuit$}] $R$ belongs to the regularity class $\emph{\normalfont Li}(r_1,r_2)$ for some positive numbers $r_1$ and $r_2$, satisfying $2(r_1+r_2)\leq a_1$.
    \end{itemize}
\end{assumptions}

The two points of Assumptions~\ref{ass:clt} are of very different kinds. The first one is a so-called Lipschitz mixing hypothesis on the Markov chain $Z_n$. Iterating equation \eqref{eq:condition:lipschitzmixing}, for two independent chains $Z_n$ and $\widetilde{Z}_n$ (with same transition kernel $P$) starting from $z$ and $\tilde{z}$, it yields
$$\| Z_n-\widetilde{Z}_n\|_{\mathbb{L}^{a_1}}\leq a_2^{n/a_1}\lvert z-\tilde{z}\rvert.$$ This provides a geometric control in $\mathbb{L}^{a_1}$-norm of the dependence on the initial state. This stability property is notably required in \cite[Theorem~6.3.17]{D97} to control the rate of convergence of means along trajectories, as an intermediate step to establish central limit theorems (see also \cite[7.6~Commentary]{MT09} for another insight on this approach). In our case, due to the spatial smoothing in estimators \eqref{eq:def:k} and \eqref{eq:def:ks}, the proof of the asymptotic normality involves a remainder term whose rate of convergence is controlled thanks to this Lipschitz mixing condition. 

In comparison, in \cite[Lemma 3]{KS21}, the authors derive a geometric $\beta$-mixing property in order to apply Talagrand's inequalities for $\beta$-mixing variables. In the stationary case, geometric $\beta$-mixing translates to $$\int_{\mathbb{R}_+}\tvnorm{ P^n(\cdot | x)-\mu(\cdot)}\mu(\ud x)=o(\alpha^n),$$
for some $\alpha<1$ (see \cite{MS21} and references therein). It ensures a control of the convergence to the invariant measure with geometric rate, while the above Lipschitz mixing condition allows a control of the distance in $\mathbb{L}^{a_1}$-norm of two independent chains, also with geometric rate. 

The second point of Assumptions~\ref{ass:clt} is a regularity condition used to control the aforementioned remainder term, which is the mean along the chain of some functional $\varphi$. In addition to the stability of $Z_n$ imposed by the Lipschitz mixing hypothesis, $\varphi$ needs to be smooth enough to apply the approach of \cite[Theorem~6.3.17]{D97}. This function $\varphi$ appears to be the transition kernel of $Z_n$ (or $Z_n^-$), but the assumption is not directly used to govern its asymptotic behavior. The regularity conditions imposed in Assumptions~\ref{ass:technical} for $P$ or $R$, depending on the estimator under consideration, are, by contrast, rather mild.

It should be noted that Assumptions~\ref{ass:technical} and \ref{ass:clt} can be difficult to verify in practice because they concern the transition kernels of $Z_n$ and $Z_n^-$, which are non-trivial functions of the local characteristics of the PDMP.

\subsection{Main results}
\label{ss:res}

Before stating our convergence theorems, we group together below the regularity conditions imposed on the kernel function $K$.

\begin{assumptions}\label{ass:kernel}
The kernel function $K$ is a bounded and Lipschitz density with compact support. In addition, we denote $\tau^2 = \int_\R K^2(x)\,\ud x$.
\end{assumptions}

The following result establishes the strong consistency of the two estimators under consideration.

\begin{theorem}\label{main_ps}
If the model satisfies Assumptions~\ref{ass:local}, \ref{ass:ergodicity} and \ref{ass:technical}, and if the kernel $K$ satisfies Assumptions~\ref{ass:kernel}, then, for any bandwidth $h_n\propto n^{-\gamma}$, $0<\gamma<1$, when $n$ goes to infinity,
\begin{itemize}
    \item[\!\circledletter{colK}{colK}{$\clubsuit$}] $\displaystyle\lambdak{n}(x) \toas \lambda(x)$,
    \item[\!\circledletter{colKS}{colKS}{$\vardiamondsuit$}] $\displaystyle\lambdaks{n}(x)\toas \lambda(x)$,
\end{itemize}
for any $x$ such that $\P_\mu(Z_0\leq x<Z_1^-)>0$.
\end{theorem}
\begin{proof} An outline of the proof of these two convergence results is provided in Subsection~\ref{ss:sketch}, with the detailed derivations collected in Appendix~\ref{app:proof}.\end{proof}
Asymptotic normality is stated in the theorem below.

 \begin{theorem}\label{main_tcl}
 If the model satisfies Assumptions~\ref{ass:local}, \ref{ass:ergodicity}, \ref{ass:technical} and \ref{ass:clt}, and if the kernel $K$ satisfies Assumptions~\ref{ass:kernel}, then, for any bandwidth $h_n\propto n^{-\gamma}$, $1/3<\gamma<1$, when $n$ goes to infinity,
 \begin{itemize}
 \item[\!\circledletter{colK}{colK}{$\clubsuit$}] $\displaystyle\sqrt{nh_n}\left[\lambdak{n}(x)-\lambda(x)\right] \todist \calN(0,\sigmak{x})$,
 \item[\!\circledletter{colKS}{colKS}{$\vardiamondsuit$}]  $\displaystyle\sqrt{nh_n}\left[\lambdaks{n}(x)-\lambda(x)\right] \todist \calN(0,\sigmaks{x})$,
 \end{itemize}
with
$$
 \sigmak{x}=\frac{\tau^2\lambda(x)^2 h'(x)}{\mu^-(x)}\qquad\text{and}\qquad
 \sigmaks{x}=\frac{\tau^2\lambda(x)^2}{\mu^-(x)},
$$
 for any $x$ such that $\mu^-(x)>0$ and $\P_\mu(Z_0\leq x<Z_1^-)>0$.
\end{theorem}
\begin{proof} An outline of the proof of these two convergence results is provided in Subsection~\ref{ss:sketch}, with the detailed derivations collected in Appendix~\ref{app:proof}.\end{proof}

The following remarks draw initial conclusions on the theoretical comparison of the estimators under consideration. These will be further developed and supported by numerical experiments in the next section through the example of TCP process.

\begin{remark}\label{rem:comp:var:k:ks}
In light of Theorem~\ref{main_tcl}, estimators $\lambdak{n}(x)$ and $\lambdaks{n}(x)$ of $\lambda(x)$ converge at the same rate $\sqrt{nh_n}$. Furthermore, in the case of deterministic fragmentation $Q(\cdot|x)=\delta_{\{h(x)\}}$, the asymptotic variances are identical up to the multiplicative factor $h'(x)$. In the particular case of linear fragmentation, $h(x)=\kappa x$ with $0<\kappa<1$, estimator $\lambdak{n}$ always has a better variance than $\lambdaks{n}$. However, variances can not generally be ordered: for instance, if $h(x)=x/(1+\exp(-x))$, which satisfies $0\leq h(x)\leq x$, we observe that $h'(1)<1$ and $h'(2)>1$. This occurs even though $\lambdak{n}$ uses additional information (the form of the transition kernel) compared to $\lambdaks{n}$.
\end{remark}

\begin{remark}\label{rem:comp:var:k:ks:amg}
As stated in Theorem~\ref{theo:lambdaamg}, estimator $\lambdaamgo{n}(x)$ converges at the rate $\sqrt{nh_n^\textrm{\normalfont\tiny s}h_n^\textrm{\normalfont\tiny t}}$, whereas we have just shown in Theorem~\ref{main_tcl} that estimators $\lambdak{n}(x)$ and $\lambdaks{n}(x)$ converge at the rate $\sqrt{nh_n}$. This difference is due to the presence of temporal smoothing in the former. Consequently, if we assume that the spatial windows of the three estimators should be equivalent, then estimator $\lambdaamgo{n}(x)$ is slower and, asymptotically, disqualified. Conversely, if we assume $h_n^\textrm{\normalfont\tiny s}h_n^\textrm{\normalfont\tiny t}\sim h_n$, the comparison of asymptotic variances $\sigmaamg{x}$, $\sigmak{x}$ and $\sigmaks{x}$ becomes meaningful and determines which estimator is asymptotically superior (at least based on the variance). In practice, we generally work with a large but finite dataset. The variances of the estimators are therefore on the order of the asymptotic variance in the central limit theorem normalized by the convergence rate. For not-too-large values of $n$, it is possible that a slower estimator achieves a lower variance than a faster one (because the asymptotic variance in the central limit theorem for the former is significantly smaller than that of the latter). The theoretical comparison of variances (normalized by the convergence rate or not), is inherently complex given the generality of our problem. Therefore, this analysis will be undertaken using the example of the TCP model in Section~\ref{s:tcp}.
\end{remark}

\subsection{Sketch of the proof}\label{ss:sketch}

We give here an overview of the proof of the results announced in Subsection~\ref{ss:res}. The full details are provided in Appendix~\ref{app:proof}.

The study of $\lambdak{n}$ is based on formula $\lambda(x) = \mu(h(x))h'(x)/D(x)$ stated in \eqref{eq:lambda:intro:k}, while the investigation on $\lambdaks{n}$ relies on $\lambda(x) = \mu^-(x)/D(x)$ given in \eqref{eq:lambda:intro}, where
\begin{equation}\label{eq:def:D:proof}
D(x) = \frac{\P_\mu\left(Z_{k-1} \le x \le Z_k^{-}\right)}{\Delta(x)}.
\end{equation}
In the first formula, on which the definition of $\lambdak{n}(x)$ is based, the numerator is estimated as follows,
    \begin{equation*}
    \widehat{\mu}_n(h(x)) = \frac{1}{n}\sum_{i=1}^{n} K_{h_n}(Z_i-h(x)),
    \end{equation*}
and in the second formula, yielding to the definition of $\lambdaks{n}(x)$, the numerator is handled analogously,
    \begin{equation*}
    \widehat{\mu}^-_n(x) = \frac{1}{n}\sum_{i=1}^{n} K_{h_n}(Z_i^--x),
    \end{equation*}
where $h_n=cn^{-\gamma}$, $1/3<\gamma<1$ (for the sake of readability, $c=1$ in the sequel of the proof).
The quotient $D(x)$ is finally estimated by 
$$\widehat{D}_n(x) =\frac{1}{n \Delta(x)}\sum_{i=1}^{n}\I_{\{Z_{k-1} \le x \le Z_k^{-}\}}.$$
Estimator $\lambdak{n}(x)$ given in \eqref{eq:def:k} is basically the quotient of $\widehat{\mu}_n(h(x))$ over $\widehat{D}_n(x)/h'(x)$ while $\lambdaks{n}(x)$ given in \eqref{eq:def:ks} is the quotient of $\widehat{\mu}^-_n(x)$ over $\widehat{D}_n(x)$. In order to establish the asymptotic properties of $\lambdak{n}(x)$ and $\lambdaks{n}(x)$, we shall study the two-dimensional vectors $(\widehat{\mu}_n(h(x)),\widehat{D}_n(x)/h'(x))^\top$ and $(\widehat{\mu}^-_n(x),\widehat{D}_n(x))^\top$. Our main goal is to show that they tend towards their deterministic counterparts, $(\mu(h(x)),D(x)/h'(x))^\top$ and $(\mu^-(x),D(x))^\top$, and catch the rates of convergence.

To this end, we introduce the following pieces of notation. First, $F_n$ denotes the $\sigma$-algebra generated by $Z_0,\dots,Z_n$, while $G_n$ is the $\sigma$-algebra generated by
$Z_0,Z_1^-,Z_1,\dots,Z_n^-,Z_n$. We also introduce $M_n = (M_n^{(1)},M_n^{(2)})^\top$  and $\widetilde{M}_n = (\widetilde{M}_n^{(1)},\widetilde{M}_n^{(2)})^\top$, the vector martingales (adapted to $F_n$ and $G_n$, respectively) defined, for $i\in\{1,2\}$, by
\begin{eqnarray}
M_n^{(i)}&=&\sum_{k=1}^n\left(A_k^{(i)}-B_k^{(i)}\right),\label{eq:def:Mn}\\  \widetilde{M}_n^{(i)}&=&\sum_{k=1}^n\left(\widetilde{A}_k^{(i)}-\widetilde{B}_k^{(i)}\right),\nonumber
\end{eqnarray}
where
\[
\begin{aligned}
A_k^{(1)} \;&=\; K_{h_n}\!\left(Z_k - h(x)\right),
&\qquad 
\widetilde{A}_k^{(1)} \;&=\; K_{h_n}\!\left(Z_k^{-} - x\right),
\\[0.6em]
{A_k^{(2)}} \;&=\;
{\frac{\I_{\{Z_{k-1} \le x \le Z_k^{-}\}}}{h'(x)\Delta(x)}},
&\qquad 
\widetilde{A}_k^{(2)} \;&=\;
{\frac{\I_{\{Z_{k-1} \le x \le Z_k^{-}\}}}{\Delta(x)}},
\\[0.6em]
B_k^{(i)} \;&=\; \mathbb{E}\!\left[A_k^{(i)} \mid F_{k-1}\right],
&\qquad 
\widetilde{B}_k^{(i)} \;&=\; \mathbb{E}\!\left[\widetilde{A}_k^{(i)} \mid G_{k-1}\right].
\end{aligned}
\]
It is worth noting that in the case of $\lambdak{n}(x)$, because the jumps are deterministic, the numerator of $A_k^{(2)}$ can be rewritten as a function of $Z_k$ and $Z_{k-1}$ only, with $\I_{\{Z_{k-1}\leq x\leq Z_k^-\}}=\I_{\{Z_{k-1}\leq x, Z_k\geq h(x)\}}$.

We also consider $R_n = (R_n^{(1)},R_n^{(2)})^\top$ and
$\widetilde{R}_n = (\widetilde{R}_n^{(1)},\widetilde{R}_n^{(2)})^\top$ defined by
\[
\begin{aligned}
R_n^{(1)} \;&=\; \frac{1}{n}\sum_{k=1}^nB_k^{(1)}-\mu(h(x)),
&\qquad 
\widetilde{R}_n^{(1)} \;&=\; \frac{1}{n}\sum_{k=1}^n\widetilde{B}_k^{(1)}-\mu^-(x),
\\[0.6em]
R_n^{(2)}\;&=\;\frac{1}{n}\sum_{k=1}^nB_k^{(2)}-\frac{D(x)}{h'(x)},
&\qquad 
\widetilde{R}_n^{(2)} \;&=\;\frac{1}{n}\sum_{k=1}^n\widetilde{B}_k^{(2)}-D(x).
\end{aligned}
\]

Then a direct calculus shows that the two-dimensional estimation errors are given by
\begin{equation}
    \label{eq:decompo:error}
\left(\begin{matrix}
 \widehat{\mu}_n(h(x))\\ \widehat{D}_n(x)/h'(x)
\end{matrix}\right) -
\left(\begin{matrix}
 \mu(h(x))\\ D(x)/h'(x)
\end{matrix}\right)
= \frac{M_n}{n} + R_n
\end{equation}
and 
\begin{equation*}
\left(\begin{matrix}
 \widehat{\mu}^-_n(x)\\ \widehat{D}_n(x)
\end{matrix}\right) -
\left(\begin{matrix}
 \mu^-(x)\\ D(x)
\end{matrix}\right)
= \frac{\widetilde{M}_n}{n} + \widetilde{R}_n.
\end{equation*}

These two decompositions are at the core of our demonstration. In order to study $\lambdak{n}(x)$, we establish that $M_n/n$ almost surely goes to $0$ by virtue of the law of large numbers for vector martingales (see Appendix~\ref{proof_as}), and state that the remainder term $R_n$ also vanishes, establishing the consistency of the estimator given in Theorem~\ref{main_ps}. For the asymptotic normality, we first state that $n^{(1-\gamma)/2}R_n$ goes to $0$ in probability (see Appendix~\ref{proof_tcl}). The convergence of the estimation error is therefore carried by the martingale term, which we study in light of the central limit theorem for vector martingales. This proves the asymptotic normality of $\lambdak{n}(x)$ given in Theorem~\ref{main_tcl}. The investigations for $\widetilde{M}_n$ and $\widetilde{R}_n$ are not presented in this paper but follow the same steps and yield similar convergence properties for $\lambdaks{n}(x)$.

\section{Comparison in the TCP model}
\label{s:tcp}

\subsection{TCP process}
\label{ss:tcp:tcp}

The TCP model is presented in Subsection~\ref{ss:intro:tcp} as the PDMP with jump rate $\lambda(x)=x$, deterministic flow $\Phi(t|x) = x+t$, and transition kernel $Q(\cdot|x) = \delta_{\{\kappa x\}}$. It adheres to the most restrictive assumption made in this paper, namely deterministic fragmentation. However, it does not match the assumptions made for the theoretical study of $\lambdaks{n}$ in the present paper and of $\lambdaamgo{n}$ in \cite{AM16}, where the kernel $Q$ is assumed to admit a density with respect to the Lebesgue measure. Nevertheless, both investigations essentially require that the kernel of an embedded chain (the kernel $R$ of $Z_n^-$ for $\lambdaks{n}$ and the kernel $P$ of $Z_n$ for $\lambdaamgo{n}$) possess a density (which is the case when $Q$ itself has a density in the given frameworks). Indeed, these Markov chains lie at the core of the theoretical analyses. In the context of the TCP model, it is straightforward to show that the aforementioned discrete-time kernels admit a density (see Subsection~\ref{ss:modelass} or \cite{K16}), and thus the asymptotic normalities given in Theorems~\ref{theo:lambdaamg} and \ref{main_tcl} remain valid in this setting. As a consequence, the TCP model is an excellent example for comparing the theoretical variances of the three methods considered, completing Remarks~\ref{rem:comp:var:k:ks} and \ref{rem:comp:var:k:ks:amg}, and also for confronting the theory with numerical simulations.

As noted in Remark~\ref{rem:comp:var:k:ks}, in this model, variance $\sigmak{x}$ is smaller than variance $\sigmaks{x}$ by a factor of $\kappa$, uniformly in $x$. These two variances are compared to variance $\sigmaamg{x}$ (given in Theorem~\ref{theo:lambdaamg}) for different values of $\kappa$ in Figure~\ref{fig:simu:sdclt}. This is a numerical evaluation based on formula \eqref{eq:mu:tcp} (together with $\mu^-(x) = \kappa\mu(\kappa x)$) without any random generation. This comparison illustrates that, within the same model, the variances intersect, showing that, at least by this criterion, none of the techniques is uniformly superior to the others. However, this comparison criterion does not account for the bias of the estimators or the rate of convergence in the central limit theorems as mentioned in Remark~\ref{rem:comp:var:k:ks:amg}. For the rest of the simulation study, we restrict ourselves to $\kappa=0.4$ for which $\sigmaamg{x}$ intersects $\sigmak{x}$ around $x=2$ and $\sigmaks{x}$ around $x=2.5$.

\begin{figure}[ht]
    \centering
    \begin{tabular}{ccc}
    $\kappa=0.3$ & $\kappa=0.4$ & $\kappa=0.5$ \\
    \includegraphics[height=3.3cm]{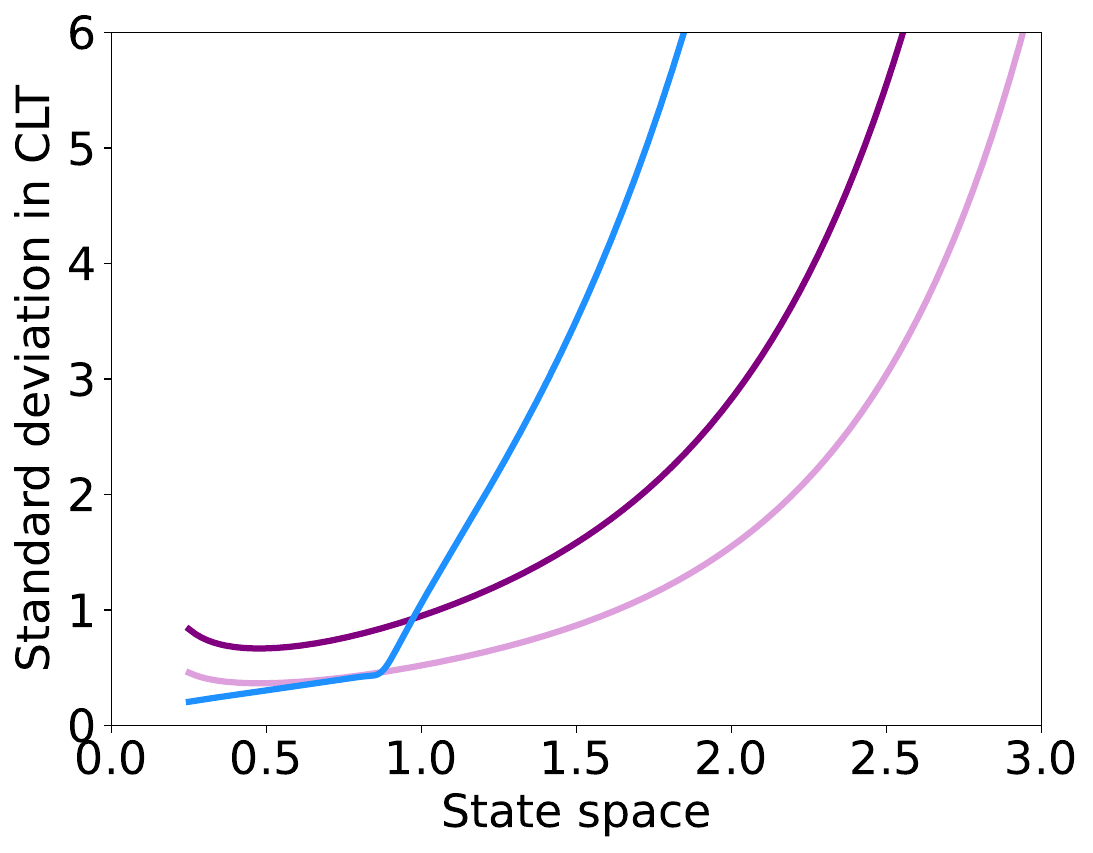}&
    \includegraphics[height=3.3cm]{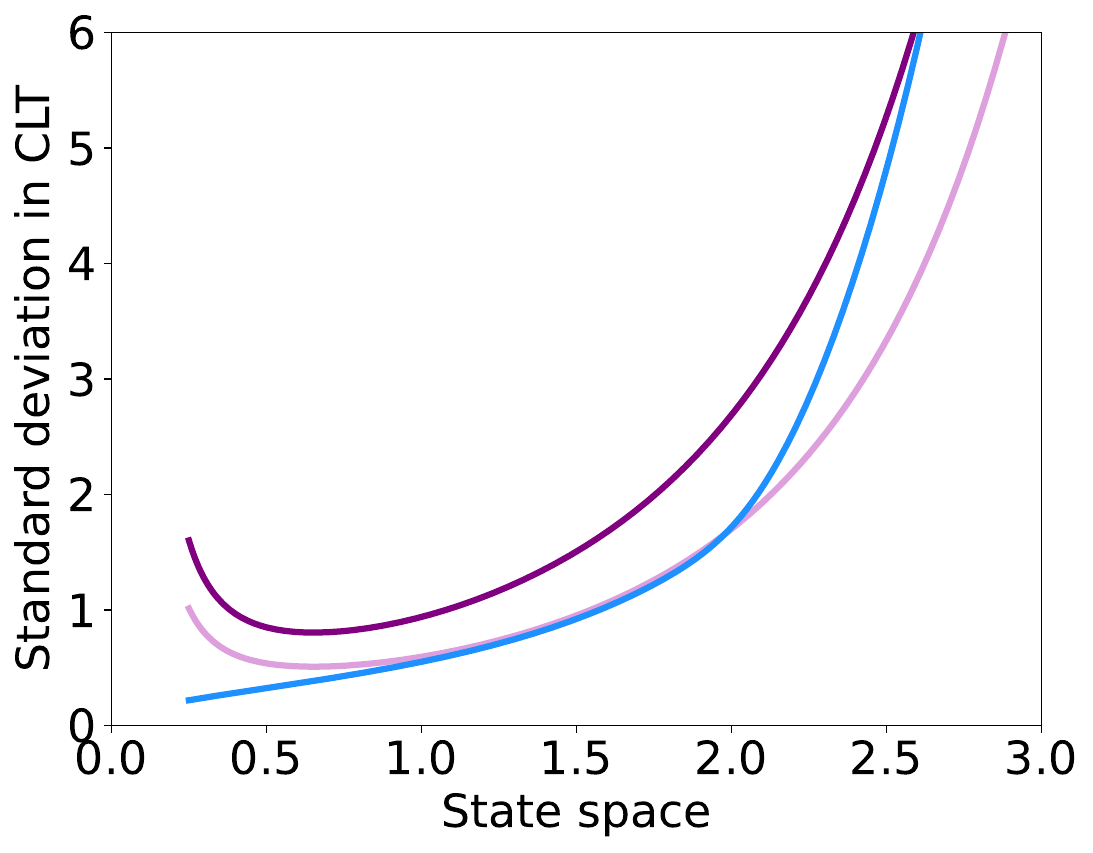}&
    \includegraphics[height=3.3cm]{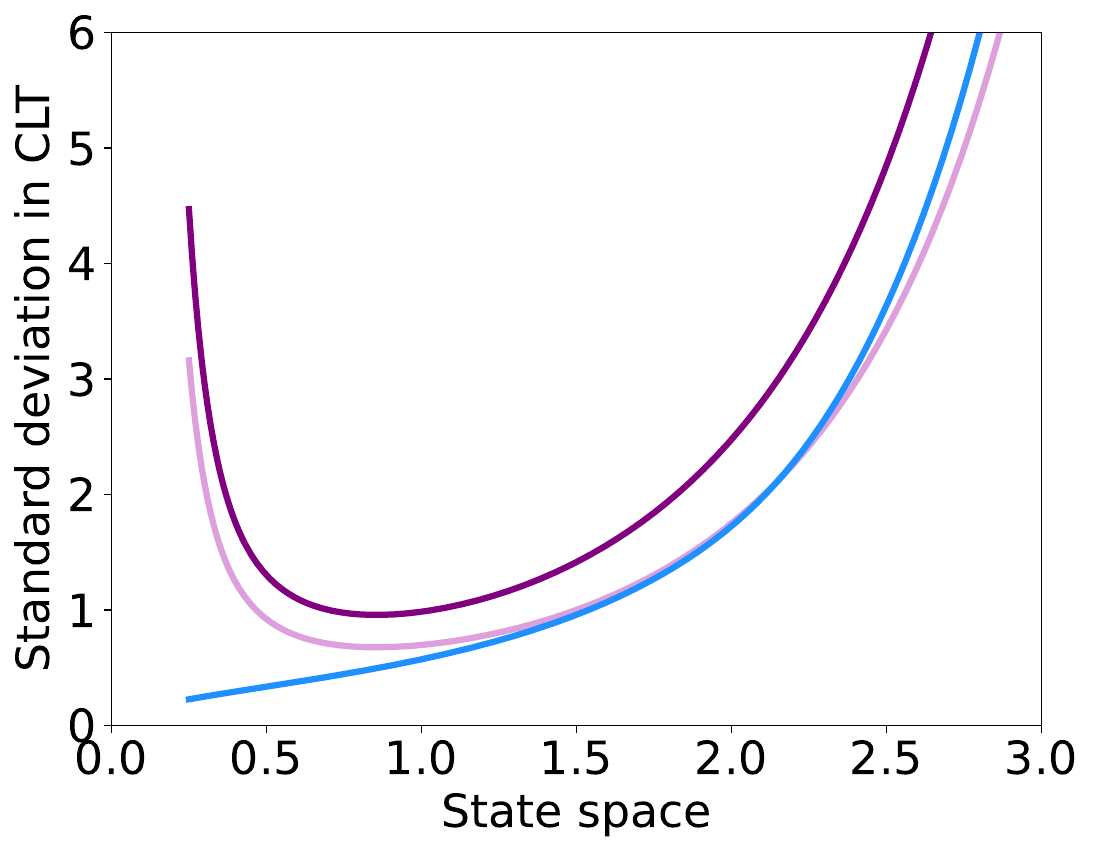} \vspace{0.5cm}\\
    
    $\kappa=0.6$ & $\kappa=0.7$ & $\kappa=0.8$ \\
    \includegraphics[height=3.3cm]{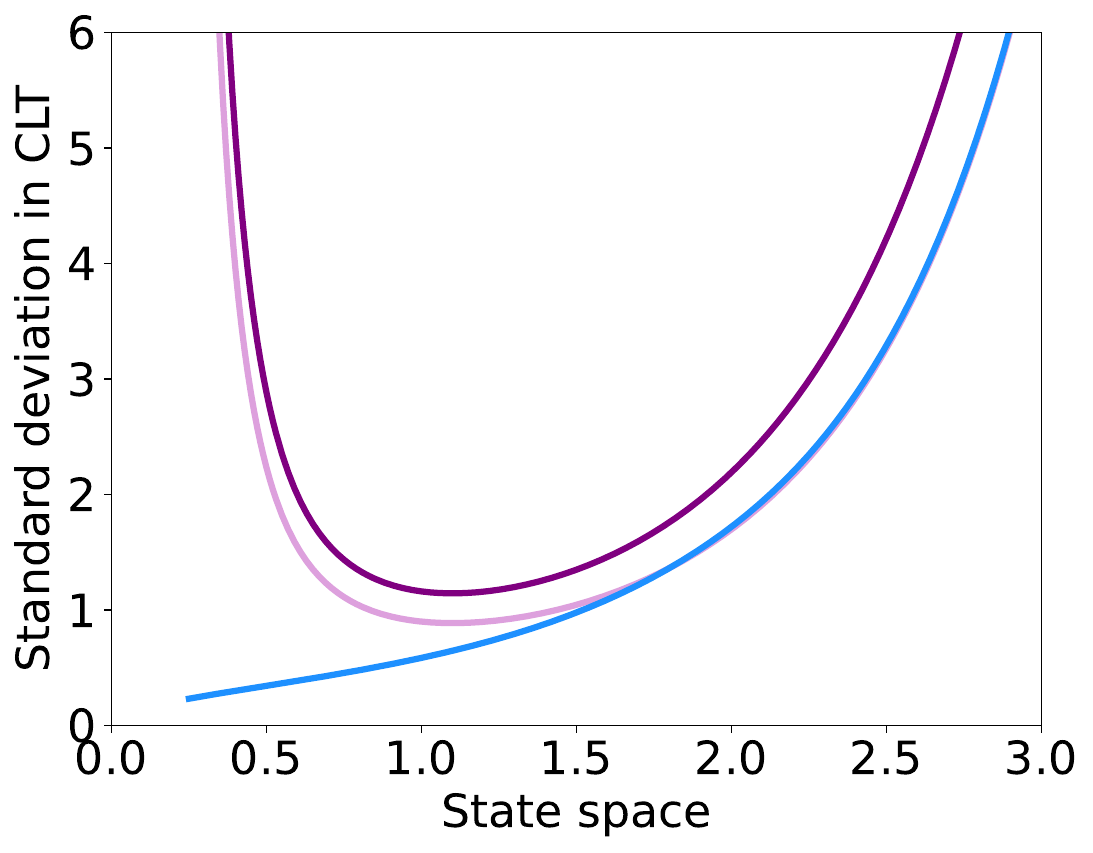}&
    \includegraphics[height=3.3cm]{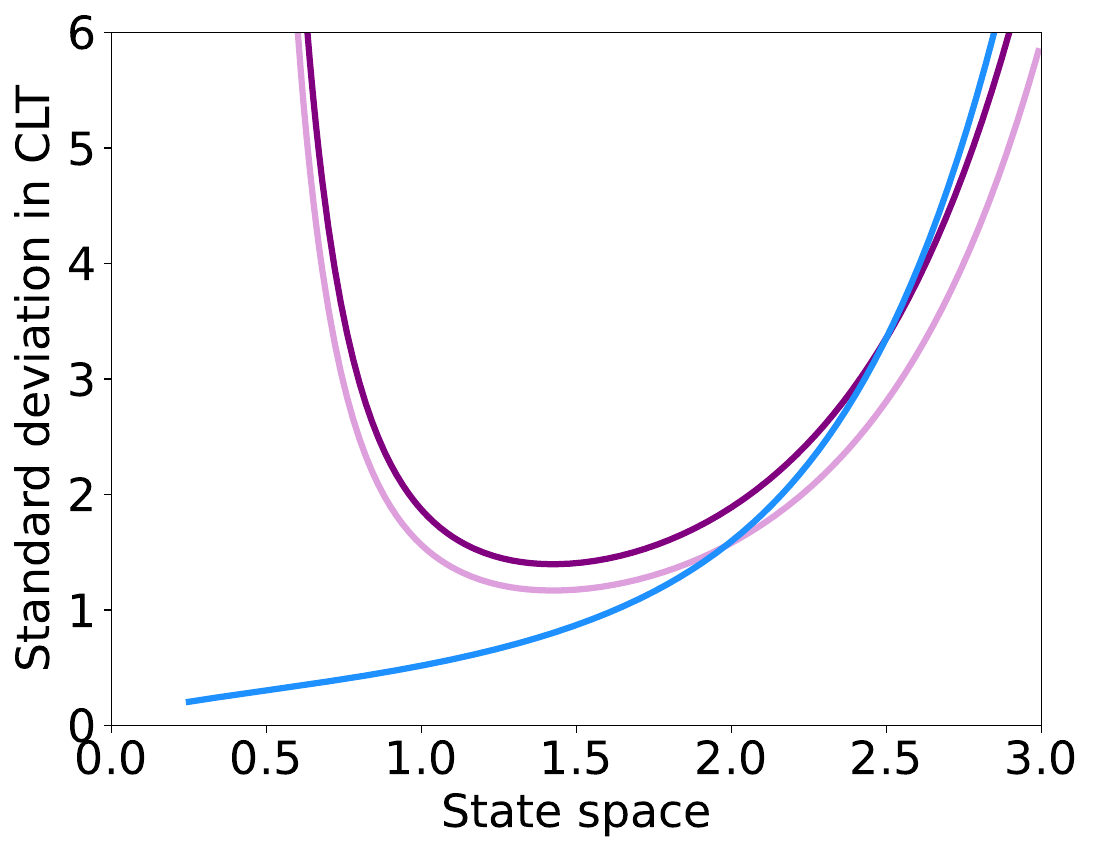}&
    \includegraphics[height=3.3cm]{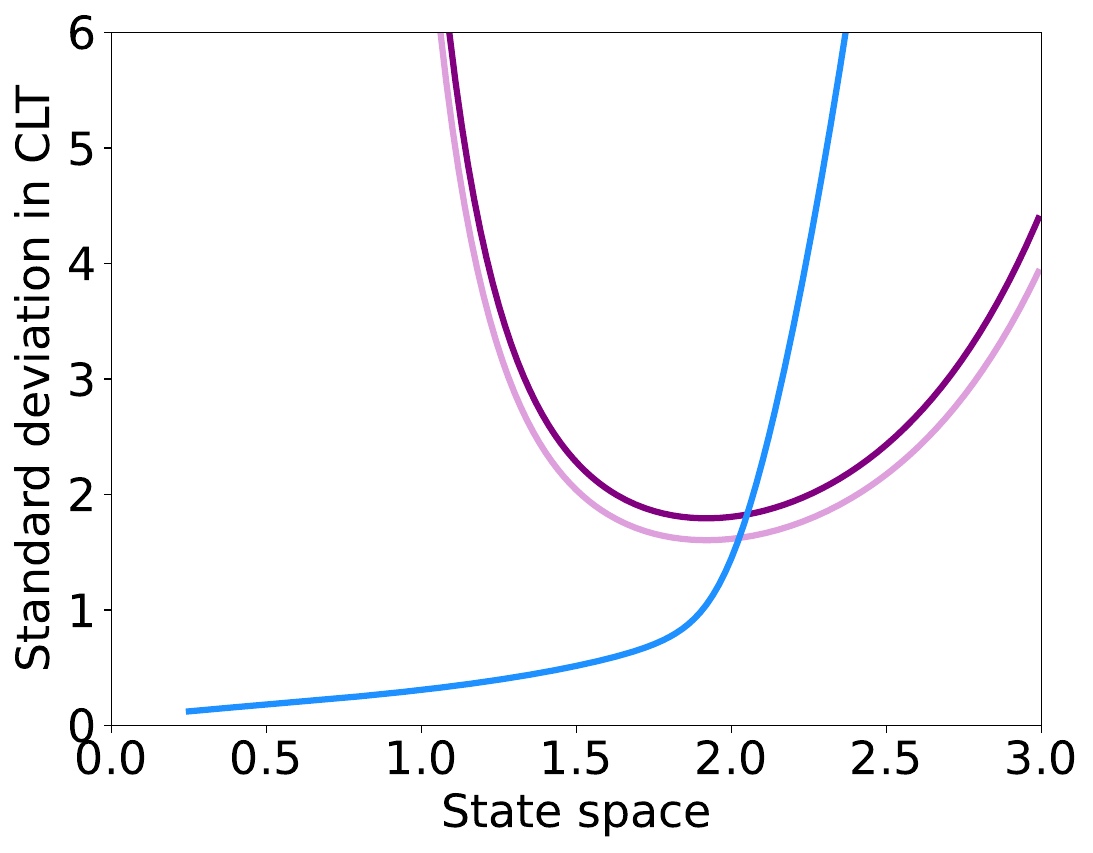}
    \end{tabular}
    \caption{Standard deviations $\sqrt{\sigmak{x}}$ in pink, $\sqrt{\sigmaks{x}}$ in purple and $\sqrt{\sigmaamg{x}}$ in light blue (given in Theorems~\ref{theo:lambdaamg} and \ref{main_tcl}) for the TCP model with different values of $\kappa$.}
    \label{fig:simu:sdclt}
\end{figure}

\subsection{Bandwidth selection}
\label{ss:tcp:bandwidth}

Kernel methods are highly sensitive to the smoothing parameter, which must therefore be chosen carefully. The estimators under consideration rely on estimating different invariant measures. For instance, $\lambdaks{n}$ involves estimating $\mu^-$, whereas $\lambdak{n}$ relies on estimating $\mu$. Consequently, there is no reason for them to share the same bandwidth. Moreover, there is no a priori justification for selecting the bandwidth independently for the kernel component of the jump rate estimator; instead, it should be based on the overall formula of the estimator. Additionally, developing a method for adjusting the smoothing bandwidth is a challenging and interesting question that goes far beyond the scope of this paper. Here, we propose a bandwidth selection procedure for each of the estimators based on our knowledge of the ground truth, which, therefore, can not be applied in practice to real-world data.

For the estimator $\widehat{\lambda}$, the smoothing parameter selected in this simulation study is the minimizer of the integrated square error
$$ \int_{0.5}^{2.5} \left[\widehat{\lambda}(x) -\lambda(x)\right]^2\ud x,$$
where the integral is evaluated with a step of $0.05$ and the numerical minimization is performed with a step of $0.025$ for $\lambdak{n}$ and $\lambdaks{n}$ and a step of $0.01$ in space and of $0.05$ in time for $\lambdaamgo{n}$ and $\lambdaamg{n}$. Indeed, it should be noted that the latter two estimators, given in \eqref{eq:def:amgo} and \eqref{eq:def:amg}, depend on two bandwidths, in space and in time. In addition, the optimal argument selection in $\lambdaamg{n}$ is operated with the same spatial bandwidth as the one used to compute the estimator. Figure~\ref{fig:simu:bandwidth} presents the results obtained over $100$ replicates for the four estimators from two different sample sizes.

\begin{figure}[ht]
    \centering
    \begin{tabular}{ccc}
    $n=1\,000$ & \hspace{0.25cm} &$n=10\,000$\\
    \includegraphics[height=4.4cm]{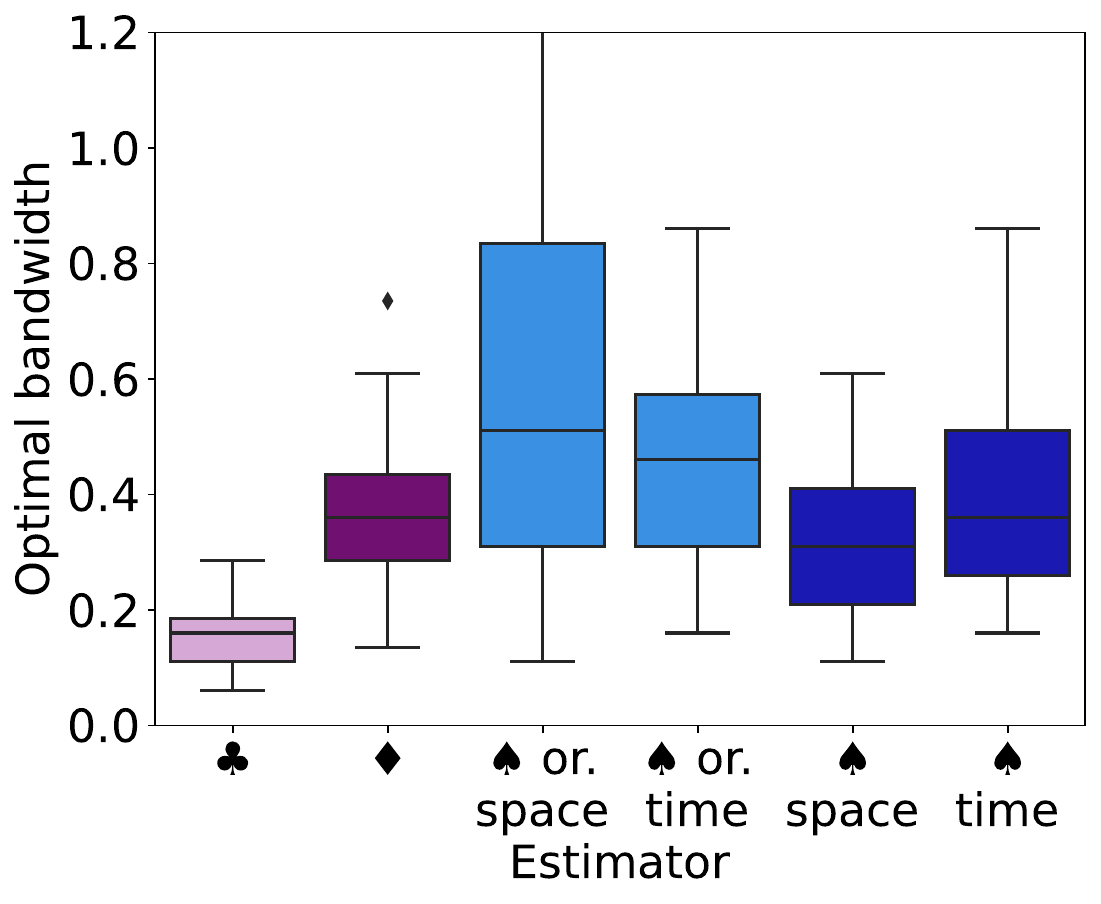}&&
    \includegraphics[height=4.4cm]{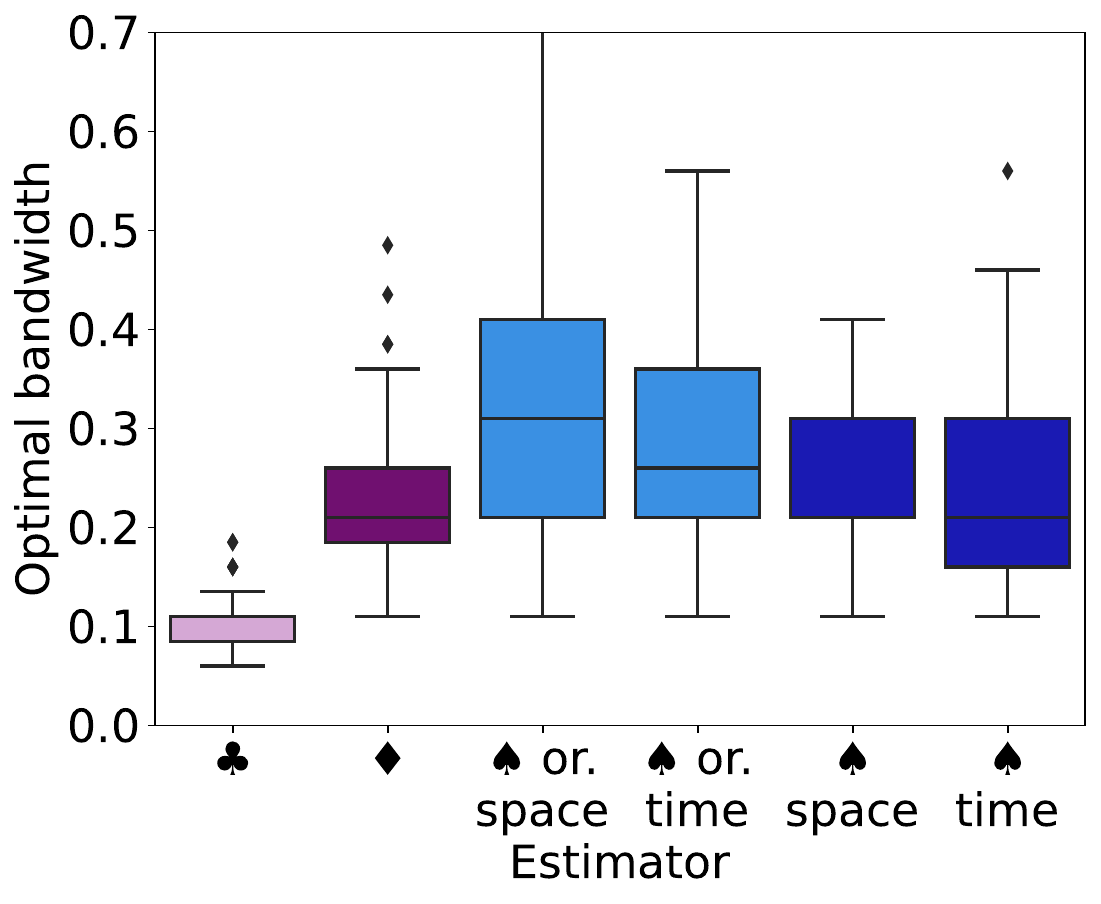}
    \end{tabular}
    \caption{Boxplots over $100$ replicates of smoothing bandwidths minimizing the integrated square error for the four estimators $\lambdak{n}$ in pink, $\lambdaks{n}$ in purple, $\lambdaamgo{n}$ in light blue and $\lambdaamg{n}$ in dark blue evaluated from trajectories of size $n=1\,000$ (left) and of size $n=10\,000$ (right) generated from the TCP model with parameter $\kappa=0.4$.}
    \label{fig:simu:bandwidth}
\end{figure}

The optimal bandwidths vary significantly depending on the estimator used, with a noticeable difference even between estimators $\lambdaamgo{n}$ and $\lambdaamg{n}$, though this difference diminishes with datasets of $10\,000$ observations. Furthermore, as expected, the bandwidths decrease as the number of data points increases. The variability is highest for estimators $\lambdaamgo{n}$ and $\lambdaamg{n}$, where smoothing is applied both to the numerator and the denominator, in both space and time. The variability of the bandwidths for all four estimators also decreases as the dataset size grows. To obtain smoothing parameters that depend only on the type of estimator and the number of available data points (and not on the data themselves), the medians of these empirical distributions are selected. Figure~\ref{fig:simu:examples} illustrates the behavior of the four estimators evaluated with these optimal bandwidths from a single trajectory.

\begin{figure}[ht]
    \centering
    \begin{tabular}{ccc}
    $n=1\,000$ & \hspace{0.25cm} &$n=10\,000$\\
    \includegraphics[height=4.4cm]{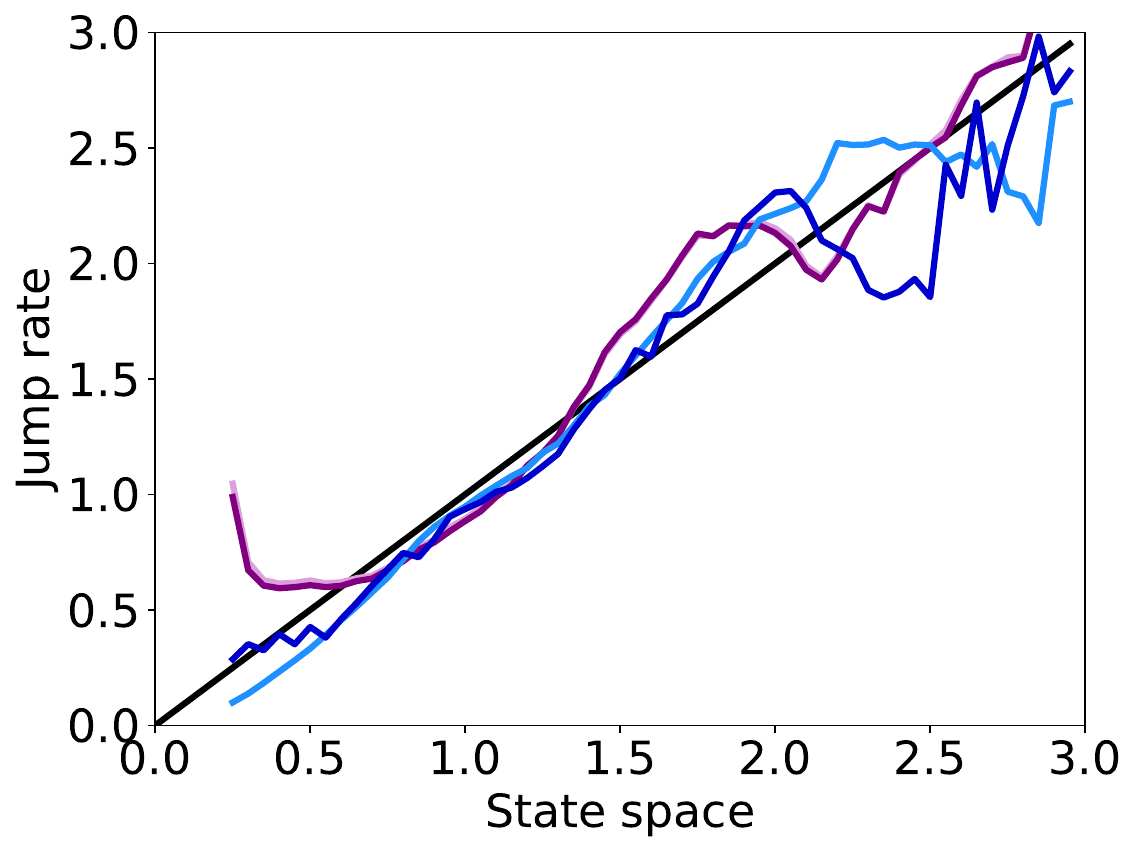}&&
    \includegraphics[height=4.4cm]{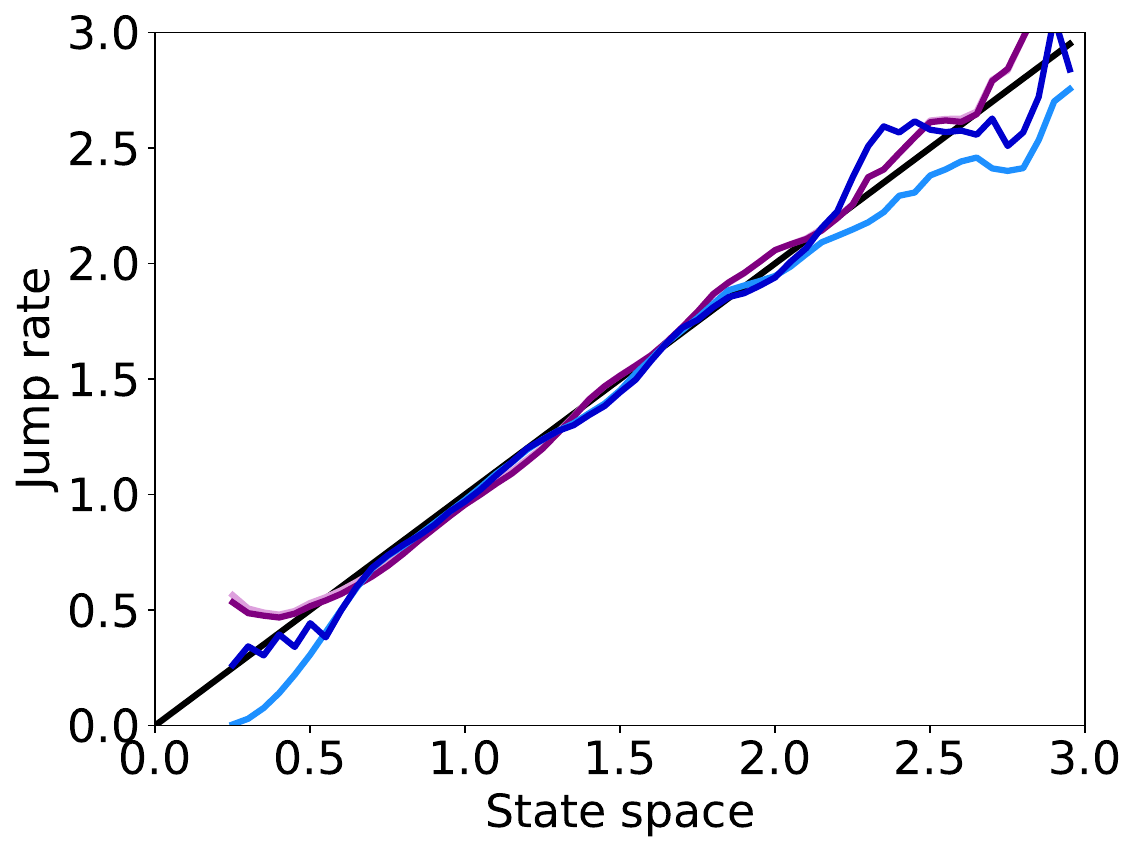}
    \end{tabular}
    \caption{Estimates $\lambdak{n}(x)$ in pink, $\lambdaks{n}$ in purple, $\lambdaamgo{n}(x)$ in light blue and $\lambdaamg{n}(x)$ in dark blue computed with the optimal bandwidths (medians of boxplots of Figure~\ref{fig:simu:bandwidth}) from a single trajectory of size $n=1\,000$ (left) and $n=10\,000$ (right) generated from the TCP model with $\kappa=0.4$. The referenced jump rate is given in black line.}
    \label{fig:simu:examples}
\end{figure}


\subsection{Asymptotic variances vs. estimation errors}
\label{ss:tcp:var}

The task now is to compare the asymptotic variances of the estimators while taking into account the rate of convergence in the central limit theorem. It should be noted that this can not be done for estimator $\lambdaamg{n}$, for which asymptotic normality has not been theoretically established. According to Theorems~\ref{theo:lambdaamg} and \ref{main_tcl}, and as mentioned in Remark~\ref{rem:comp:var:k:ks:amg}, the pointwise variance in the estimation of $\lambda(x)$ that we expect to observe is normalized by the rate of convergence and thus of the order $\sigmak{x}/(n\hk{n})$ for $\lambdak{n}(x)$, $\sigmaks{x}/(n\hks{n})$ for $\lambdaks{n}(x)$ and $\sigmaamg{x}/(n\hamgs{n}\hamgt{n})$, where $\hk{n}$, $\hks{n}$, $\hamgs{n}$ and $\hamgt{n}$ are the optimal bandwidths selected above (as medians of boxplots of Figure~\ref{fig:simu:bandwidth}). Figure~\ref{fig:simu:3tcl} illustrates the asymptotic normality at $x=2$.

\begin{figure}[ht]
    \centering
    \begin{tabular}{ccc}
    \multicolumn{3}{c}{$n=1\,000$} \\
    \includegraphics[height=3.3cm]{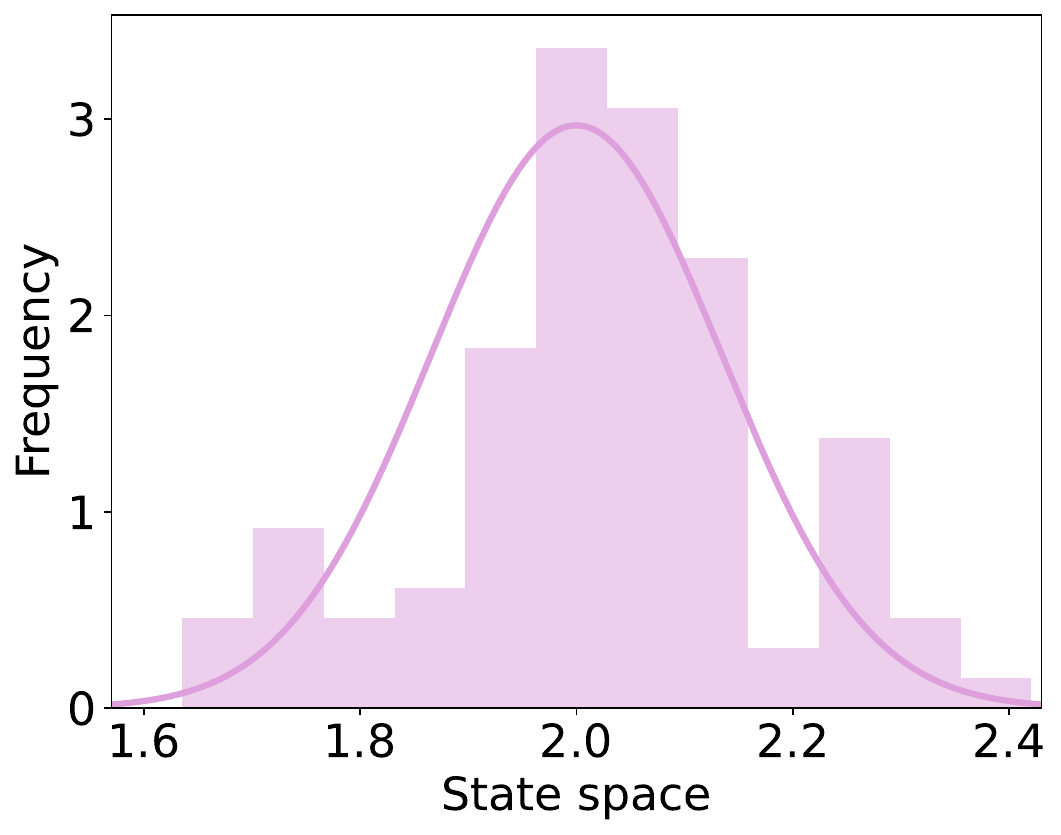}&
    \includegraphics[height=3.3cm]{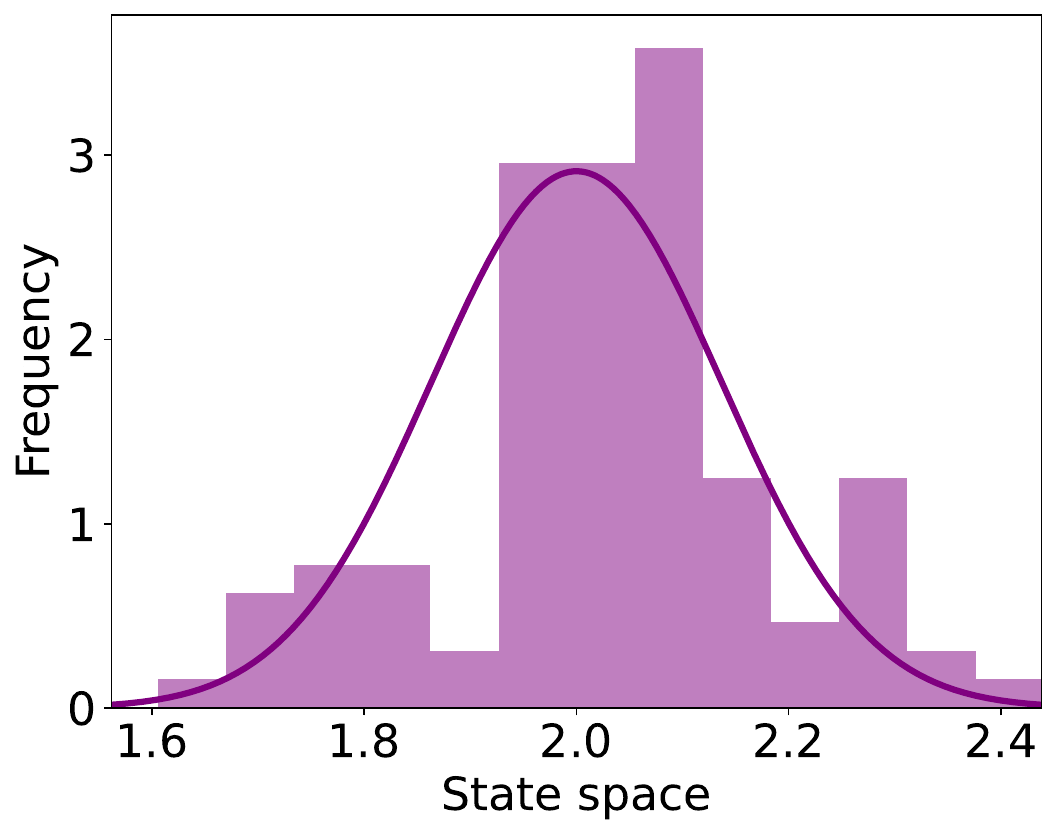}&
    \includegraphics[height=3.3cm]{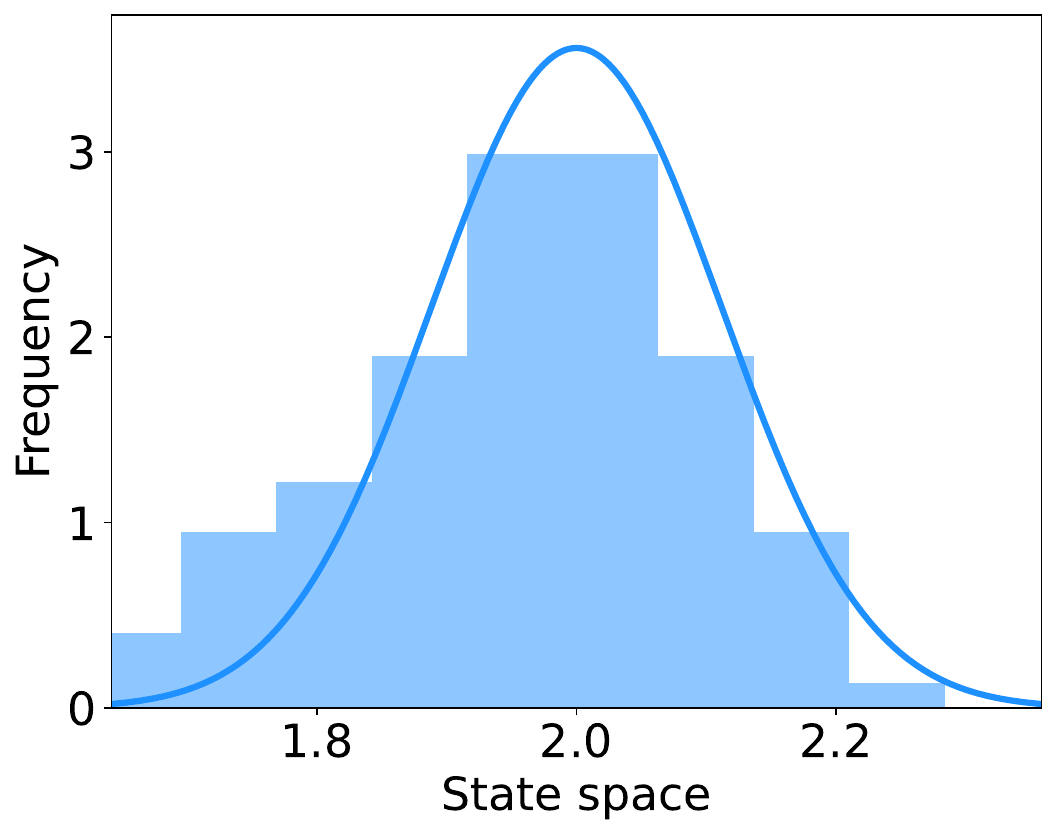}\\
    $\clubsuit$&$\vardiamondsuit$&$\spadesuit$ (oracle)\vspace{0.5cm}\\
    
    \multicolumn{3}{c}{$n=10\,000$} \\
    \includegraphics[height=3.3cm]{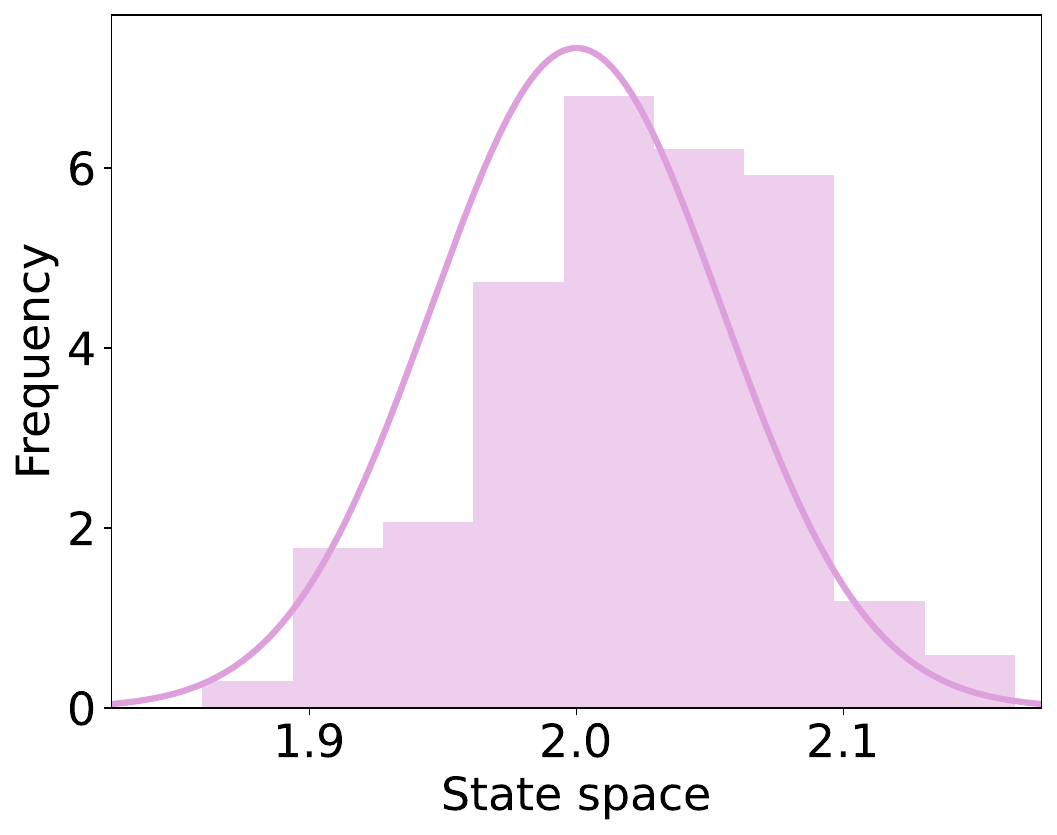}&
    \includegraphics[height=3.3cm]{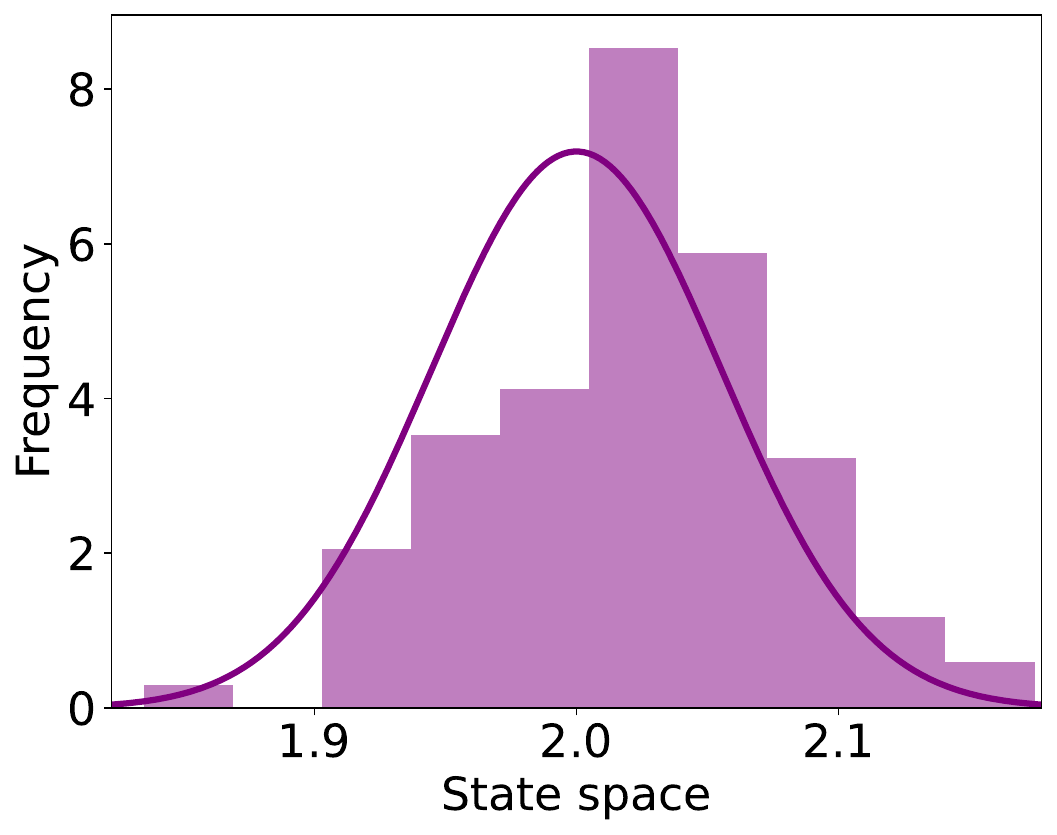}&
    \includegraphics[height=3.3cm]{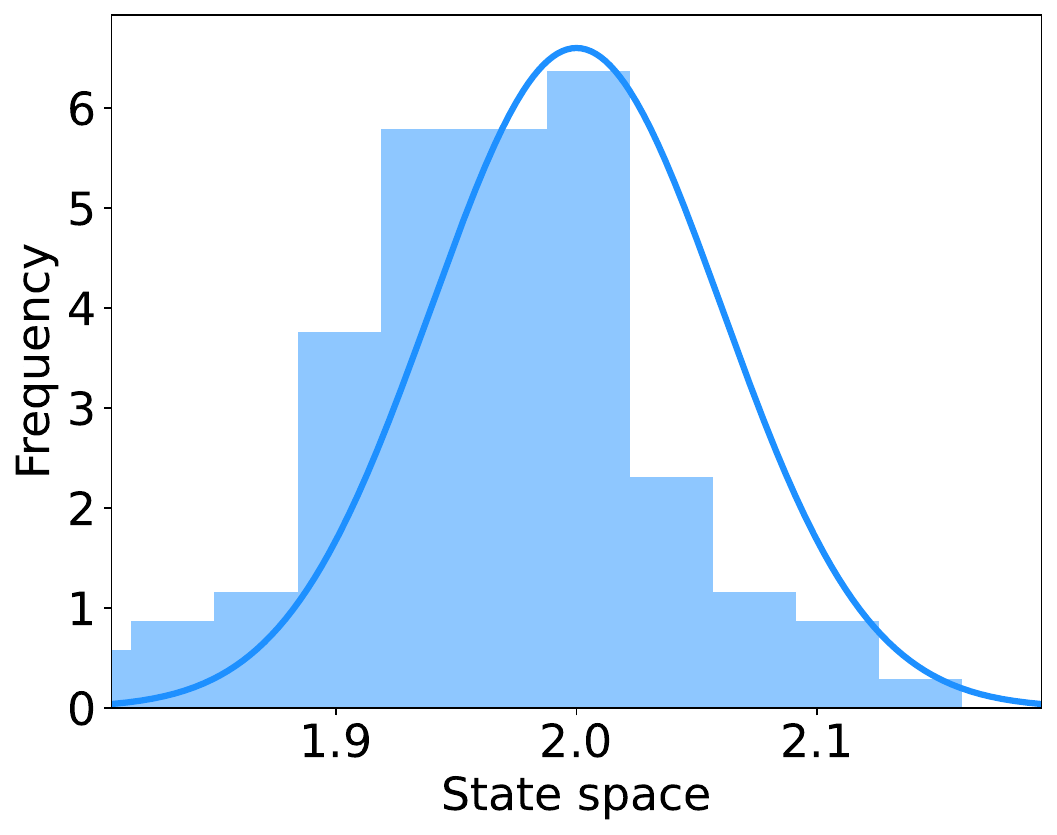} \\
    $\clubsuit$&$\vardiamondsuit$&$\spadesuit$ (oracle)
    \end{tabular}
    \caption{Distribution over $100$ replicates of $\lambdak{n}(2)$ in pink (left), $\lambdaks{n}(2)$ in purple (center) and $\lambdaamgo{n}(2)$ in light blue (right) evaluated with optimal bandwidth parameters $\hk{n}$, $\hks{n}$, $\hamgs{n}$ and $\hamgt{n}$, from trajectories of size $n=1\,000$ (top) and of size $n=10\,000$ (bottom) generated from the TCP model with parameter $\kappa=0.4$, and referenced Gaussian distributions with mean $\lambda(2)=2$ and variance $\sigmak{x}/(n\hk{n})$ in pink line (left), with variance $\sigmaks{x}/(n\hks{n})$ in purple line (center), and with variance $\sigmaamg{x}/(n\hamgs{n}\hamgt{n})$ in light blue line (right).}
    \label{fig:simu:3tcl}
\end{figure}

The empirical distributions shown in Figure~\ref{fig:simu:3tcl} align reasonably well with the expected theoretical distributions, particularly for trajectories of size $10\,000$. It should be noted that this analysis falls outside the scope of the theorems, as the bandwidths were chosen numerically and do not follow the theoretically prescribed decay scheme. A slight bias is observable when $n=10\,000$, which highlights (as is well-known) that asymptotic variance alone is insufficient to fully characterize the desirable properties of an estimator. Figure~\ref{fig:simu:norm:sdclt} presents the three normalized variances across the state space.

\begin{figure}[ht]
    \centering
    \begin{tabular}{ccc}
    $n=1\,000$ & \hspace{0.25cm} &$n=10\,000$\\
    \includegraphics[height=4.4cm]{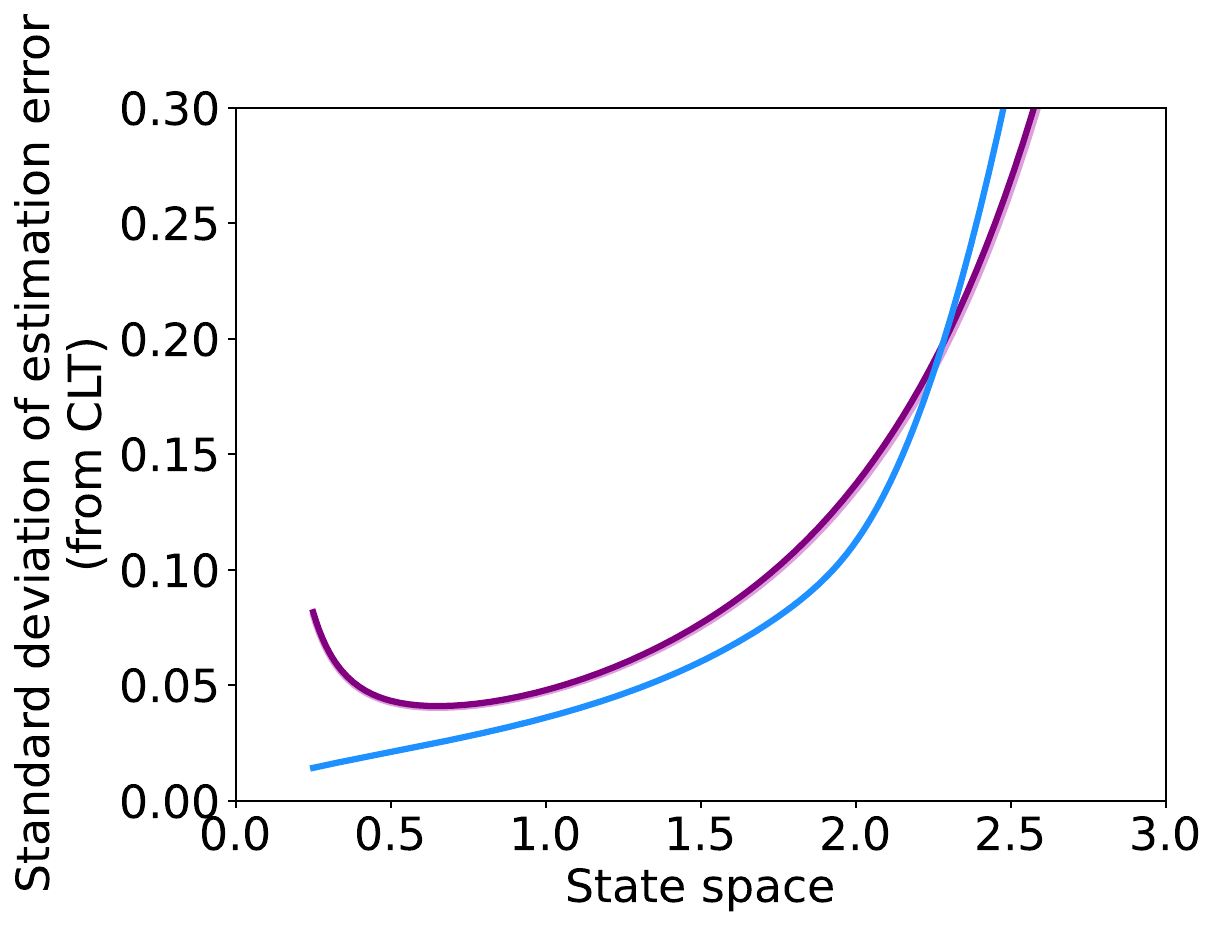}&&
    \includegraphics[height=4.4cm]{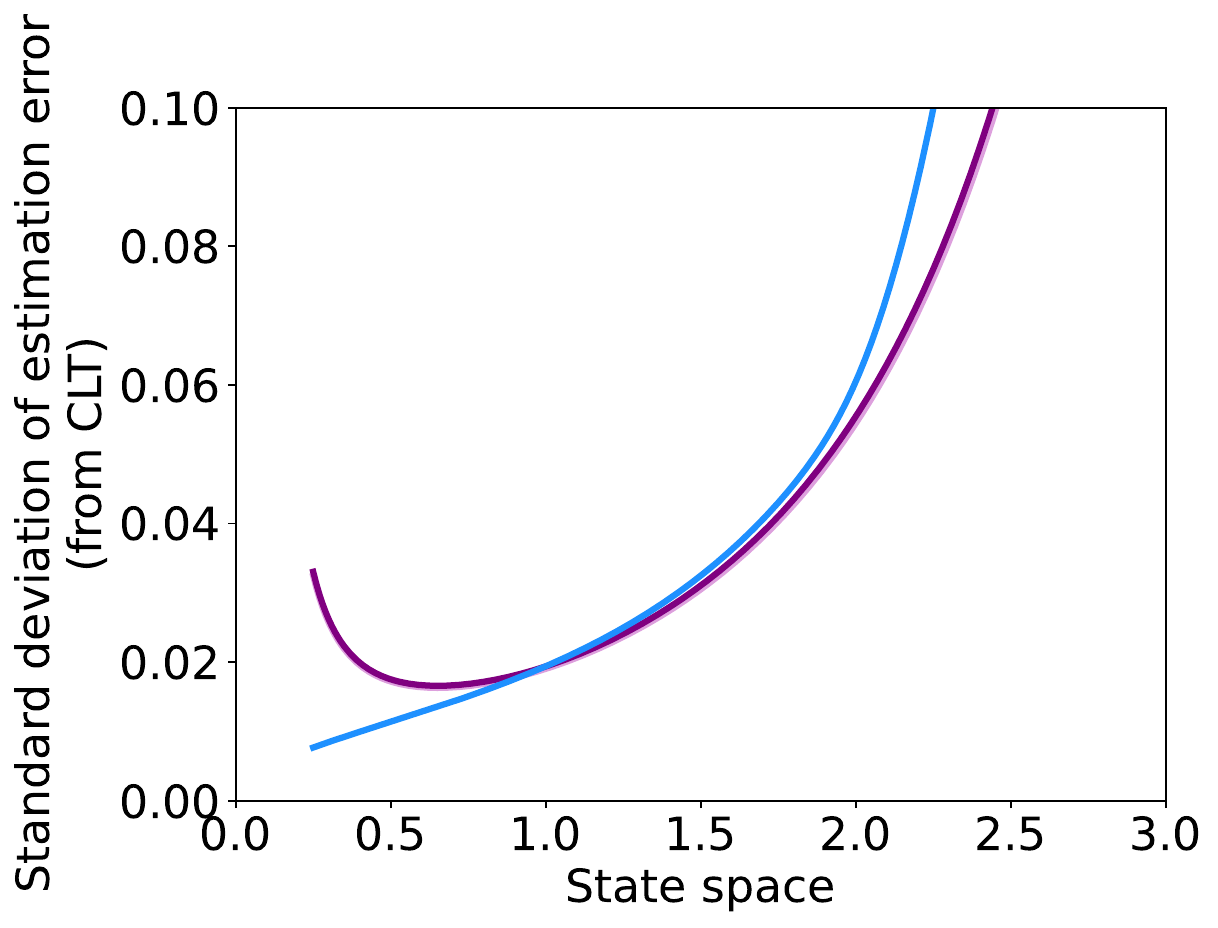}
    \end{tabular}
    \caption{Normalized standard deviations $\sqrt{\sigmak{x}/(n\hk{n})}$ in pink, $\sqrt{\sigmaks{x}/(n\hks{n})}$ in purple and $\sqrt{\sigmaamg{x}/(n\hamgs{n}\hamgt{n})}$ in light blue (where $\hk{n}$, $\hks{n}$, $\hamgs{n}$ and $\hamgt{n}$ are the bandwidths selected as medians of boxplots of Figure~\ref{fig:simu:bandwidth}) in the TCP model with parameter $\kappa=0.4$, and with $n=1\,000$ (left) and $n=10\,000$ (right).}
    \label{fig:simu:norm:sdclt}
\end{figure}

The normalized variances in Figure~\ref{fig:simu:norm:sdclt}, which approximate the variance of the estimation error for a finite data set, show that on the criterion of variance, the quality of estimation of the two estimators $\lambdak{n}$ and $\lambdaks{n}$ is very similar (since the curves almost coincide), as was already apparent in the example in Figure~\ref{fig:simu:examples}. Moreover, despite having a slower rate of convergence, the normalized variance of estimator $\lambdaamgo{n}$ remains better than the other two for certain regions of the state space: specifically, before approximately $x=2.25$ for $n=1\,000$ and before $x=1$ for $n=10\,000$.

These results suggest that oracle estimator $\lambdaamgo{n}$, and consequently estimator $\lambdaamg{n}$ (assuming their behaviors are expected to be similar as explained in Remark~\ref{rem:lambdaamgo:vs:lambdaamg}), should perform better than the other two on the left side of the state space, and conversely on the right, with a boundary around $x=2.25$ for 
$n=1\,000$ and $x=1$ for $n=10\,000$. However, this prediction does not account for the estimators' bias, the potential differences in behavior between estimators $\lambdaamgo{n}$ and $\lambdaamg{n}$, or the gap between theoretical results and the setup of numerical simulations. We now aim to assess the accuracy of this prediction. To this end, we evaluated over $100$ replicates the pointwise estimation error $|\widehat{\lambda}(x)-\lambda(x)|$ across the state space for the four estimators under consideration. Results are given in Figure~\ref{fig:simu:estim:errors}.

\begin{figure}[ht]
    \centering
    $n = 1\,000$\\
    \includegraphics[height=4.4cm]{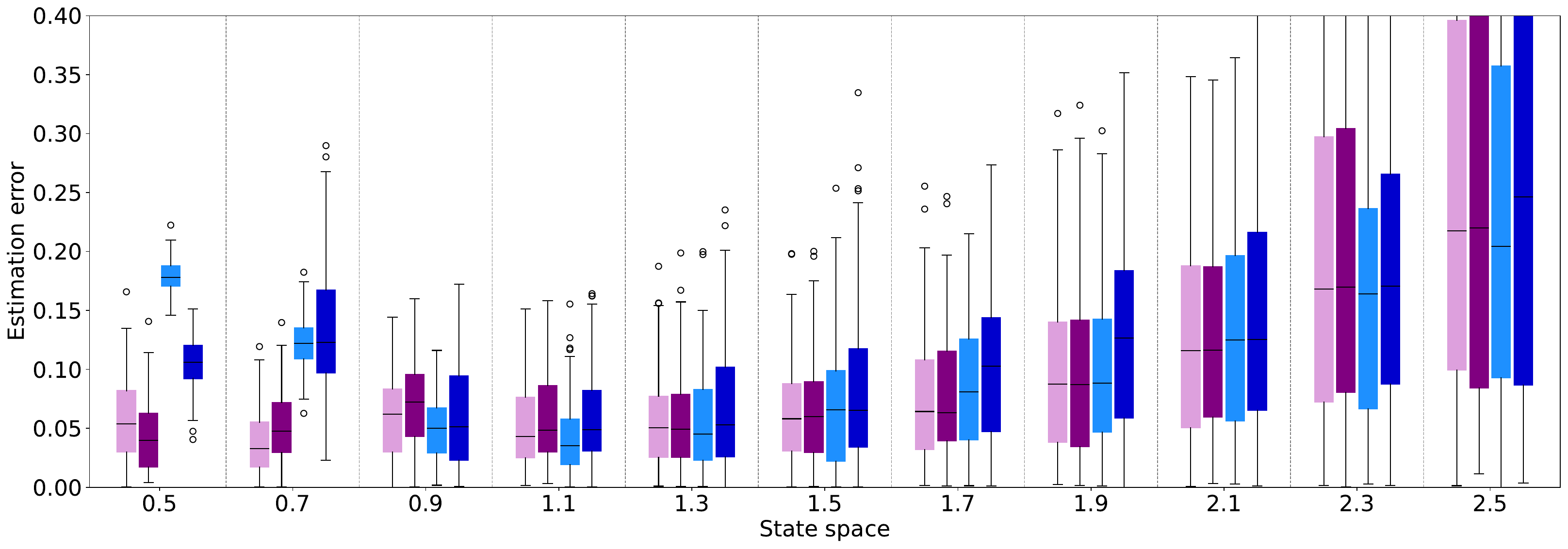}\vspace{0.5cm}\\
     $n = 10\,000$ \\
    \includegraphics[height=4.4cm]{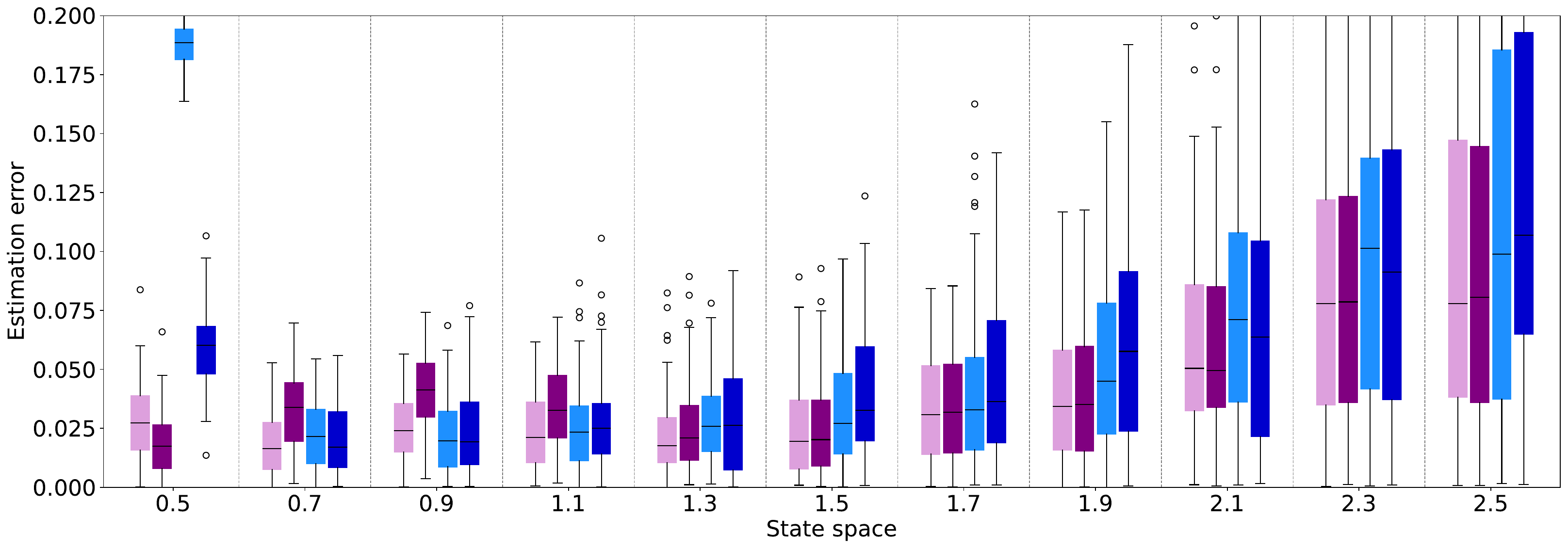}
    \caption{Pointwise estimation error over $100$ replicates for the four estimators $\lambdak{n}(x)$ in pink, $\lambdaks{n}(x)$ in purple, $\lambdaamgo{n}(x)$ in light blue and $\lambdaamg{n}(x)$ in dark blue, from left to right for each value of $x$ (with $x$ between $0.5$ and $2.5$ with a step of $0.2$), computed from a trajectory of size $n=1\,000$ (top) and of size $n=10\,000$ (bottom).}
    \label{fig:simu:estim:errors}
\end{figure}

The general pattern of variability aligns fairly well with the theoretical predictions from Figure~\ref{fig:simu:norm:sdclt}. However, in addition to exhibiting increasing variance, all four estimators appear to suffer from a bias toward the right side of the state space ($x>2$), regardless of the sample size. On the left side of the state space ($x<1$), only estimators $\lambdaamgo{n}$ and $\lambdaamg{n}$ experience significant bias, along with greatly reduced variance (as predicted by theory). In the central and right parts of the state space ($x>1$), all four estimators display very similar behavior, with estimator $\lambdak{n}$ having a nearly uniform advantage.

Before delving into details, we exclude estimator $\lambdak{n}$, which uses the form of the transition kernel, and the oracle estimator $\lambdaamgo{n}$, whose argument selection relies on the invariant law of the process. This leaves us comparing estimators $\lambdaks{n}$ and $\lambdaamg{n}$, which are based on exactly the same data. In this comparison, aside from the bias suffered by $\lambdaamg{n}$ on the left side of the state space, the theoretical predictions are generally supported by the numerical simulations, especially from trajectories of size $n=10\,000$: for $0.7\leq x\leq 1.1$, $\lambdaamg{n}(x)$ performs better than $\lambdaks{n}(x)$, while for $x\geq1.3$, the situation reverses.

While the previously mentioned limitations of applying theoretical results in this context remain valid (and are particularly evident when comparing the four estimators in detail), they nonetheless allow for a reasonably reliable comparison of the two main estimators studied. Moreover, the numerical analysis confirms, using the TCP model as an example, that none of the methods is uniformly better across the state space.

\subsection{Adaptive estimators}\label{ss:adaptive:simus}

The estimator constructed in \cite{KS21} relies on an adaptive projection estimator of the invariant distribution $\mu^-$. Although the theoretical framework of the present paper covers neither projection-based estimation nor adaptive methods, it is nonetheless of interest to investigate numerically whether the results previously obtained for kernel-based methods carry over to this setting. To this end, we propose to compare estimators $\lambdak{n}$ and $\lambdaks{n}$ when the invariant distributions ($\mu$ for $\lambdak{n}$ and $\mu^-$ for $\lambdaks{n}$) appearing in the numerator are estimated adaptively via projection. Estimators $\lambdaamgo{n}$ and $\lambdaamg{n}$ are excluded from this comparison, as implementing such an approach for them is substantially more technical and falls outside the scope of this paper. Indeed, both the numerator and the denominator would require an adaptive projection procedure, and the numerator would additionally involve a two-dimensional estimation problem.

We present below how to estimate by an adaptive projection method the invariant distribution $\mu$ appearing at the numerator of formula \eqref{eq:lambda:intro:k}. The estimator is evaluated from the first $n$ post-jump locations and is used to define an adaptive projection version of estimator $\lambdak{n}$. To this end, we follow the approach given in \cite{chagny:hal-02132884} to estimate a density function. The procedure is the same to estimate, from the first $n$ pre-jump locations, the distribution $\mu^-$ at the numerator of \eqref{eq:lambda:intro}, yielding an adaptive projection version of $\lambdaks{n}$.

Here, $\mu$ is assumed to be in $\LL^2_{[a,b]}$, where $a=0.05$ and $b=3$. As a consequence, one can decompose $\mu$ on the trigonometric basis,
$$\mu = \sum_{m\geq0}\alpha^{(m)}\varphi^{(m)},$$
making amenable to estimate $\mu$ as
$$\widehat{\mu}_n^{(M)} = \sum_{m=0}^M \widehat{\alpha}_n^{(m)}\varphi^{(m)},$$
with
$$\widehat{\alpha}_n^{(m)} = \frac{1}{n}\sum_{i=0}^{n-1}\varphi^{(m)}(Z_i).$$
The dimension of the model is then chosen by minimizing the penalized contrast,
$$M^\star_n = \argmin_{0\leq M\leq \overline{M}} \left\|\widehat{\mu}_n^{(M)}\right\|_{\LL^2_{[a,b]}}^2 - \frac{2}{n} \sum_{i=0}^{n-1} \widehat{\mu}_n^{(M)}(Z_i) + c\frac{M+1}{n},$$
with $\overline{M}=25$ and $c=1$. Figure~\ref{fig:simu:estim:errors:adaptive} presents the distribution of the pointwise error for adaptive projection versions of estimators $\lambdak{n}$ and $\lambdaks{n}$.

We first observe that the estimation error as a function of the state space follows the same general trend already highlighted for non-adaptive kernel estimators in Figure~\ref{fig:simu:estim:errors}. In particular, this trend can be seen from a theoretical perspective in Figure~\ref{fig:simu:sdclt} ($\kappa=0.4$), prior to normalization by the rate $\sqrt{nh_n}$ which is not relevant in adaptive projection estimation. In conclusion, the pattern does not appear to depend on the estimation technique used, at least in this example. A close comparison of Figures~\ref{fig:simu:estim:errors} and \ref{fig:simu:estim:errors:adaptive} shows that the order of magnitude of the pointwise error is the same for kernel estimation and for adaptive projection estimation. However, the simulations reveal that the behavior of the adaptive version of $\lambdak{n}$ is significantly better than that of the adaptive version of $\lambdaks{n}$, whereas their non-adaptive versions are very similar, both in the simulations and in the normalized theoretical variances (shown in Figure~\ref{fig:simu:norm:sdclt}).

\begin{figure}[ht]
    \centering
    $n = 1\,000$\\
    \includegraphics[height=4.4cm]{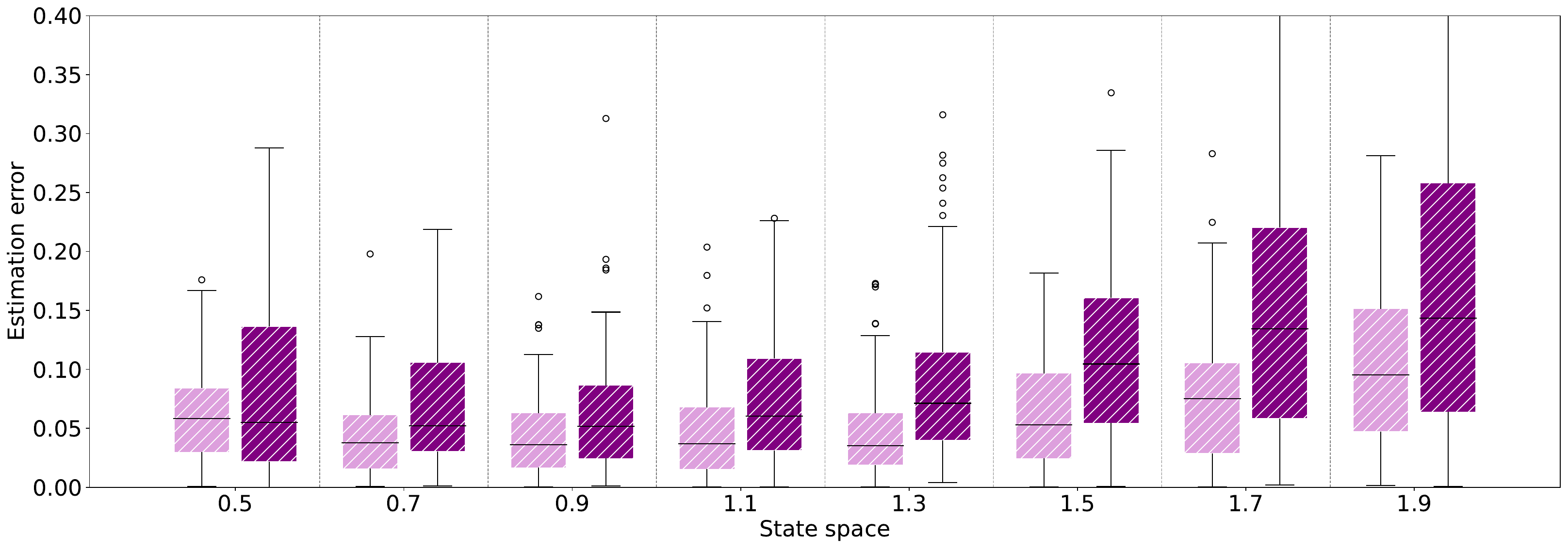}\vspace{0.5cm}\\
     $n = 10\,000$ \\
    \includegraphics[height=4.4cm]{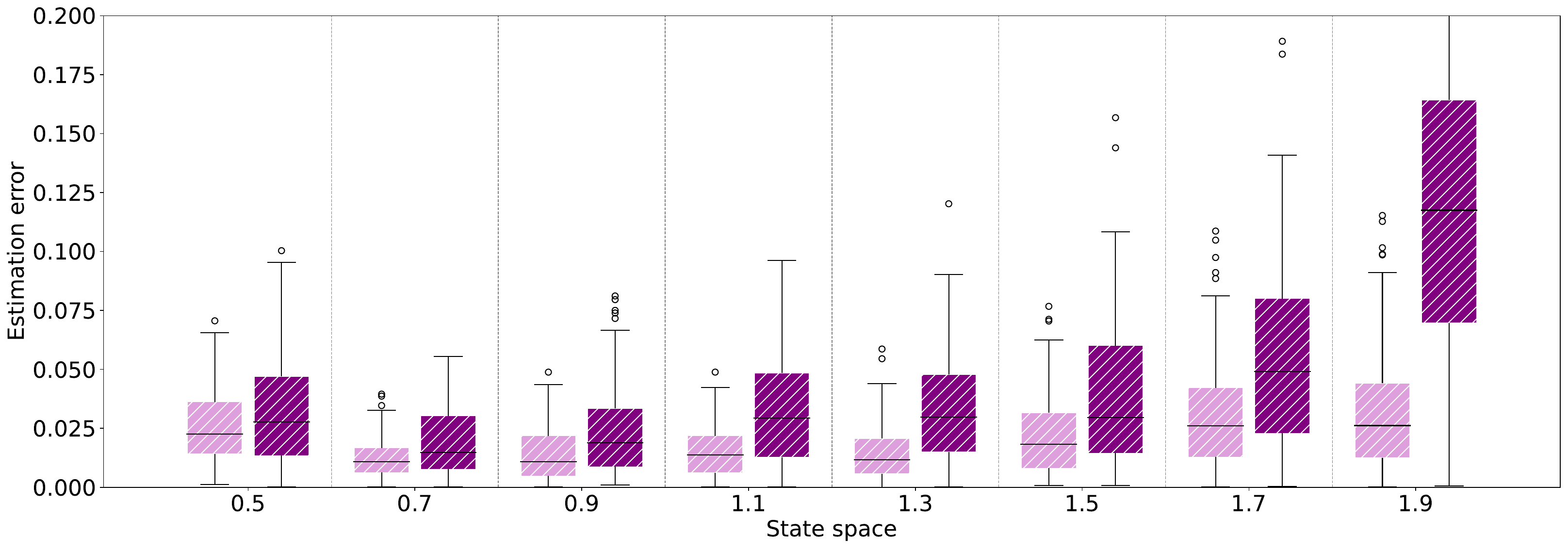}
    \caption{Pointwise estimation error over $100$ replicates for the adaptive projection versions of $\lambdak{n}(x)$ in pink and of $\lambdaks{n}(x)$ in purple, from left to right for each value of $x$ (with $x$ between $0.5$ and $1.9$ with a step of $0.2$), computed from a trajectory of size $n=1\,000$ (top) and of size $n=10\,000$ (bottom).}
    \label{fig:simu:estim:errors:adaptive}
\end{figure}

\section{Real data analysis}
\label{s:data}

\subsection{Context}
\label{ss:data:context}

The life cycle of a cell alternates phases of growth and division. This is a typical application of piecewise-deterministic models in dimension 1 (see for instance \cite{CDGMMY17,DHKR15,K24} and references therein): growth is considered exponential and divisions, assumed to be quasi-instantaneous, occur at random times. In addition, the division mechanism appears to be linked to cell size \cite{RHKARD14}, making the idea of a jump rate as a function of cell size relevant. In this context, the process under consideration models the size of a cell (and its progeny) over time: its flow is of the form $(x,t)\mapsto x\exp(\theta t)$, $\theta>0$, and its transition kernel can be written as $\delta_{\{x/2\}}$ (or more generally as any distribution with support $[0,x]$ to take into account both the intrinsic variability and an eventual bias in the division process). The jump rate, which is very difficult to parameterize in this kind of application, is the typical function of interest.

In this section, we propose to use single-cell data from \textit{Escherichia coli} \cite{TPPHBY17} to implement and compare, within the framework of the probabilistic model just described, the jump rate estimation strategies studied in this paper. The data in question are measurements of cell size, obtained by microscopy, under different temperature conditions (25°C, 27°C, or 37°C). For each, multiple independent data sets from different mother cells are available.

In Subsection~\ref{ss:data:37}, we describe the available data and explain how a piecewise-deterministic model is fitted to them. Summary statistics are provided in the main document only for the 37°C temperature condition; Appendix~\ref{app:fig} gathers corresponding figures for the other temperatures. The results of the jump rate estimation for the three temperature conditions are presented in Subsection~\ref{ss:data:jr}, with a numerical validation of the model developed in Subsection~\ref{ss:data:valid}.

\subsection{Data description and model fitting}
\label{ss:data:37}

Whatever the temperature condition, the data of interest are organized in different files corresponding to independent cell lineages. For each, cell size is measured every minute. In addition, a division indicator is available which, in a piecewise-deterministic model, precisely indicates the jump times. Consequently, the observation scheme chosen in the paper is exactly that of these data. Table~\ref{tab:data} provides information on the sample sizes available, while Figure~\ref{fig:data:37} shows some of the process statistics under a temperature condition of 37°C (see Figure~\ref{fig:data:25:27} in Appendix~\ref{app:fig} for temperatures of 25°C and 27°C). We can already see that the behavior of the process appears to be strongly temperature-dependent.

\begin{table}[ht]
    \centering
    \begin{tabular}{c|lccc}
     Temperature &       & Lineages & Measurements & Divisions \\ \hline 
    \multirow{ 2}{*}{25°C} & Total: & 65 & 307\,999 & 4\,485\\
    & Average (per lineage): & --- & 4738.44 & 68.67\\ \hline 
    \multirow{ 2}{*}{27°C} & Total: & 54 & 202\,086 & 3\,726\\
    & Average (per lineage): & --- & 3742.33 & 69 \\ \hline 
    \multirow{ 2}{*}{37°C} & Total: & 160 & 364\,920 & 11\,040\\
    & Average (per lineage): & --- & 2280.75 & 69
    \end{tabular}
    \caption{Summary of growth data under each of the temperature conditions.}
    \label{tab:data}
\end{table}

\begin{figure}[ht]
    \centering
    \begin{minipage}{0.6\textwidth}
        \centering
        \includegraphics[height=4.4cm]{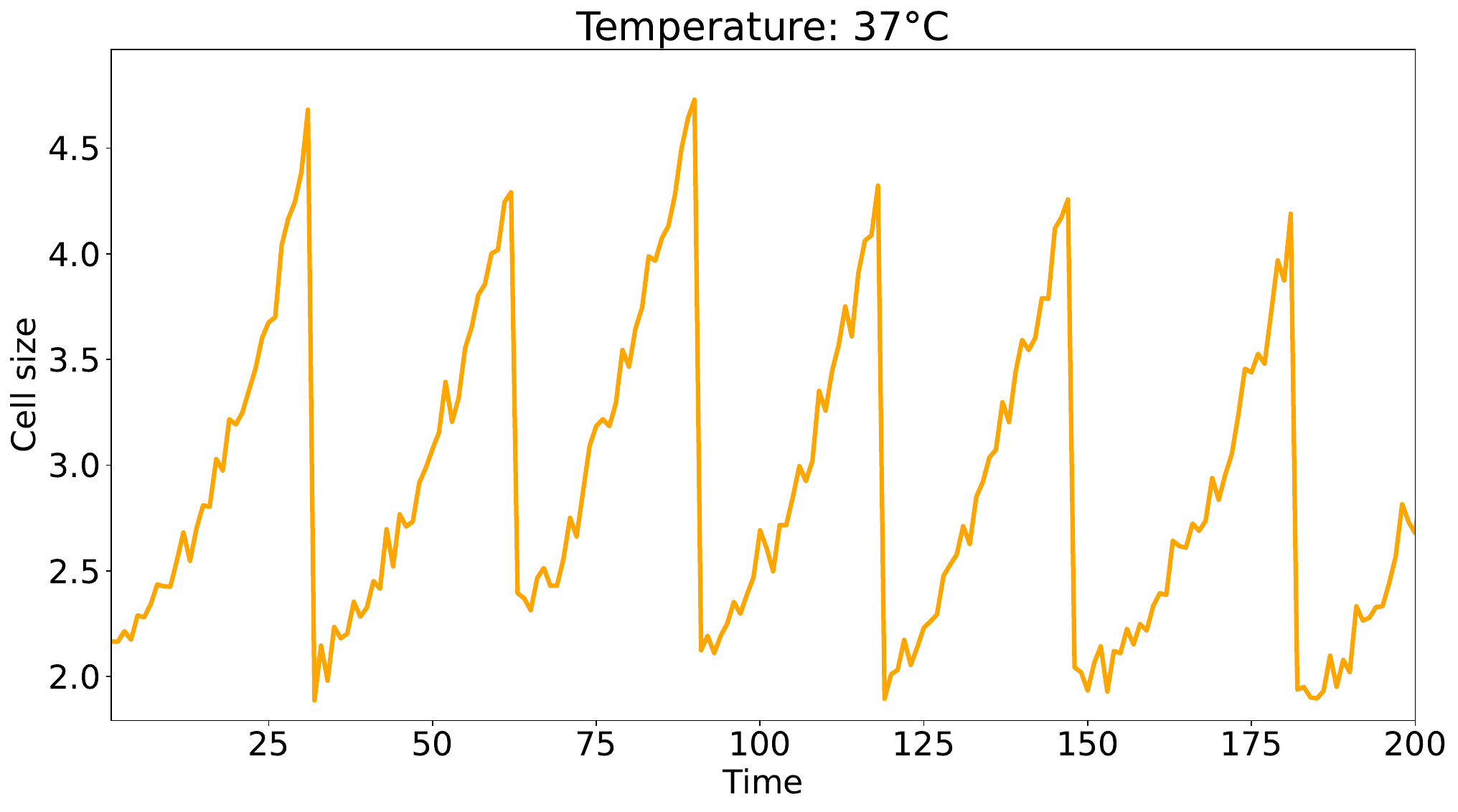}
    \end{minipage}%
    \begin{minipage}{0.35\textwidth}
        \centering
        \includegraphics[height=2.2cm]{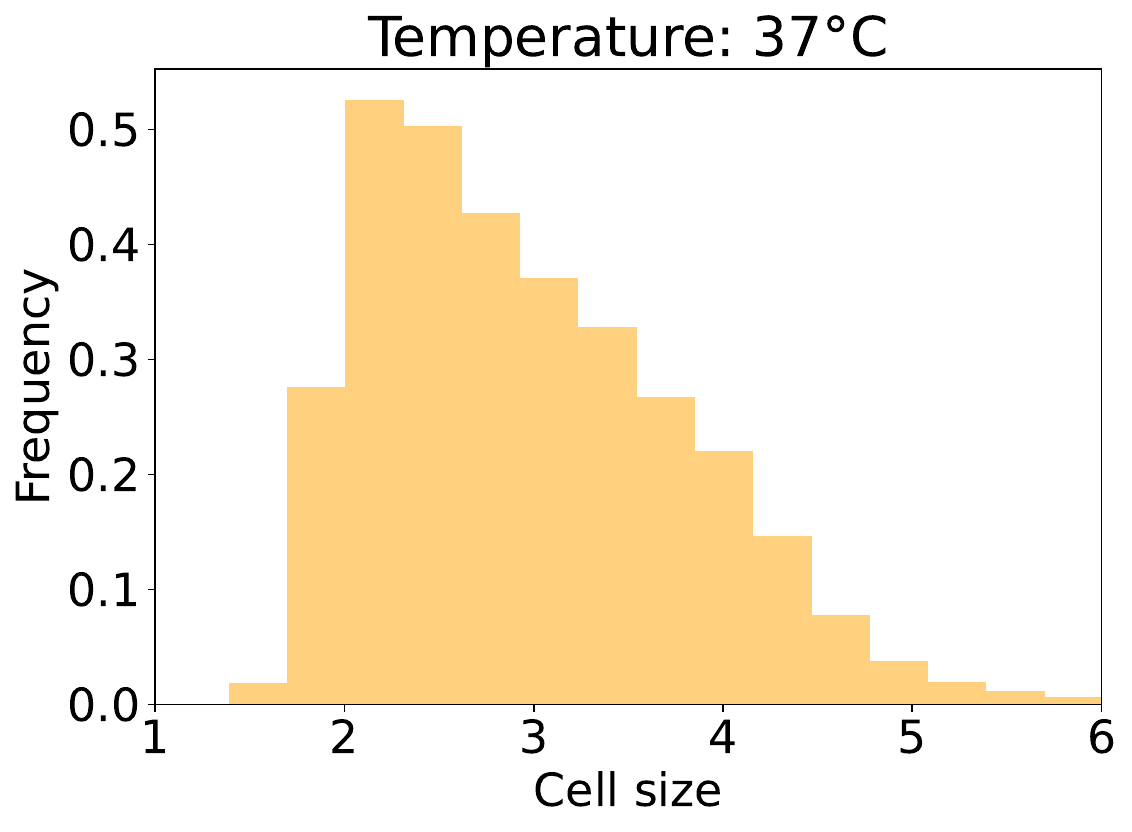} \\
        \includegraphics[height=2.2cm]{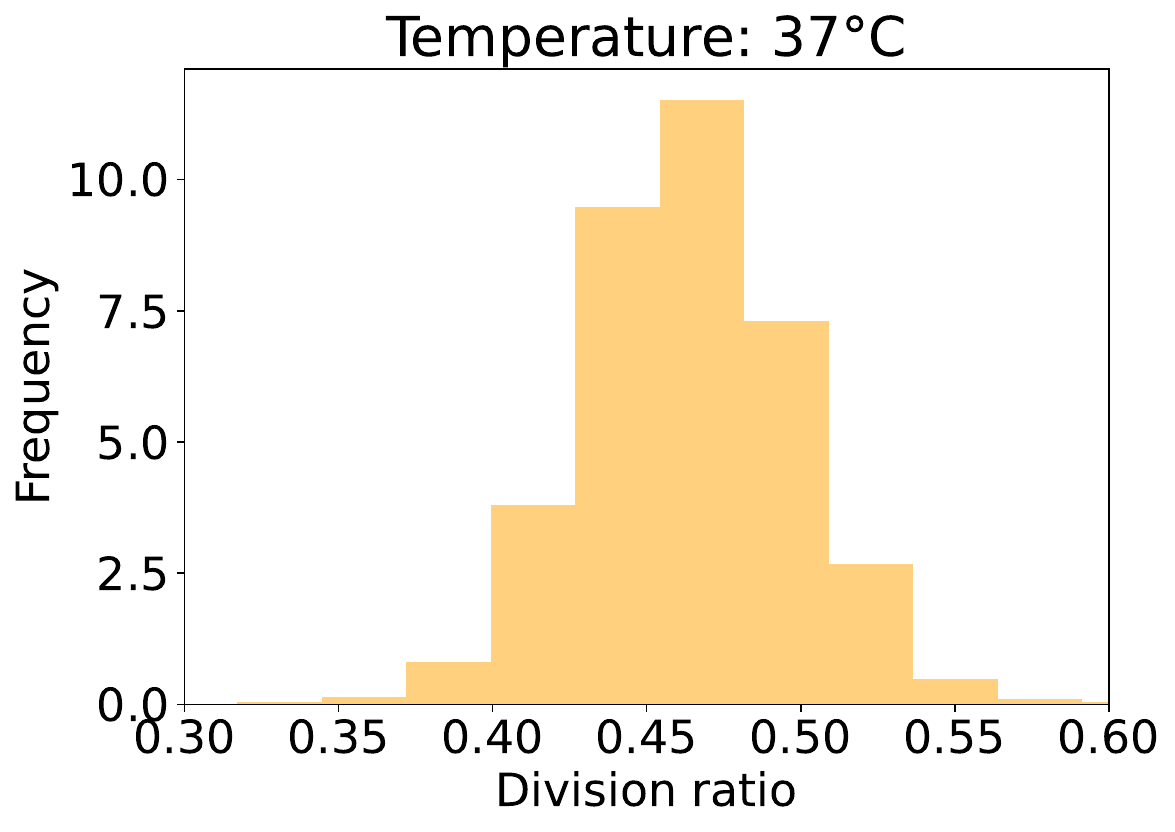}
    \end{minipage}
    
    \caption{Under a temperature condition of 37°C, cell size measurements (unit of size is micrometer) until time 200 (unit of time is minute) from lineage \texttt{xy01\_02} (left), distribution of cell size from $10\,000$ consecutive measurements (mixing different lineages) (top right), and distribution of division ratio (bottom right).}
    \label{fig:data:37}
\end{figure}

We now turn our attention to piecewise-deterministic Markov process modeling in order to apply the statistical procedures studied in this article. Looking at the empirical law of division ratios in Figures~\ref{fig:data:37} and \ref{fig:data:25:27}, we are convinced that, whatever the temperature experiment, the transition kernel can not be modeled by a Dirac mass, which rules out $\lambdak{n}$ that assumes deterministic fragmentations. That leaves us with the question of how to model the flow. For this, we assume exponential cell growth mentioned above, and model the logarithm of cell size by a linear function, i.e. $\Phi(t|x) = x+\theta t$. To obtain a one-dimensional model, $\theta$ must not depend on any quantity (except perhaps the temperature condition which is expected to play a significant role), especially on the cell in question. For each of the temperature condition, for each of the thousands of cells along the hundreds of lineages measured, we fit a linear model to their growth. Some results for a temperature of 37°C are given in Figure~\ref{fig:model:37} (see Figure~\ref{fig:model:25:27} in Appendix~\ref{app:fig} for temperatures of 25°C and 27°C). We observe that the histogram of estimated slopes depends on the temperature but is always unimodal with very low variance. We therefore accept the constant (but temperature-dependent) slope hypothesis and take for $\theta$ its mean value: $\theta = 0.012$ at 25°C, $\theta = 0.014$ at 27°C, and $\theta=0.025$ at 37°C. Figures~\ref{fig:model:37} and \ref{fig:model:25:27} also illustrate the quality of the piece\-wise-li\-near fit to the data. Interestingly, other types of models and experimental settings may also be considered, in particular those that track the entire cell genealogy, as in \cite{delyon:hal-01201923}, where the growth rate is not assumed to be constant but instead follows a Markov chain itself.

\begin{figure}[ht]
    \centering
    \begin{minipage}{0.55\textwidth}
        \centering
        \includegraphics[height=4.4cm]{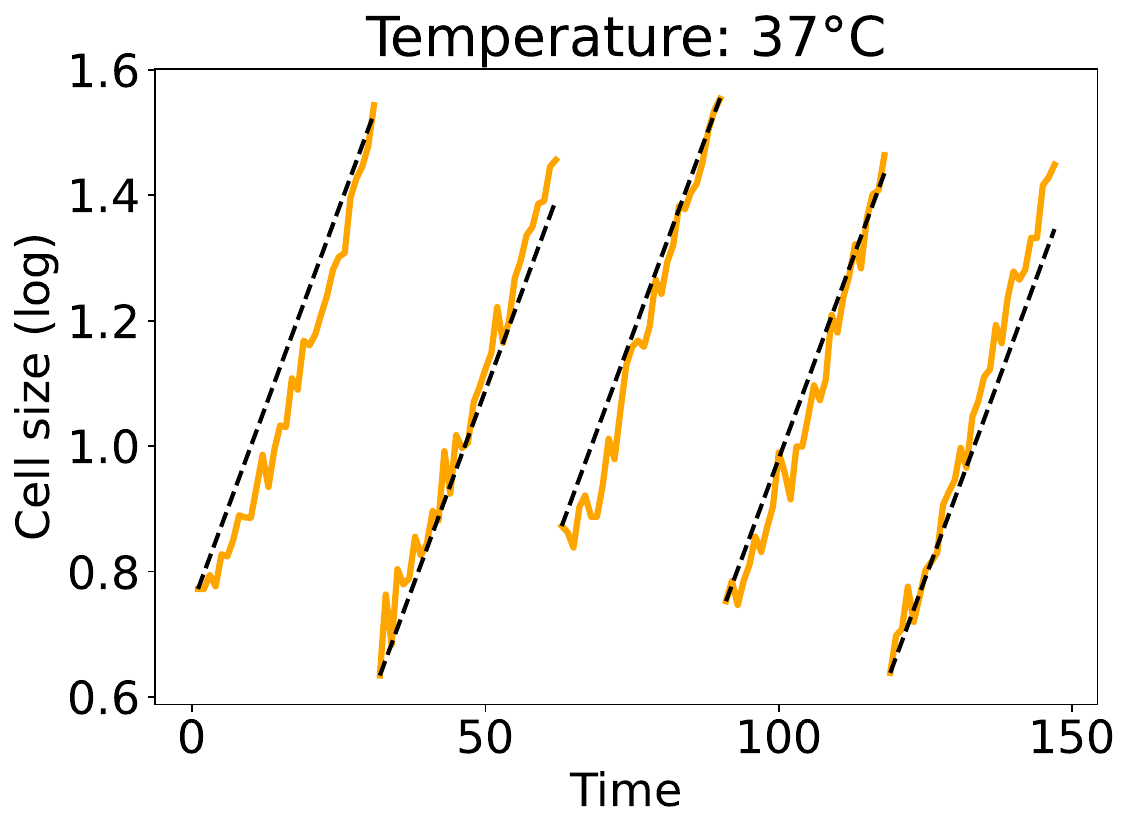}
    \end{minipage}%
    \begin{minipage}{0.2\textwidth}
        \centering
        \includegraphics[height=2cm]{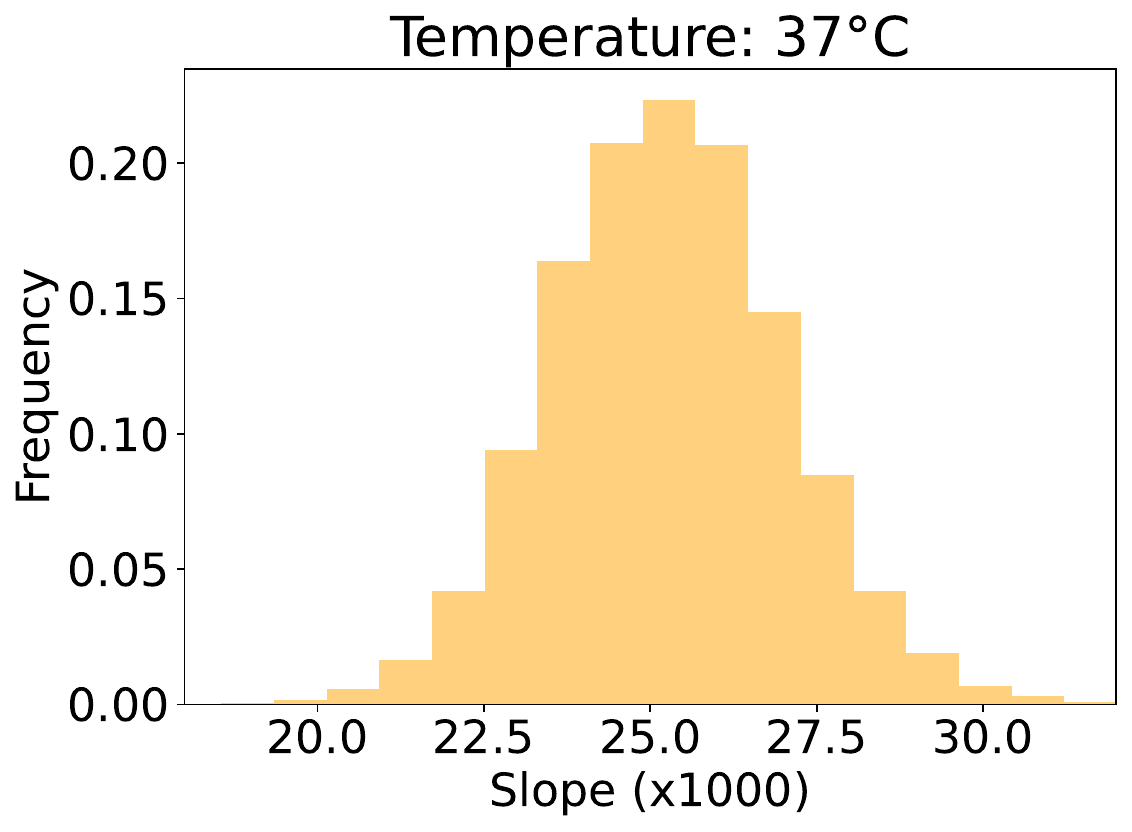} \\
        \includegraphics[height=2cm]{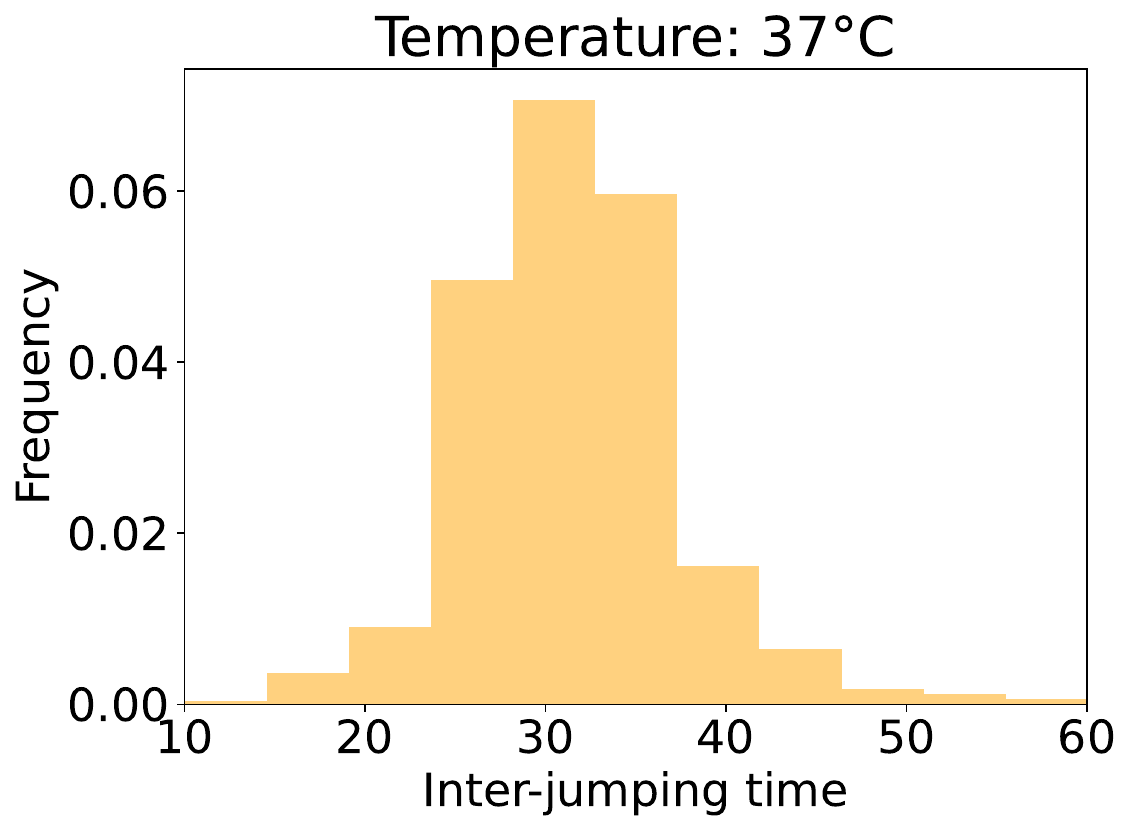}
    \end{minipage}
    \begin{minipage}{0.2\textwidth}
        \centering
        \includegraphics[height=2cm]{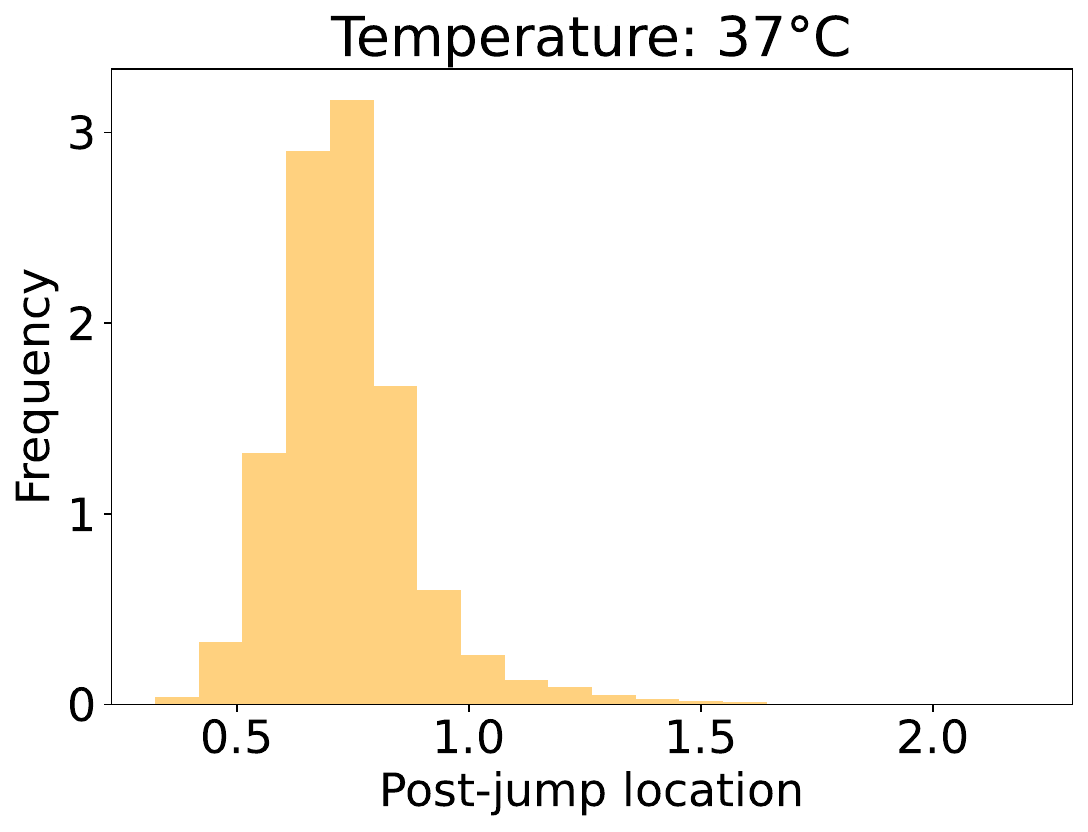} \\
        \includegraphics[height=2cm]{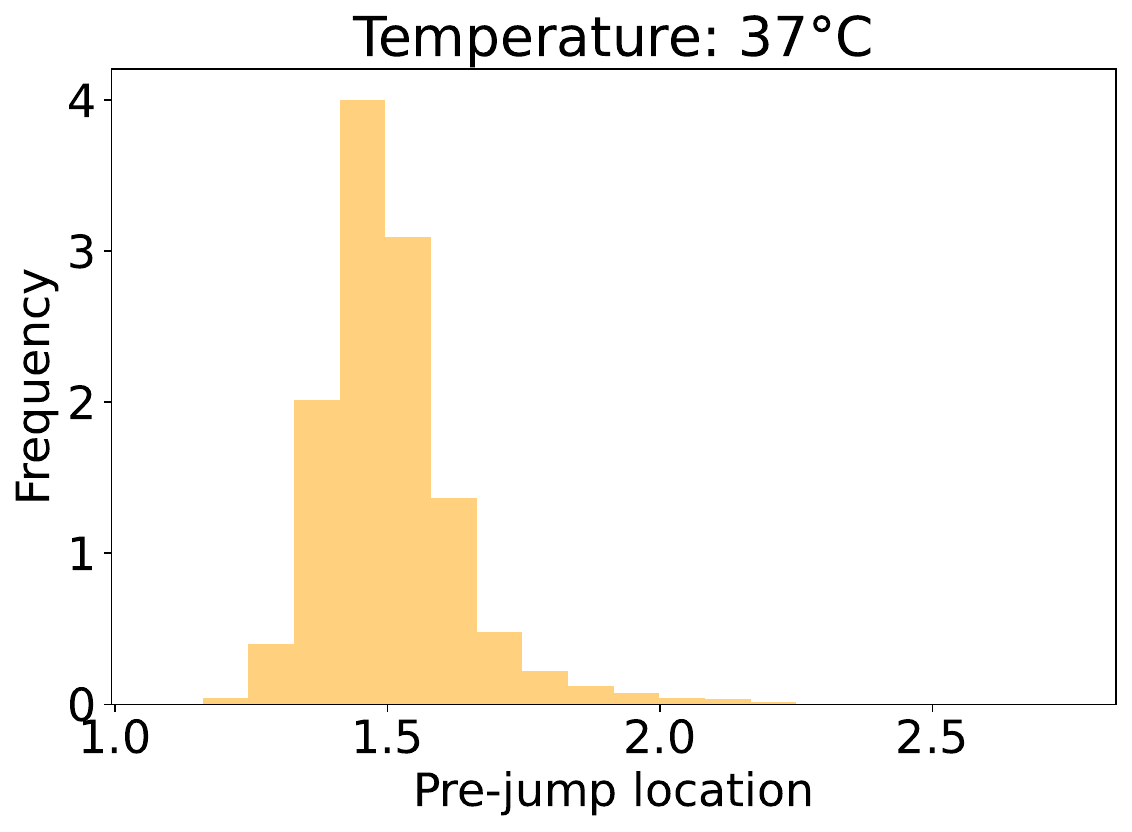}
    \end{minipage}
    \caption{Under a temperature condition of 37°C, logarithm of cell size measurements before the fifth division event and fitted linear growth (left), distribution of estimated slope (top center), distribution of time between two consecutive division events (bottom center), distribution of logarithm of cell size at division time (top right), and distribution of logarithm of cell size just before division (bottom right).}
    \label{fig:model:37}
\end{figure}

\subsection{Jump rate estimation}
\label{ss:data:jr}

Now that the model has been fitted, we can calculate the estimators of the jump rate. Estimator $\lambdaks{n}$ only requires the evaluation of $\Delta$ (the derivative of the flow), which here is simply a constant equal to $\theta$. To calculate estimator $\lambdaamg{n}$, we simplify the procedure from \cite{AM16}. As can be seen from Figures~\ref{fig:model:37} and \ref{fig:model:25:27}, the process of interest is rather stereotyped, with empirical distributions of the embedded Markov chain quite concentrated around their mode. For instance, inter-jumping times are mainly around $\tau=31.6$ on average at 37°C ($\tau = 66.6$ at 25°C and $\tau=52.4$ at 27°C). To estimate $\lambda(x)$, we therefore skip the complex step of optimal argument selection and evaluate the estimator at $(\xi,\tau)$ where $\xi = x - \theta\tau$ so that $\Phi(\tau|\xi) = x$.

The final step is to select the smoothing parameters. In the absence of ground truth, these are chosen by hand to avoid over-fitting (resulting in excessive oscillations) and under-fitting (no apparent variation). We also use our knowledge from numerical simulations: spatial bandwidths of the two estimators are very close. In addition, both bandwidths are expected to decrease in the sample size. Selected parameters are given in Table~\ref{tab:bandwidths}.

\begin{table}[ht]
    \centering
    \begin{tabular}{c|c|ccc}
     Temperature &  Sample size & $\vardiamondsuit$ bandwidth & $\spadesuit$ bandwidth & $\spadesuit$ bandwidth \\
     			 &			 &			  			  & (space)			  & (time) \\ \hline 
    25°C & 4\,485 & 0.05 & 0.06 & 4\\
    27°C & 3\,726 & 0.07 & 0.08 & 8\\
    37°C & 11\,040 & 0.02 & 0.03 & 3
    \end{tabular}   
    \caption{Summary of bandwidth parameters for $\lambdaks{n}$ and $\lambdaamg{n}$ (space and time) under each of the temperature conditions.}
    \label{tab:bandwidths}
\end{table}

To estimators $\lambdaks{n}$ and $\lambdaamg{n}$, we add the adaptive projection-based version of $\lambdaks{n}$ introduced in Subsection~\ref{ss:adaptive:simus}. Estimation results are given in Figure~\ref{fig:realdata:results}. To help with interpretation, we represent the function of cell size $x\mapsto\lambda(\exp(x))$. Both methods show that the behavior of the jump rate depends strongly on the temperature of the experiment, as expected. The non-adaptive kernel estimator $\lambdaks{n}$ and its adaptive projection version yield estimates based on \eqref{eq:lambda:intro} that are extremely close. By contrast, estimator $\lambdaamg{n}$ exhibits a qualitatively similar but quantitatively different behavior when examined in detail, e.g. growth in the jump rate starts earlier for $\lambdaamg{n}$ than for $\lambdaks{n}$, whatever the temperature condition. This supports the view that the method used to estimate the invariant law matters less than the formula employed to capture the function of interest. This analysis also highlighted the high sensitivity of $\lambdaamg{n}$ to smoothing parameters, particularly for the smallest samples, making it a more difficult method to calibrate.

\begin{figure}[ht]
    \centering
    \includegraphics[height=3.3cm]{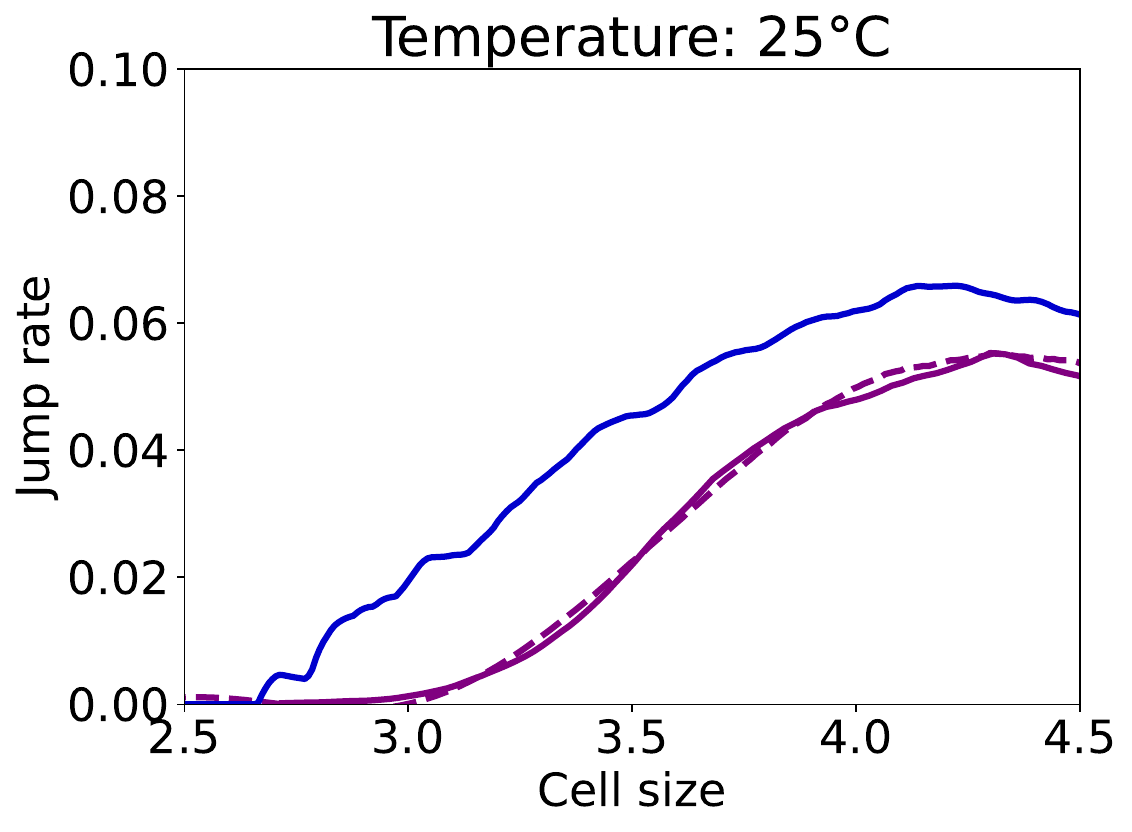}
    \includegraphics[height=3.3cm]{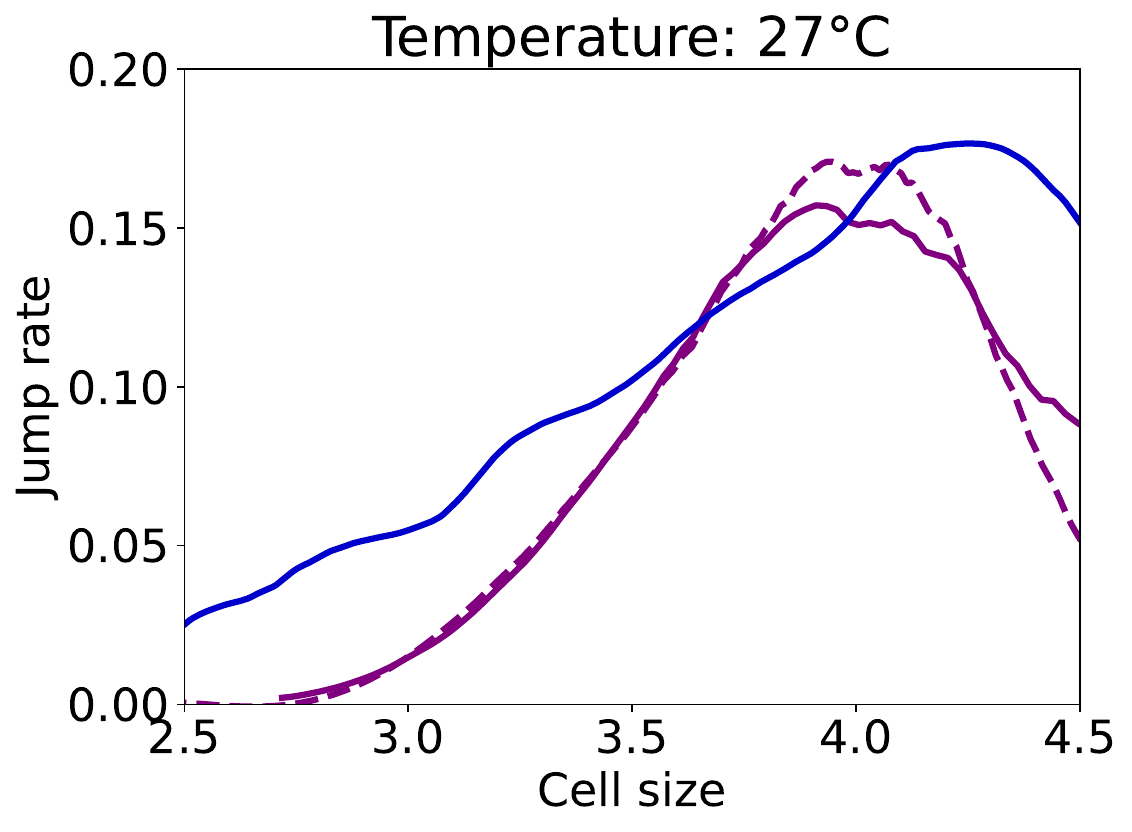}
    \includegraphics[height=3.3cm]{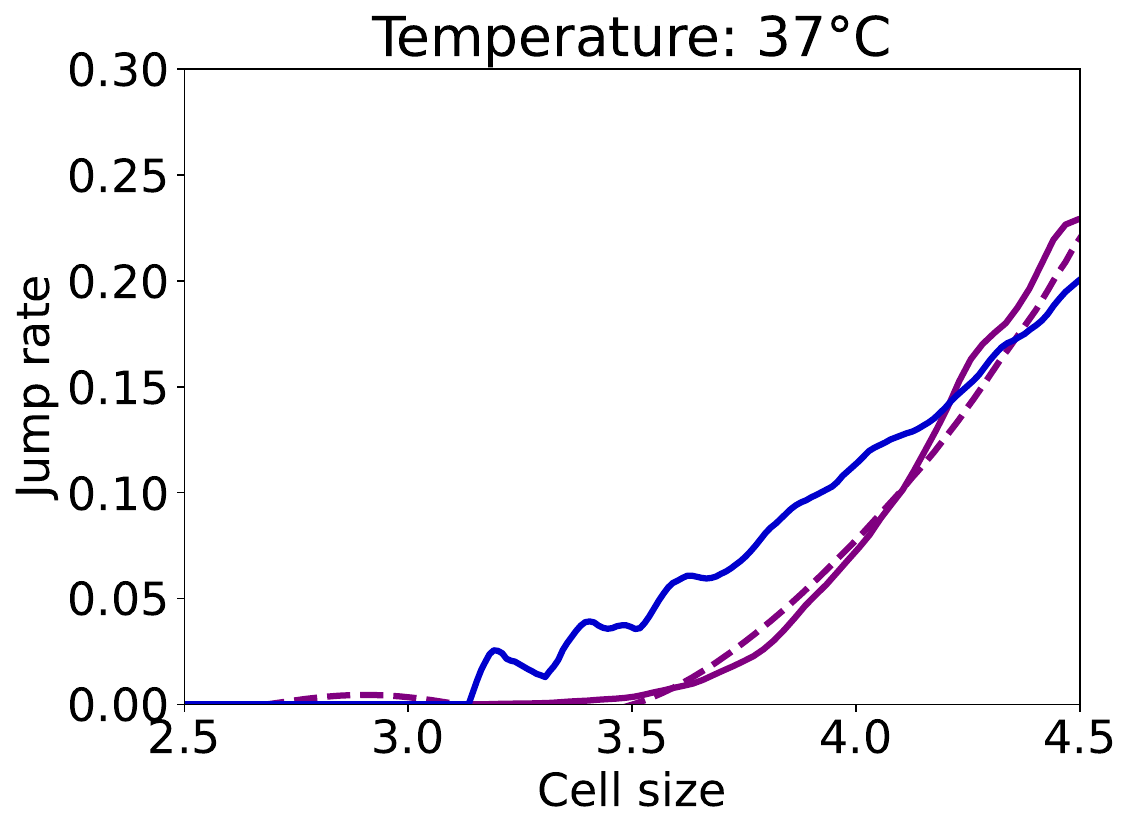}
    \caption{Estimated jump rates as a function of cell size using $\lambdaks{n}$ (purple, full  line), the adaptive projection-based version of $\lambdaks{n}$ (purple, dotted line) and $\lambdaamg{n}$ (dark blue) under different temperature conditions.}
    \label{fig:realdata:results}
\end{figure}

\subsection{A posteriori validation}\label{ss:data:valid}

The aim of this work is the comparison of estimation strategies of the jump rate in PDMPs. In the real data context of this section, we can not compare them on any theoretical criterion because there is no ground truth: the mechanism triggering cell division is unknown. To overcome this difficulty, we propose to compare our estimation results with the empirical data based on their distribution in continuous time rather than on the jump mechanism. In other words, we propose to return to the continuous-time invariant distribution of the model, computed using the estimated jump rates, which is directly comparable to the distribution of the raw data. It should be noted that the latter does not depend on the model fitting and can therefore play the role of theoretical reference.

For each temperature condition, we consider the PDMPs with the following local characteristics:
\begin{itemize}[wide=4.5pt,labelsep=4.5pt]
\item The flow $\Phi$ is defined by $ \Phi(t|x) = x+\theta t$ where the slope $\theta$ depends on the temperature condition and was estimated in Subsection \ref{ss:data:37}.
\item The fragmentations are defined by $Z_i=K_iZ_i^-$ where the $K_i$'s are independent and Gaussian with mean and variance estimated from division ratios under each of the temperature conditions (see Figures~\ref{fig:data:37} and \ref{fig:data:25:27}).
\item The jump rate is one of the estimates under consideration in Subsection~\ref{ss:data:jr}.
\end{itemize}

The continuous-time invariant distribution is difficult to evaluate exactly. We therefore propose to estimate it from simulated data generated according to the local characteristics that we have just specified.
To this end, for each temperature condition and for each jump rate estimate, we simulate the first $1\,000$ jumps of the process starting from $X_0=1.5$ on a minute-based discrete time grid (like in the real data sets). We then extract from the resulting trajectory the $10\,000$ last positions, which approximately follow the invariant distribution. We finally estimate this distribution with a classical kernel estimator. The comparison with raw data is given in Figure \ref{fig:realdata:valid}.

\begin{figure}[ht]
    \centering
    \includegraphics[height=3.3cm]{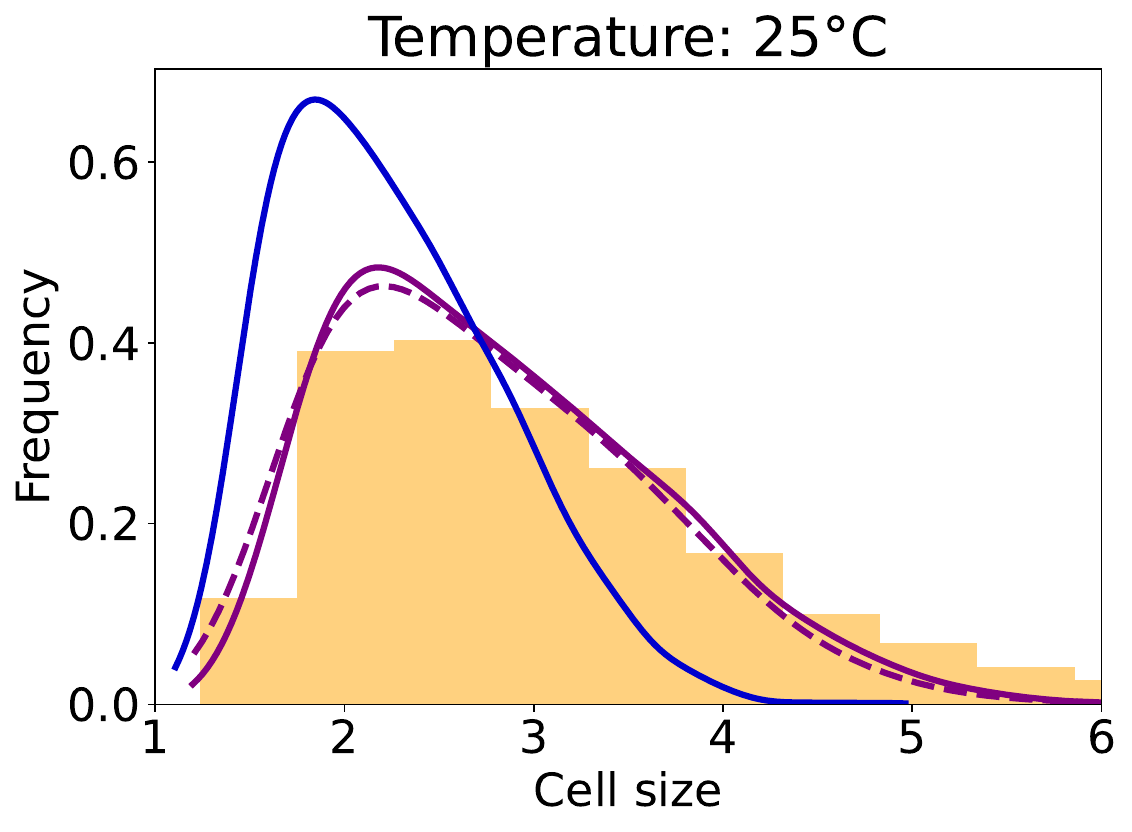}
    \includegraphics[height=3.3cm]{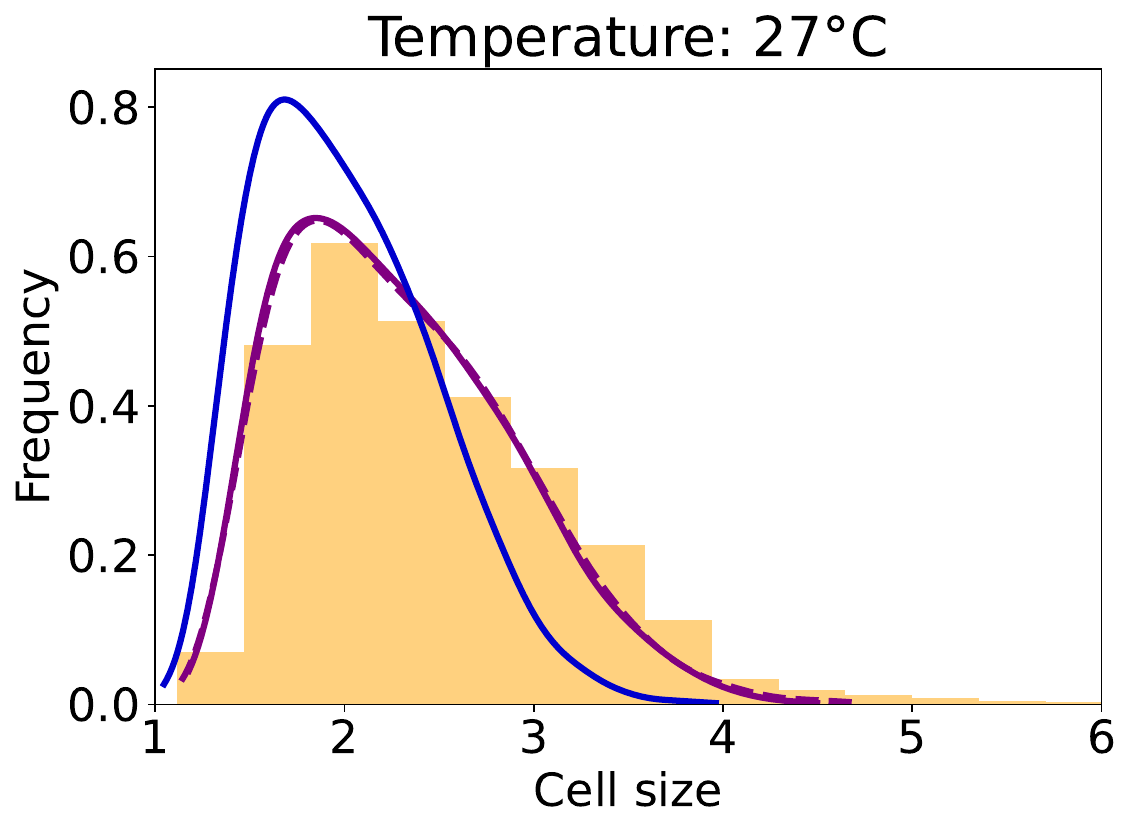}
    \includegraphics[height=3.3cm]{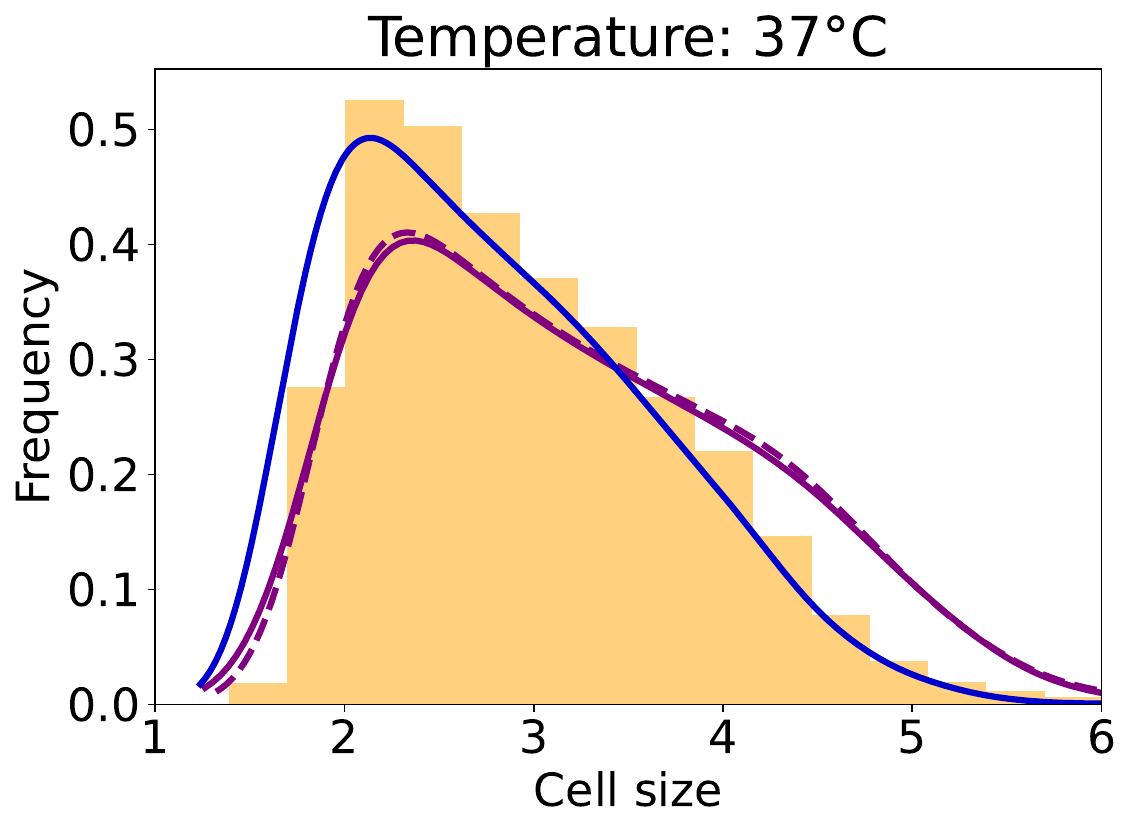}
    \caption{Cell size distribution from raw data (histogram) and continuous-time invariant distribution of the TCP model with jump rate $\lambdaks{n}$ (purple, full line), the adaptive projection-based version of $\lambdaks{n}$ (purple, dotted line) and $\lambdaamg{n}$ (dark blue), under different temperature conditions.}
    \label{fig:realdata:valid}
\end{figure}

Under temperature conditions of 25°C and 27°C, the invariant measure of the PDMPs with jump rate $\lambdaks{n}$ and its adaptive projection version fits well the raw data whereas the one with jump rate $\lambdaamg{n}$ is clearly disqualified: the law associated with $\lambdaamg{n}$ assigns greater weight to smaller cell sizes, as expected since the jump rate is higher than in the other two estimates (see Figure~\ref{fig:realdata:results}). The phenomenon is reversed at a temperature of 37°C, where the distribution associated with $\lambdaamg{n}$ is the closest to the raw data. As a consequence, neither of the two available estimation strategies consistently dominates the other. A notable difference that should be taken into account to explain this reversal lies in the number of observations available for the estimation, which is significantly higher at a temperature of 37 °C. Nevertheless, this fitting quality to the real cell size distribution in both cases suggests that the PDMP model is a reliable approximation of cell growth and division.

\section{Concluding remarks}
\label{s:conclusion}

Non-parametric estimation of the jump rate of PDMPs requires the development of specific methods. In dimension 1, the state of the art is limited to five main approaches \cite{ADGP14,AM16,F13,K16,KS21} that address the function of interest in different ways. In the first part of the article, we highlighted a connection between three of them \cite{F13,K16,KS21} that has not been noticed in the literature: they implicitly capture the jump rate from formula \eqref{eq:lambda:intro}. It is important to note that the statistical techniques and subsequent results in these papers are nevertheless different and complementary.

The objective of the present paper is not to compare the statistical approaches developed in the literature, but rather to compare the expressions of the jump rate that underlie these methods. To this end, we introduced two new kernel estimators of the jump rate $\lambda$. The estimator $\lambdak{n}$ is based on formula \eqref{eq:lambda:intro:k} (in the setting of deterministic fragmentations), similarly to \cite{K16}, whereas $\lambdaks{n}$ is based on formula \eqref{eq:lambda:intro}, in the same spirit as the estimators proposed in \cite{F13,KS21}. These estimators are to be compared with the kernel estimator $\lambdaamg{n}$ and its oracle version $\lambdaamgo{n}$, both based on formula \eqref{eq:lambdacircphi:intro} as in \cite{AM16}. While all these methods yield consistent estimators of the jump rate, the aim of the paper is to compare them at a higher order, through their convergence rate and their asymptotic variance. The approach based on the multiplicative intensity model \cite{ADGP14} was set aside, primarily because it does not provide a direct estimator of the jump rate. However, we believe that this method warrants further investigation and should be considered in a broader comparative analysis.

We have shown new results of consistency and asymptotic normality for $\lambdak{n}$ and $\lambdaks{n}$, making them comparable with $\lambdaamgo{n}$. The proof involves the study of vector martingales constructed along the embedded chain of the PDMP. The theoretical results prove that the three estimators can not in general be ordered based on their asymptotic variance, even within the same model. This result is surprising for two main reasons: (i) $\lambdak{n}$ is the only estimator that leverages the deterministic nature of the transitions, which does not guarantee that it achieves the lowest variance everywhere; (ii) the technique developed to construct $\lambdaamgo{n}$ in \cite{AM16} does not exploit the one-dimensional nature of the state space (as it was developed to operate in any dimension), which does not prevent it from achieving the lowest variance in certain regions of the state space.

We have demonstrated that these theoretical results align with numerical simulations in the context of the TCP model, a representative application of PDMPs in dimension 1. However, the numerical experiments also highlight the greater sensitivity of the estimators $\lambdaamgo{n}$ and $\lambdaamg{n}$ to the choice of smoothing parameters. Unlike the other two methods $\lambdak{n}$ and $\lambdaks{n}$ (which involve smoothing only in space), this approach requires smoothing in both space and time, as well as an optimal argument selection, making it the most complex technique to implement. The application to real data confirms this trend, underscoring the practical advantages of $\lambdaks{n}$. Furthermore, implementing the adaptive projection-based versions of $\lambdak{n}$ and $\lambdaks{n}$, which behave in a very similar way to their non-adaptive kernel versions, suggests that the technique used to estimate the jump rate matters less than the formula capturing it. As a conclusion, given the current state of knowledge, the estimator $\lambdaks{n}$ (or its variant) is likely the most suitable choice in dimension 1, despite exhibiting higher variance in certain cases.

The most promising solution is likely to involve aggregating the estimators studied in the present paper to leverage the strengths of each approach. For example, a compelling direction suggested by our analysis would be to select, at each point in the state space, the estimator with the lowest asymptotic variance.

Extending the approach developed in the present paper to multidimensional PDMPs would be relevant in many settings, given the variety of such models (see, for instance, \cite{CdSRSV25,GH25,GT12,HBEG17} and references therein).
The estimation strategies resulting from formulas \eqref{eq:lambda:intro} and \eqref{eq:lambda:intro:k} are both based on a time-space change of variable, which is naturally specific to dimension 1. In the literature in statistics, the only general nonparametric approach linking the jump rate of interest to the invariant distribution of multidimensional PDMPs relates to \eqref{eq:lambdacircphi:intro}. However, it might be possible to exploit the structure of some particular multidimensional models to establish a tractable formula for the jump rate, as it is done in one-dimensional growth-fragmentation processes \cite{K16,KS21}.

To the best of our knowledge, in the aforementioned multidimensional applications of PDMPs, no closed form of the invariant distribution is available in a general setting (with nonparametric jump rate). This is not surprising since the study of the invariant distribution in multidimensional models is in general a difficult problem. For instance, establishing the uniqueness of the invariant distribution or understanding its qualitative properties appear to be challenging questions, even in specific frameworks (see \cite{BH12, BLBMZ15} and references therein). 

Nevertheless, one-dimensional PDMPs might help to better understand some particular multidimensional processes. As an illustration, consider the class of two-dimensional PDMPs where the flow is exponential, i.e. $\boldsymbol{\dot x}=\alpha \boldsymbol{x}$ with $\alpha>0$, the jump rate depends only on the Euclidean distance to the origin, i.e. $\lambda(x)=l(\lVert x \rVert)$, and the jump mechanism is characterized by an instantaneous contraction of the distance to the origin and a new random direction. The polar form of the process is also a PDMP with independent coordinates which radial component is a growth-fragmentation process. Using one-dimensional formulas \eqref{eq:lambda:intro} and \eqref{eq:lambda:intro:k} on the radial component, one can derive two estimation formulas for the jump rate of interest $\lambda$. Even though this is a hypothetical example, it reinforces the interest of studying one-dimensional processes, in addition to their numerous applications in biology (cell growth-division models developed in \cite{DHKR15,K24,RHKARD14} or the cancer model studied in \cite{Cleynen2025}), in insurance (see for instance \cite{KP11}) or in reliability theory (see \cite{doi:10.1177/1748006X16651170,SDZE12} and references therein).

\bibliographystyle{acm}
\bibliography{biblio.bib}

@unpublished{K24,
author = {Nathalie Krell},
title = {Branching processes and bacterial growth},
year = {2024},
note = {arXiv:2409.03317},
}

@article{TPPHBY17,
author = {Tanouchi, Yu and Pai, Anand and Park, Heungwon and Huang, Shuqiang and Buchler, Nicolas E. and You, Lingchong},
year = "2017",
title = {Long-term growth data of {E}scherichia coli at a single-cell level},
journal = {Scientific Data},
page = {170036},
volume = {4},
issue = {1},
}

@article{RHKARD14,
author = {Robert, Lydia and Hoffmann, Marc and Krell, Nathalie and Aymerich, St\'ephane and Robert, J\'er\^ome and Doumic, Marie},
year = {2014},
title = {Division in {E}scherichia coli is triggered by a size-sensing rather than a timing mechanism},
journal = {BMC Biology},
page = {17},
volume = {12},
issue = {1},
}

@article{A14,
  author = {Aza\"{\i}s, Romain},
  title = {A recursive nonparametric estimator for the transition kernel of a piecewise-deterministic {Markov} process},
  journal = {ESAIM: PS},
  year = "2014",
  volume = "18",
  pages = "726-749",
}

@article{AM16,
  author={Aza\"{i}s, Romain and Muller-Gueudin, Aur\'elie},
  fjournal = "Electronic Journal of Statistics",
  journal = "Electron. J. Statist.",
  number = "2",
  pages = "3648--3692",
  publisher = "The Institute of Mathematical Statistics and the Bernoulli Society",
  title = "Optimal choice among a class of nonparametric estimators of the jump rate for piecewise-deterministic {M}arkov processes",
  volume = "10",
  year = "2016"
}

@article{ADGP14,
author = {Aza\"{\i}s, Romain and Dufour, Fran\c{c}ois and G\'egout-Petit, Anne},
title = {Non-Parametric Estimation of the Conditional Distribution of the Interjumping Times for Piecewise-Deterministic {Markov} Processes},
journal = {Scandinavian Journal of Statistics},
volume = {41},
number = {4},
pages = {950-969},
doi = {https://doi.org/10.1111/sjos.12076},
url = {https://onlinelibrary.wiley.com/doi/abs/10.1111/sjos.12076},
eprint = {https://onlinelibrary.wiley.com/doi/pdf/10.1111/sjos.12076},
year = {2014}
}

@book{D97, author = {Duflo, Marie and Wilson, Stephen S.}, title = {Random Iterative Models}, year = {1997}, isbn = {3540571000}, publisher = {Springer-Verlag}, address = {Berlin, Heidelberg}, edition = {1st} }

@book{MT09, place={Cambridge}, edition={2}, series={Cambridge Mathematical Library}, title={Markov Chains and Stochastic Stability}, publisher={Cambridge University Press}, author={Meyn, Sean and Tweedie, Richard L. and Glynn, Peter W.}, year={2009}, collection={Cambridge Mathematical Library}}

@article{K16,
  author = {Krell, Nathalie},
  title = {Statistical estimation of jump rates for a piecewise
          deterministic {Markov} processes with deterministic increasing motion and jump
            mechanism},
  journal = {ESAIM: Probability and Statistics},
  year = 2016,
  volume = 20,
  pages = "196-216",
}

@article{KS21,
author = {Nathalie Krell and {\'E}meline Schmisser},
title = {{Nonparametric estimation of jump rates for a specific class of piecewise deterministic {Markov} processes}},
volume = {27},
journal = {Bernoulli},
number = {4},
publisher = {Bernoulli Society for Mathematical Statistics and Probability},
pages = {2362 -- 2388},
year = {2021},
}

@article{DHKR15,
author = {Marie Doumic and Marc Hoffmann and Nathalie Krell and Lydia Robert},
title = {Statistical estimation of a growth-fragmentation model observed on a genealogical tree},
volume = {21},
journal = {Bernoulli},
number = {3},
publisher = {Bernoulli Society for Mathematical Statistics and Probability},
pages = {1760 -- 1799},
year = {2015},
}

@article{F13, title={Nonparametric Estimation for a Class of Piecewise-Deterministic {M}arkov Processes}, volume={50}, number={4}, journal={Journal of Applied Probability}, author={Fujii, Takayuki}, year={2013}, pages={931–942}}

@article{D84,
 ISSN = {00359246},
 author = {Davis, M.H.A.},
 journal = {Journal of the Royal Statistical Society. Series B (Methodological)},
 number = {3},
 pages = {353--388},
 publisher = {[Royal Statistical Society, Wiley]},
 title = {Piecewise-Deterministic {Markov} Processes: A General Class of Non-Diffusion Stochastic Models},
 volume = {46},
 year = {1984}
}

@article{DGR02,
  title = {A {M}arkovian analysis of additive-increase multiplicative-decrease algorithms},
  volume = {34},
  rights = {https://www.cambridge.org/core/terms},
  issn = {0001-8678, 1475-6064},
  url = {https://www.cambridge.org/core/product/identifier/S000186780001140X/type/journal_article},
  doi = {10.1239/aap/1019160951},
  pages = {85--111},
  number = {1},
  journaltitle = {Advances in Applied Probability},
    journal = {Advances in Applied Probability},
  shortjournal = {Advances in Applied Probability},
  author = {Dumas, Vincent and Guillemin, Fabrice and Robert, Philippe},
  urldate = {2024-09-02},
  date = {2002-03},
    year={2002},
  langid = {english},
}

@article {A78,
    AUTHOR = {Aalen, Odd O.},
     TITLE = {Nonparametric inference for a family of counting processes},
   JOURNAL = {Ann. Statist.},
  FJOURNAL = {The Annals of Statistics},
    VOLUME = {6},
      YEAR = {1978},
    NUMBER = {4},
     PAGES = {701--726},
      ISSN = {0090-5364},
     CODEN = {ASTSC7},
   MRCLASS = {62M09 (62B05 62P10)},
  MRNUMBER = {491547 (80a:62137)},
MRREVIEWER = {Rashid Ahmad},
       URL =
              {http://links.jstor.org/sici?sici=0090-5364(197807)6:4<701:NIFAFO>2.0.CO;2-U&origin=MSN},
}

@book {ABGK93,
    AUTHOR = {Andersen, Per Kragh and Borgan, {\O}rnulf and Gill, Richard D.
              and Keiding, Niels},
     TITLE = {Statistical models based on counting processes},
    SERIES = {Springer Series in Statistics},
 PUBLISHER = {Springer-Verlag},
   ADDRESS = {New York},
      YEAR = {1993},
     PAGES = {xii+767},
      ISBN = {0-387-97872-0},
   MRCLASS = {60G55 (60G44 62G05 62P10)},
  MRNUMBER = {1198884 (94c:60079)},
MRREVIEWER = {Jack Cuzick},
}

@article{BL08, title={On level crossings for a general class of piecewise-deterministic {M}arkov processes}, volume={40}, DOI={10.1239/aap/1222868187}, number={3}, journal={Advances in Applied Probability}, author={Borovkov, K. and Last, G.}, year={2008}, pages={815–834}}

@book {D93,
    AUTHOR = {Davis, Mark H. A.},
     TITLE = {Markov models and optimization},
    SERIES = {Monographs on Statistics and Applied Probability},
    VOLUME = {49},
 PUBLISHER = {Chapman \& Hall},
   ADDRESS = {London},
      YEAR = {1993},
     PAGES = {xiv+295},
      ISBN = {0-412-31410-X},
   MRCLASS = {90-02 (49L20 60J75 90C15 93E20)},
  MRNUMBER = {1283589 (96b:90002)},
MRREVIEWER = {O. L. V. Costa},
}

@article{CD08,
author = {Costa, O. L. V. and Dufour, F.},
title = {Stability and Ergodicity of Piecewise Deterministic {Markov} Processes},
journal = {SIAM Journal on Control and Optimization},
volume = {47},
number = {2},
pages = {1053-1077},
year = {2008},
doi = {10.1137/060670109},
URL = {https://doi.org/10.1137/060670109},
eprint = {https://doi.org/10.1137/060670109}
}

@article{SDZE12,
  author    = {de Saporta, B. and Dufour, F. and Zhang, H. and Elegbede, C.},
  title     = {Optimal stopping for the predictive maintenance of a structure subject to corrosion},
  journal   = {Proceedings of the Institution of Mechanical Engineers, Part O: Journal of Risk and Reliability},
  year      = {2012},
  volume    = {226},
  number    = {2},
  pages     = {169--181},
  doi       = {10.1177/1748006X11413681}
}

@article{CDGMMY17,
	author = {Cloez, Bertrand and Dessalles, Renaud and Genadot, Alexandre and Malrieu, Florent and Marguet, Aline and Yvinec, Romain},
	title = {Probabilistic and Piecewise Deterministic models in Biology},
	DOI= "10.1051/proc/201760225",
	url= "https://doi.org/10.1051/proc/201760225",
	journal = {ESAIM: Procs},
	year = 2017,
	volume = 60,
	pages = "225-245",
}

@article{KP11,
volume = {78},
number = {4},
month = {December},
title = {Does insurance help to escape the poverty trap? - A ruin theoretic approach},
publisher = {Wiley},
year = {2011},
journal = {Journal of Risk and Insurance},
doi = {10.1111/j.1539-6975.2010.01396.x},
pages = {1003--1028},
issn = {0022-4367},
author = {Kovacevic, R. M. and Pflug, G. C.}
}

@article{CMP10,
title = {On the long time behavior of the {TCP} window size process},
journal = {Stochastic Processes and their Applications},
volume = {120},
number = {8},
pages = {1518-1534},
year = {2010},
issn = {0304-4149},
doi = {https://doi.org/10.1016/j.spa.2010.03.019},
url = {https://www.sciencedirect.com/science/article/pii/S0304414910001055},
author = {Djalil Chafa\"{i} and Florent Malrieu and Katy Paroux}
}

@article{ADGP13,
author = {Romain Aza{\"i}s and Fran{\c{c}}ois Dufour and Anne G{\'e}gout-Petit},
title = {{Nonparametric estimation of the jump rate for non-homogeneous marked renewal processes}},
volume = {49},
journal = {Annales de l'Institut Henri Poincaré, Probabilités et Statistiques},
number = {4},
publisher = {Institut Henri Poincaré},
pages = {1204 -- 1231},
year = {2013},
doi = {10.1214/12-AIHP503},
URL = {https://doi.org/10.1214/12-AIHP503}
}

@misc{Cleynen2025,
  author       = {Alice Cleynen and Beno{\^\i}te de Saporta and Am\'elie Vernay},
  title        = {Estimating relapse time distribution from longitudinal biomarker trajectories using iterative regression and continuous time {M}arkov processes},
  howpublished = {arXiv:2503.10448},
  year         = {2025},
  url          = {https://arxiv.org/abs/2503.10448},
}

@article{doi:10.1177/1748006X16651170,
author = {Anis Ben Abdessalem and Romain Aza\"{\i}s and Marie Touzet-Cortina and Anne Gégout-Petit and Monique Puiggali},
title ={Stochastic modelling and prediction of fatigue crack propagation using piecewise-deterministic Markov processes},
journal = {Proceedings of the Institution of Mechanical Engineers, Part O: Journal of Risk and Reliability},
volume = {230},
number = {4},
pages = {405-416},
year = {2016},
doi = {10.1177/1748006X16651170},
URL = {https://doi.org/10.1177/1748006X16651170},
eprint = {https://doi.org/10.1177/1748006X16651170},
}

@article{BardetChristenGuillinMalrieuZitt2013,
  author       = {Jean-Baptiste Bardet and Alejandra Christen and Arnaud Guillin and Florent Malrieu and Pierre-André Zitt},
  title        = {Total variation estimates for the {TCP} process},
  journal      = {Electronic Journal of Probability},
  volume       = {18},
  number       = {1720},
  pages        = {1--21},
  year         = {2013},
  doi          = {10.1214/EJP.v18-1720},
  url          = {https://doi.org/10.1214/EJP.v18-1720}
}

@article{delyon:hal-01201923,
  TITLE = {{Investigation of asymmetry in E. coli growth rate}},
  AUTHOR = {Delyon, Bernard and de Saporta, Beno{\^i}te and Krell, Nathalie and Robert, Lydia},
  URL = {https://inria.hal.science/hal-01201923},
  JOURNAL = {{Case Studies in Business, Industry and Government Statistics}},
  PUBLISHER = {{Soci{\'e}t{\'e} Fran{\c c}aise de Statistique}},
  VOLUME = {7},
  NUMBER = {1},
  PAGES = {1-13},
  YEAR = {2018},
  HAL_ID = {hal-01201923},
  HAL_VERSION = {v1},
}

@article{chagny:hal-02132884,
  TITLE = {An introduction to nonparametric adaptive estimation},
  AUTHOR = {Chagny, Ga{\"e}lle},
  URL = {https://hal.science/hal-02132884},
  JOURNAL = {{The Graduate Journal of Mathematics}},
  PUBLISHER = {{Mediterranean Institute for the Mathematical Sciences (MIMS)}},
  VOLUME = {2016},
  NUMBER = {2},
  PAGES = {105-120},
  YEAR = {2016},
  MONTH = Dec,
  HAL_ID = {hal-02132884},
  HAL_VERSION = {v1},
}

@article{RamlauHansen1983,
  author  = {Ramlau-Hansen, Henrik},
  title   = {Smoothing Counting Process Intensities by Means of Kernel Functions},
  journal = {The Annals of Statistics},
  year    = {1983},
  volume  = {11},
  number  = {2},
  pages   = {453--466},
  month   = {June},
  doi     = {10.1214/aos/1176346152}
}

@article{MS21,
 ISSN = {00219002, 14756072},
 URL = {https://www.jstor.org/stable/48656626},
 author = {Meitz, Mika and Saikkonen, Pentti},
 journal = {Journal of Applied Probability},
 number = {3},
 pages = {pp. 594--608},
 publisher = {Applied Probability Trust},
 title = {SUBGEOMETRIC ERGODICITY AND $\beta$-MIXING},
 urldate = {2026-07-24},
 volume = {58},
 year = {2021}
}

@inbook{GH25,
   title={Efficient Stochastic Simulation of Gene Regulatory Networks Using Hybrid Models of Transcriptional Bursting},
   ISBN={9783032014368},
   ISSN={1611-3349},
   url={http://dx.doi.org/10.1007/978-3-032-01436-8_7},
   DOI={10.1007/978-3-032-01436-8_7},
   booktitle={Computational Methods in Systems Biology},
   publisher={Springer Nature Switzerland},
   author={Gaillard, Mathilde and Herbach, Ulysse},
   year={2025},
   month=Aug, pages={109–125}
   }

@article{HBEG17,
  author = {Herbach, Ulysse and Bonnaffoux, Arnaud and Espinasse, Thibault and Gandrillon, Olivier},
  date = {2017/11/21},
  doi = {10.1186/s12918-017-0487-0},
  id = {Herbach2017},
  isbn = {1752-0509},
  journal = {BMC Systems Biology},
  number = {1},
  pages = {105},
  title = {Inferring gene regulatory networks from single-cell data: a mechanistic approach},
  url = {https://doi.org/10.1186/s12918-017-0487-0},
  volume = {11},
  year = {2017}
  }

@misc{CdSRSV25,
      title={Bridging Impulse Control of Piecewise Deterministic {Markov} Processes and {Markov} Decision Processes: Frameworks, Extensions, and Open Challenges}, 
      author={Cleynen, Alice and de Saporta, Beno{\^i}te and Rossini, Orlane and Sabbadin, R\'egis and Vernay, Am\'elie},
      year={2025},
      eprint={2501.04120},
      archivePrefix={arXiv},
      primaryClass={stat.ME},
      url={https://arxiv.org/abs/2501.04120}, 
}

@article{BH12,
doi = {10.1088/0951-7715/25/10/2937},
url = {https://doi.org/10.1088/0951-7715/25/10/2937},
year = {2012},
month = {sep},
publisher = {IOP Publishing},
volume = {25},
number = {10},
pages = {2937},
author = {Bakhtin, Yuri and Hurth, Tobias},
title = {Invariant densities for dynamical systems with random switching},
journal = {Nonlinearity}
}

@article{BLBMZ15,
     author = {Bena{\"\i}m, Michel and Le Borgne, St\'ephane and Malrieu, Florent and Zitt, Pierre-Andr\'e},
     title = {Qualitative properties of certain piecewise deterministic {Markov} processes},
     journal = {Annales de l'I.H.P. Probabilit\'es et statistiques},
     pages = {1040--1075},
     year = {2015},
     publisher = {Gauthier-Villars},
     volume = {51},
     number = {3},
     doi = {10.1214/14-AIHP619},
     mrnumber = {3365972},
     zbl = {1325.60123},
     language = {en},
     url = {https://www.numdam.org/articles/10.1214/14-AIHP619/}
}

@article{GT12,
author = {Genadot, Alexandre and Thieullen, Mich{\`e}le},
title = {{Averaging for a fully coupled piecewise-deterministic {Markov} process in infinite dimensions}},
volume = {44},
journal = {Advances in Applied Probability},
number = {3},
publisher = {Applied Probability Trust},
pages = {749 -- 773},
year = {2012},
doi = {10.1239/aap/1346955263},
URL = {https://doi.org/10.1239/aap/1346955263}
}

\vspace{0.75cm}

\begin{center}
    \includegraphics[width=7cm]{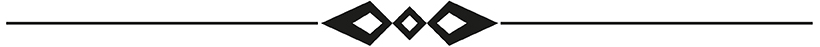}
\end{center}

\vspace{0.5cm}

\appendix
\section{Proof of the main results}
\label{app:proof}
The goal of this appendix is to prove the consistency (given in Theorem~\ref{main_ps}) and the asymptotic normality (given in Theorem~\ref{main_tcl}) of jump rate estimators under consideration.
The proof is given in detail below for estimator $\lambdak{n}(x)$ and follows the outline and notations presented in Subsection~\ref{ss:sketch}.
The proof for $\lambdaks{n}(x)$ follows exactly the same reasoning and is omitted for brevity.

\subsection{Square variation process}
\label{app:predict_square_var}
The predictable square variation process of $M_n=\left(M_n^{(1)},M_n^{(2)}\right)^\top$, introduced in \eqref{eq:def:Mn}, is given by
$$\left\langle M\right\rangle_n=\sum_{k=1}^n\left(\E\left[M_k^\top M_k|F_{k-1}\right]-M_{k-1}^\top M_{k-1}\right).$$
In this section we aim to establish the following lemma that describes the asymptotic behavior of $\left\langle M\right\rangle_n$.

\begin{lemma}\label{lem:asym:crochet}
When $n$ goes to infinity, the following equality holds almost surely,
$$
\left\langle M\right\rangle_n=\left[\begin{matrix}
    n^{1+\gamma}\omega_1(1+o(1))&O(n)\\O(n)&n\omega_2(1+o(1))
\end{matrix}\right],
$$
where $\omega_1>0$ and $\omega_2\geq0$.
\end{lemma}
\begin{proof}
We analyze the coefficients of the matrix separately.

\noindent
\textbf{$\bullet$ Study of $\left\langle M\right\rangle_n^{(1,1)}$}

A direct calculus shows that
\begin{eqnarray*}
    \left\langle M\right\rangle_n^{(1,1)}&=&\sum_{k=1}^n\var(M_k^{(1)}|F_{k-1})\\
    &=&\sum_{k=1}^n\var\left(\sum_{j=1}^k(A_j^{(1)}-B_j^{(1)})\Big|F_{k-1}\right) \\
    &=& \sum_{k=1}^n\left(\E\left[(A_k^{(1)})^2\big|F_{k-1}\right]-\E\left[A_k^{(1)}\big|F_{k-1}\right]^2\right),
\end{eqnarray*}
where the terms $A_k^{(i)}$ and $B_k^{(i)}$, $i\in \{1,2\}$ and $k\leq n$, are defined in Subsection~\ref{ss:sketch}.
Recalling that $P$ is the conditional density of $Z_k$ given $Z_{k-1}$, with a change of variable, one has
\begin{align*} \E\left[(A_k^{(1)})^2\big|F_{k-1}\right] &=
\frac{1}{h_n^2}\mathbb{E}\left[K^2\left(\frac{Z_k-h(y)}{h_n}\right)|Z_{k-1}\right]\\&=\frac{1}{h_n^2}\int\limits_{\R}K^2\left(\frac{z-h(y)}{h_n}\right)P(z \lvert Z_{k-1})\ \mathrm{d}z\\&=\frac{1}{h_n}\int\limits_{\R}K^2\left(u\right)P(h_nu+h(y)\lvert Z_{k-1})\ \mathrm{d}u,
\end{align*}
and analogously
\begin{equation*}
    \E\left[A_k^{(1)}\big|F_{k-1}\right]^2=\left(\int_\R K\left(u\right)P(h_nu+h(y)|Z_{k-1})\ud u\right)^2.
\end{equation*}

This yields
\begin{equation}\label{eq:def:crochetMn11}
\left\langle M\right\rangle_n^{(1,1)}=nT_n^{(1)}-T_n^{(2)},\end{equation}
where
\begin{eqnarray*}
T_n^{(1)}&=&\frac{1}{nh_n}\sum_{k=1}^n\int_\R K^2\left(u\right)P(h_nu+h(y)|Z_{k-1})\ud u,\\
T_n^{(2)}&=&\sum_{k=1}^n\left(\int_\R K\left(u\right)P(h_nu+h(y)|Z_{k-1})\ud u\right)^2.
\end{eqnarray*}
We shall investigate the limit behavior of $T_n^{(1)}$ and $T_n^{(2)}$. To this end, we define
\begin{eqnarray*}
\hat{T}_n^{(1)}&=&\frac{\tau^2}{n}\sum_{k=1}^n P(h(x)|Z_{k-1}),\\
\tilde{T}_n^{(1)}&=&\frac{1}{n}\sum_{k=1}^n\int_\R K^2\left(u\right)P(h_nu+h(x)|Z_{k-1})\ud u,
\end{eqnarray*}
where $\tau^2=\int_\R K^2(x)\ud x$. Using the fact that $P$ is Lipschitz as supposed in Assumptions~\ref{ass:technical} and $K$ has compact support by Assumptions~\ref{ass:kernel}, one has
\begin{eqnarray}
|\hat{T}_n^{(1)}-\tilde{T}_n^{(1)}| &\leq& \frac{1}{n}\sum_{k=1}^n\int_\R K^2(u)\left|P(h_nu+h(x)|Z_{k-1})-P(h(x)|Z_{k-1})\right|\ud u\nonumber \\
&\leq& \frac{\lip{P}}{n}\sum_{k=1}^n\int_\R K^2(u)|h_nu|\ud u\nonumber \\
&\leq& h_n\lip{P}\int_\R K^2(u)|u|\ud u~\to 0,\label{eq:Tns:same:limit}
\end{eqnarray}
when $n$ goes to infinity, because $h_n=n^{-\gamma}$. Furthermore, by virtue of the almost sure ergodic theorem \cite[Theorem~17.1.7]{MT09} (which we can apply by Assumption~\ref{ass:ergodicity}) and recalling that $\mu$ is the invariant measure of $P$,
$$ \hat{T}_n^{(1)} \toas \tau^2\E_\mu[P(h(x)|Z_0)] = \tau^2\mu(h(x)).$$
Since $\hat{T}_n^{(1)}$ and $\tilde{T}_n^{(1)}$ have the same limit by \eqref{eq:Tns:same:limit}, and writing $T_n^{(1)} = \tilde{T}_n^{(1)}/h_n$, one gets the following almost sure equivalent when $n$ goes to infinity,
\begin{equation}\label{eq:eq:nTn}
n T_n^{(1)} \sim n^{1+\gamma}\tau^2\mu(h(x)).
\end{equation}
In addition, since $P$ is bounded according to Assumptions~\ref{ass:technical},
$$T_n^{(2)} \leq \|P\|_\infty^2 n.$$
Together with \eqref{eq:def:crochetMn11} and \eqref{eq:eq:nTn}, one gets
$$
\frac{\left\langle M\right\rangle_n^{(1,1)}}{n^{1+\gamma}} \toas \tau^2\mu(h(x)),$$
which, together with $\mu(h(x))=\mu^-(x)/h'(x)$ where $\mu^-(x)>0$ and $h'(x)>0$ ($h$ is assumed to be strictly increasing), establishes the result for coefficient $\left\langle M\right\rangle_n^{(1,1)}$.

$\bullet$ \textbf{Study of $\left\langle M\right\rangle_n^{(2,2)}$}

One has
$$ \left\langle M \right\rangle_n^{(2,2)} = \sum_{k=1}^n\var(A_k^{(2)}|Z_{k-1}).$$
Each term, $\var(A_k^{(2)}|Z_{k-1})$, of the sum is a function of $Z_{k-1}$ and is independent of $n$. This contrasts with its counterpart, $\var(A_k^{(1)}|Z_{k-1})$, in $\left\langle M \right\rangle_n^{(1,1)}$, which exhibits an explicit dependence on $n$. In that context, one can directly apply the almost sure ergodic theorem \cite[Theorem~17.1.7]{MT09} and get
\begin{equation*}
    \frac{\left\langle M\right\rangle_n^{(2,2)}}{n}\toas\frac{1}{h'(x)^2\Delta(x)^2}\var_{\mu}\left(\I_{\{Z_0\leq x < Z_1^-\}}\right),
\end{equation*}
where the limit can be rewritten using $$ \var_{\mu}\left(\I_{\{Z_0\leq x < Z_1^-\}}\right)=\displaystyle\frac{1}{h'(x)^2\Delta(x)^2}\int_{\R_+} p(y)(1-p(y)) \mu(\ud y),$$
with $p(y) = \P(Z_0\leq x< Z_1^-|Z_0=y)$,
stating the expected result for $\left\langle M\right\rangle_n^{(2,2)}$.

$\bullet$ \textbf{Study of $\left\langle M\right\rangle_n^{(1,2)}=\left\langle M\right\rangle_n^{(2,1)}$}

With a change of variable, it is easy to see that both $$\E[A_k^{(1)}A_k^{(2)}|Z_{k-1}]$$
and
$$\E[A_k^{(1)}|Z_{k-1}]\,\E[A_k^{(2)}|Z_{k-1}]$$
are bounded by $\|P\|_\infty/(h'(x)\Delta(x))$. Together with
$$
\left\langle M\right\rangle_n^{(1,2)} =\sum_{k=1}^n\E[A_k^{(1)}A_k^{(2)}|Z_{k-1}]-\E[A_k^{(1)}|Z_{k-1}]\,\E[A_k^{(2)}|Z_{k-1}],
$$
this shows that the non-diagonal coefficients are in $O(n)$, as expected.
\end{proof}

\subsection{Consistency}
\label{proof_as}

If the martingale term $M_n/n$ and the remainder term $R_n=\left(R_n^{(1)},R_n^{(2)}\right)^\top$, defined in Subsection~\ref{ss:sketch}, go to $0$, then, in light of \eqref{eq:decompo:error}, $\lambdak{n}(x)$ is consistent. We tackle the two terms separately.

\subsubsection{Remainder term}

With a change of variable,
\begin{equation*}
R_n^{(1)}=\frac{1}{n}\sum_{k=1}^n\int_{\R}K\left(u\right)P\left(h_nu+h(x)|Z_{k-1}\right)\ud u-\mu(h(x)).
\end{equation*}
Then, using $\int_\R K(u)\ud u=1$ (see Assumptions~\ref{ass:kernel}),
\begin{equation}\label{eq:decompo:Rn1}
R_n^{(1)}=R_n^{(1,1)}+R_n^{(1,2)},
\end{equation}
where
\begin{eqnarray}
R_n^{(1,1)}&=&\frac{1}{n}\sum_{k=1}^n\int_{\R}K\left(u\right)\left[P\left(h_nu+h(x)|Z_{k-1}\right)-P\left(h(x)|Z_{k-1}\right)\right]\ud u,\nonumber\\
R_n^{(1,2)}&=&\frac{1}{n}\sum_{k=1}^n P\left(h(x)|Z_{k-1}\right)-\mu(h(x)).\label{eq:def:Rn12}
\end{eqnarray}
Since $P$ is Lipschitz according to Assumptions~\ref{ass:technical},
\begin{equation}
    \label{eq:Rn11}
|R_n^{(1,1)}|\leq h_n\lip{P}\int_{\R}K(u)|u|\ud u~\to0.
\end{equation}
In addition, $R_n^{(1,2)}$ almost surely vanishes by direct application of the almost sure ergodic theorem \cite[Theorem~17.1.7]{MT09}. This proves that $R_n^{(1)}$ tends to $0$. Furthermore,
\begin{equation*}
R_n^{(2)} = \frac{1}{nh'(x)\Delta (x)}\sum_{k=1}^n\I_{\{x\geq Z_{k-1}\}}\E[\I_{\{Z_k\geq h(x)\}}|Z_{k-1}] -\frac{D(x)}{h'(x)},
\end{equation*}
which directly goes to $0$, in light of the almost sure ergodic theorem again, and by definition \eqref{eq:def:D:proof} of $D(x)$.

\subsubsection{Martingale term}
\label{as martingale}

The submultiplicativity of the Frobenius norm, together with the asymptotic behavior of $\left\langle M\right\rangle_n$ stated in Lemma~\ref{lem:asym:crochet}, imply that
\begin{eqnarray*}
\frac{\|M_n\|^2}{n^2}
&=& \frac{\|\left\langle M\right\rangle_n^{1/2}\left\langle M\right\rangle_n^{-1/2} M_n\|^2}{n^2}\\
&\leq& \|\left\langle M\right\rangle_n^{-1/2}M_n\|^2 n^{\gamma-1}\max(\omega_1,\omega_2).
\end{eqnarray*}
By virtue of the law of large numbers for vector martingales \cite[Theorem~1.3.15]{D97}, for any $\eta>0$, when $n$ goes to infinity,
$$
\|\left\langle M\right\rangle_n^{-1/2}M_n\|^2=o\left(\frac{\log(\mathcal{T}_n)^{1+\eta}}{\mathcal{E}_n}\right),
$$
where $\calE_n$ ($\calT_n$, respectively) denotes the smallest eigenvalue (the trace, respectively) of $\left\langle M\right\rangle_n$. By Lemma~\ref{lem:asym:crochet}, it is clear that $\calT_n = O(n^{1+\gamma})$ and $\calE_n\sim \omega_2 n$ almost surely. All together, this establishes that $M_n/n$ almost surely vanishes, which ends the proof of the consistency of $\lambdak{n}(x)$.

\subsection{Asymptotic normality}
\label{proof_tcl}

Using \eqref{eq:decompo:error}, we derive a central limit theorem for $\lambdak{n}(x)$ by proving that the martingale term converges to a Gaussian distribution as the remainder term vanishes.

\subsubsection{Remainder term}
\label{rate_remainder}

We begin with $R_n^{(1)}$ written as the sum \eqref{eq:decompo:Rn1} of $R_n^{(1,1)}$ and $R_n^{(2,2)}$. With \eqref{eq:Rn11},
$$ n^{(1-\gamma)/2}R_n^{(1,1)} \leq n^{(1-3\gamma)/2}\lip{P}\int_\R K(u)|u|\ud u,$$
which goes to $0$ as long as $1/3<\gamma<1$. In addition, by virtue of \cite[Theorem~6.3.17]{D97} applied to functional $P$ (which satisfies Assumptions~\ref{ass:clt}) along Markov chain $Z_n$, one has almost surely
$$\left|\sum_{k=0}^n\left(P(h(x)|Z_k)-\E_{\mu}\left[P(h(x)|Z_0)\right]\right)\right|=o\left((n\ln(n)^{1+\gamma})^{\frac{r_1+r_2}{a_1}}\right).$$
With $\E_{\mu}\left[P(h(x)|Z_0)\right]=\mu(h(x))$, one gets
$$\left|\frac{1}{n^{\frac{1+\gamma}{2}}}\sum_{k=0}^n\left(P(h(x)|Z_k)-\mu(h(x))\right)\right|=o\left(n^{\frac{r_1+r_2}{a_1}-\frac{1+\gamma}{2}}\ln(n)^{(1+\gamma)\frac{r_1+r_2}{a_1}}\right).$$
Moreover $(r_1+r_2)/a_1 - (\gamma+1)/2\leq-\gamma/2<0$ because $(r_1+r_2)/a_1\leq 1/2$ in light of Assumptions~\ref{ass:clt}. Together with the definition \eqref{eq:def:Rn12} of $R_n^{(1,2)}$, this shows that $n^{(1-\gamma)/2}R_n^{(1,2)}$ almost surely vanishes. As a conclusion, $n^{(1-\gamma)/2}R_n^{(1)}$ almost surely goes to $0$.

On the other hand, the central limit theorem for Markov chains \cite[Theorem~17.5.3]{MT09} proves that $\sqrt{n}R_n^{(2)}$ weakly converges to a Gaussian distribution with mean $0$ and finite variance. Consequently, $n^{(1-\gamma)/2}R_n^{(2)}$ tends to $0$, and
\begin{equation}\label{eq:conv:rate:Rn}
n^{(1-\gamma)/2}R_n \toprob 0.
\end{equation}

\subsubsection{Martingale term}
\label{rate_martingale}

To establish the asymptotic normality of the martingale term $M_n/n$ with rate $n^{(1-\gamma)/2}$, we rely on the central limit theorem for vector martingales \cite[Corollary~2.1.10]{D97}. To this end, we need to check its conditions of applicability. First, in light of Lemma~\ref{lem:asym:crochet}, it is clear that
\begin{equation}\label{eq:conv:crochetMn}
\frac{\left\langle M\right\rangle_n}{n^{1+\gamma}} \toas 
\left[\begin{matrix}
    \tau^2\mu(h(x))&0\\
    0&0
    \end{matrix}\right].
\end{equation}
Second, we have to check Lindeberg's condition. To this end, we fix $k\leq n$. Using a change of variable and the definition \eqref{eq:def:Mn} of $M_n$,
\begin{eqnarray*}
\left|M_k^{(1)} - M_{k-1}^{(1)}\right| &=& \left|A_k^{(1)} - B_k^{(1)}\right| ~\leq~ n^{\gamma} \|K\|_{\infty} + \|P\|_{\infty},\\
\left|M_k^{(2)}-M_{k-1}^{(2)}\right|&=&
\left|A_k^{(2)}-B_k^{(2)}\right| ~\leq~ \frac{1}{h'(x)\Delta(x)}.
\end{eqnarray*}
Therefore, $\|M_k-M_{k-1}\|^2 = O(n^{2\gamma})$. In particular, with $\gamma<1$, one has $$\|M_k-M_{k-1}\| = o(n^{(1+\gamma)/2}).$$
For any $\varepsilon>0$, for $n$ large enough, the indicator function $\I_{\{\|M_k-M_{k-1}\|\geq\varepsilon n^{(1+\gamma)/2}\}}$ is thus almost surely null, which proves that 
$$\frac{1}{n^{1+\gamma}}\sum_{k=1}^n\E\left[\|M_k-M_{k-1}\|^2\I_{\{\|M_k-M_{k-1}\|\geq\varepsilon n^{(1+\gamma)/2}\}}\big|F_{k-1}\right]\toprob0.$$
As a consequence, the central limit theorem for vector martingales yields
\begin{equation}\label{eq:clt:Mn}
\frac{M_n}{n^{(1+\gamma)/2}} \todist \calN(0,\Sigma(x)),
\end{equation}
where $\Sigma(x)$ denotes the limit in \eqref{eq:conv:crochetMn}, which states the asymptotic normality of the martingale term.

Gathering \eqref{eq:decompo:error}, \eqref{eq:conv:rate:Rn}, \eqref{eq:conv:crochetMn} and \eqref{eq:clt:Mn} together, and by virtue of Slutsky's theorem, we obtain
$$
n^{(1-\gamma)/2}\left[\lambdak{n}(x)-\lambda(x)\right] \todist \calN\left(0,\frac{\tau^2(h'(x))^2\mu(h(x))}{D(x)^2}\right),$$
which proves the expected result. The formula of the variance is obtained from \eqref{eq:lambda:intro:k}.

\section{Additional figures from real data}\label{app:fig}

This section gathers Figures~\ref{fig:data:25:27} and \ref{fig:model:25:27} displaying both the observed data trajectories and the calibrated model at 25°C and 27°C. They constitute the analogues of Figures~\ref{fig:data:37} and \ref{fig:model:37} provided for the 37°C temperature condition.

\begin{figure}[ht]
    \centering
    \begin{minipage}{0.6\textwidth}
        \centering
        \includegraphics[height=4.4cm]{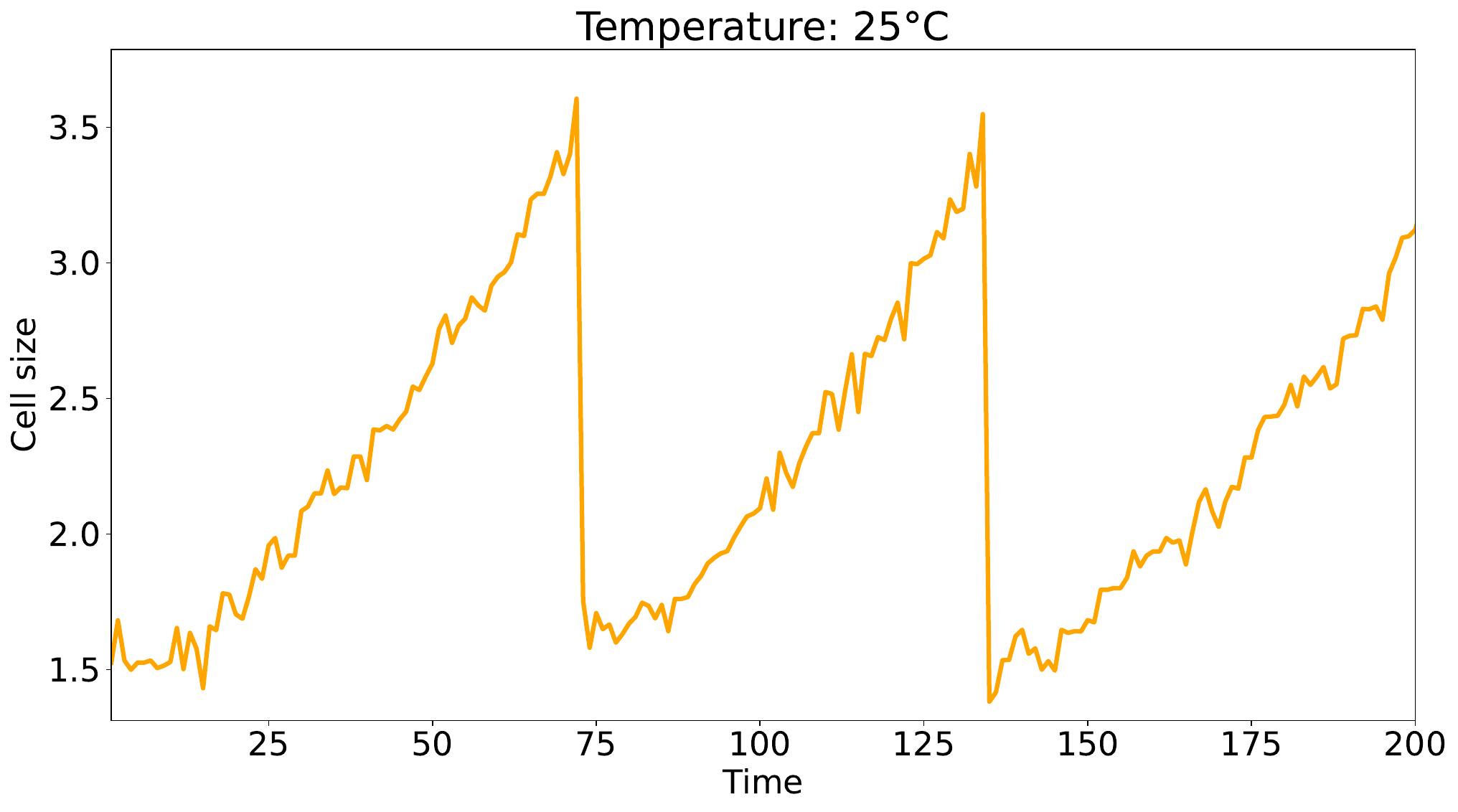}
    \end{minipage}%
    \begin{minipage}{0.35\textwidth}
        \centering
        \includegraphics[height=2.2cm]{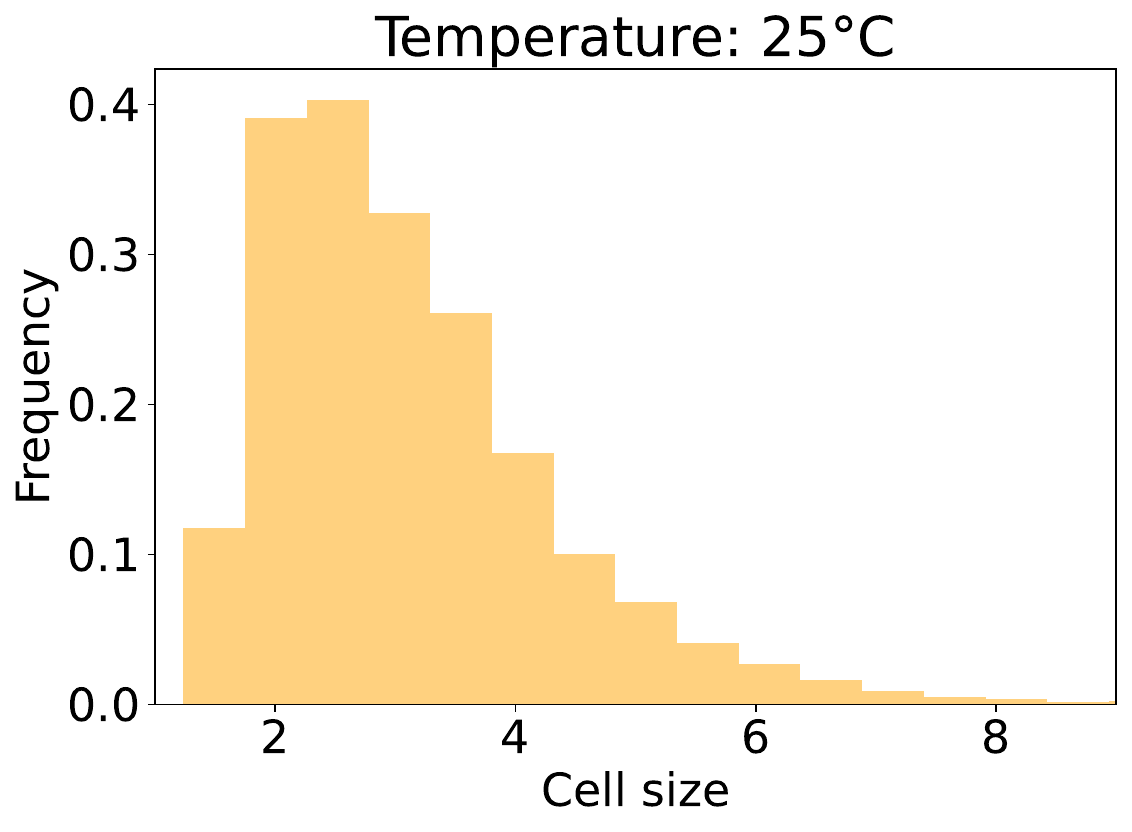} \\
        \includegraphics[height=2.2cm]{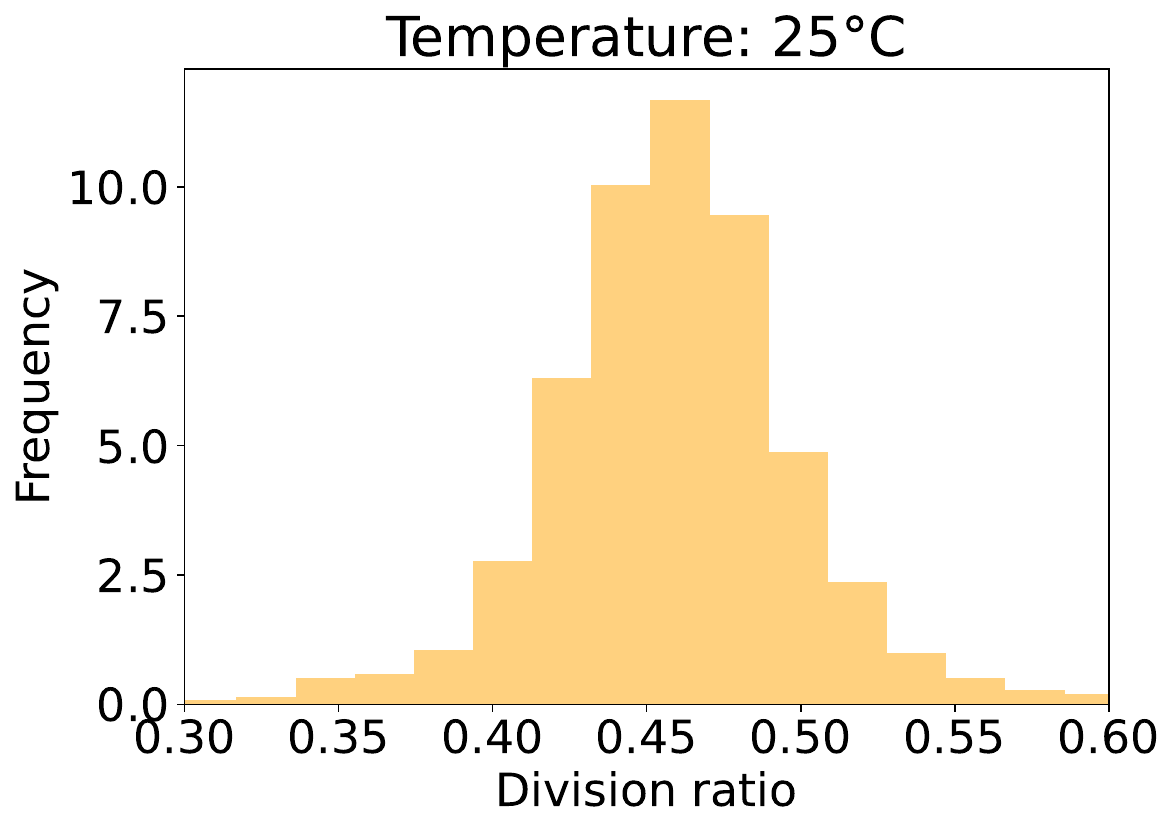}
    \end{minipage}

    \vspace{1cm}
    
    \begin{minipage}{0.6\textwidth}
        \centering
        \includegraphics[height=4.4cm]{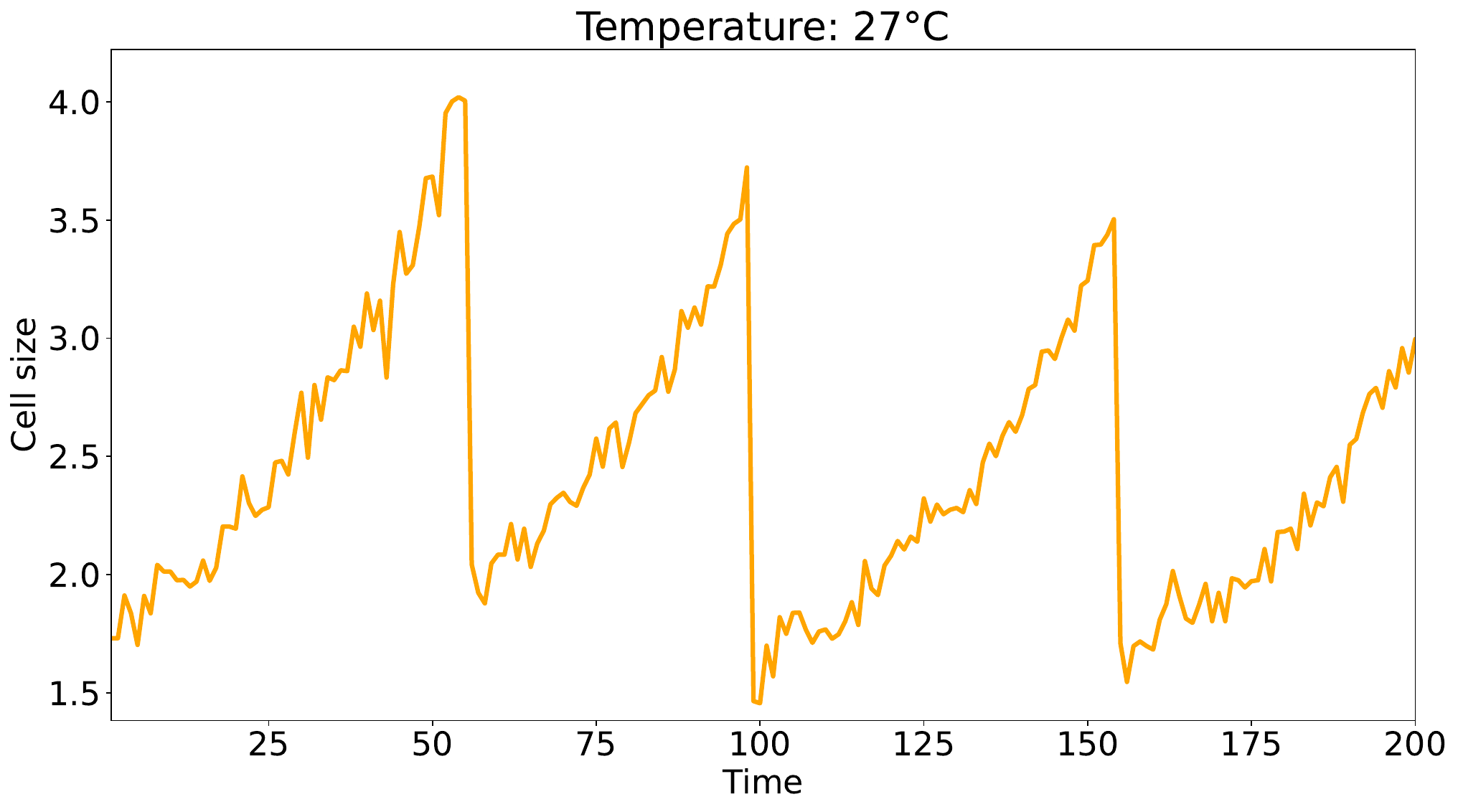}
    \end{minipage}%
    \begin{minipage}{0.35\textwidth}
        \centering
        \includegraphics[height=2.2cm]{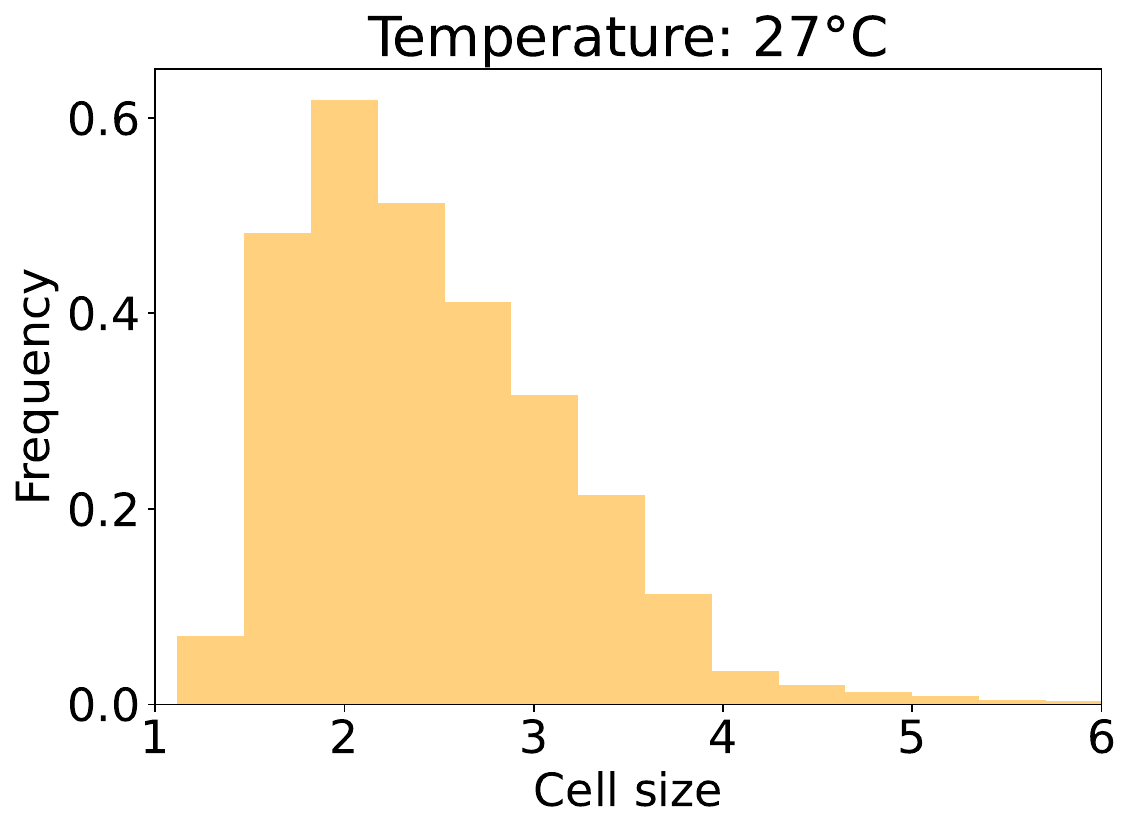} \\
        \includegraphics[height=2.2cm]{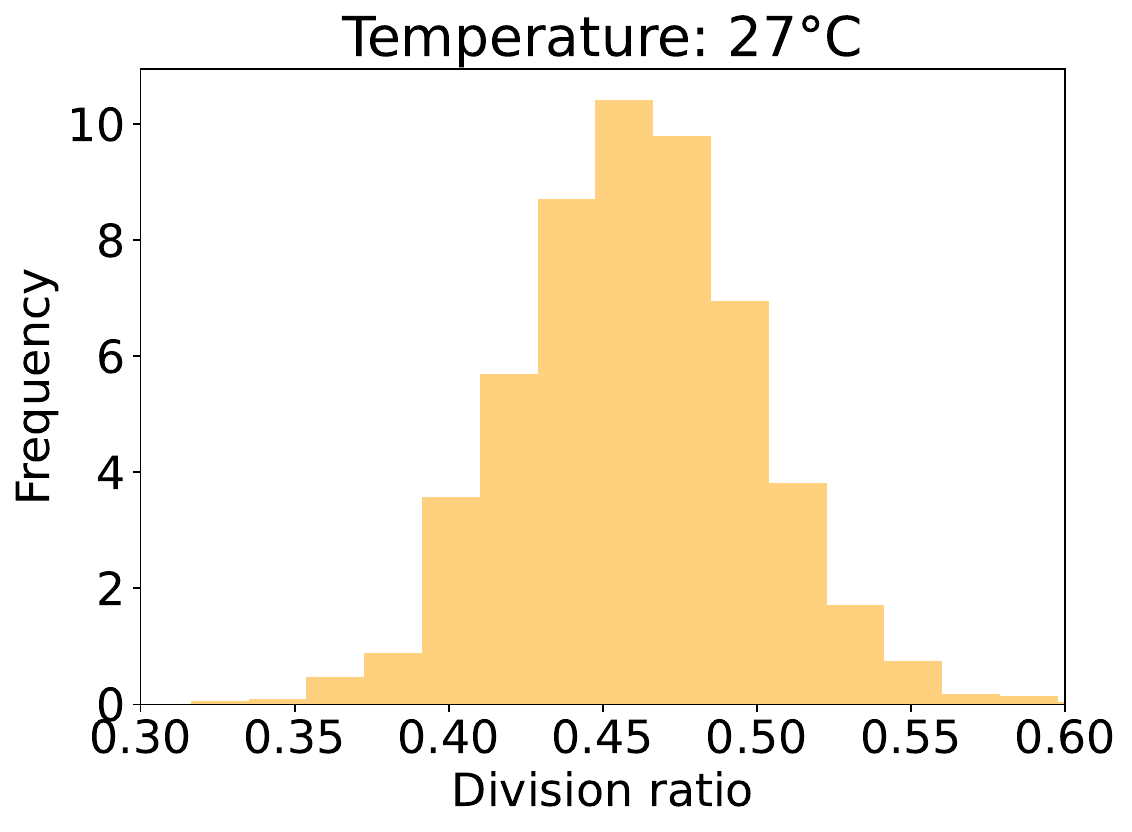}
    \end{minipage}
    
    \caption{Under temperature conditions of 25°C (upper panel) and of 27°C (lower panel), cell size measurements (unit of size is micrometer) until time 200 (unit of time is minute) from lineage \texttt{xy01\_02} (left), distribution of cell size from $10\,000$ consecutive measurements (mixing different lineages) (top right), and distribution of division ratio (bottom right).}
    \label{fig:data:25:27}
\end{figure}

\begin{figure}[ht]
    \centering
    
    \begin{minipage}{0.55\textwidth}
        \centering
        \includegraphics[height=4.4cm]{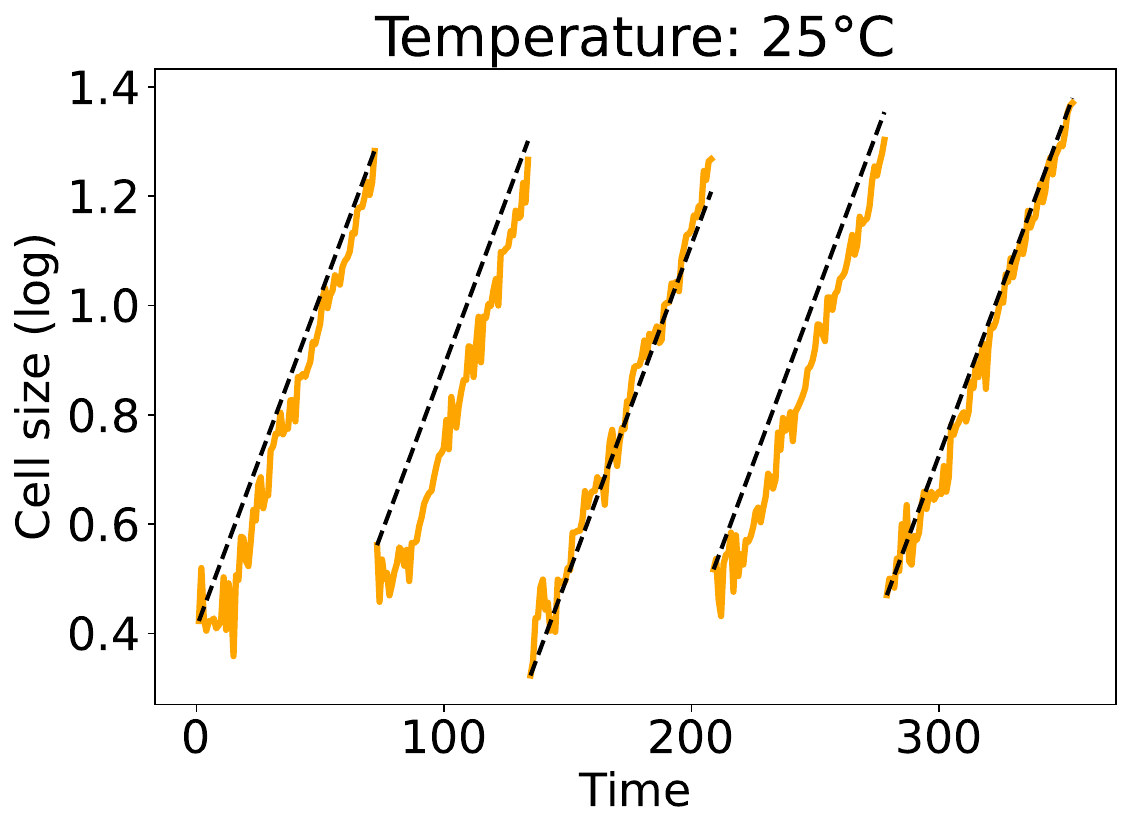}
    \end{minipage}%
    \begin{minipage}{0.2\textwidth}
        \centering
        \includegraphics[height=2.2cm]{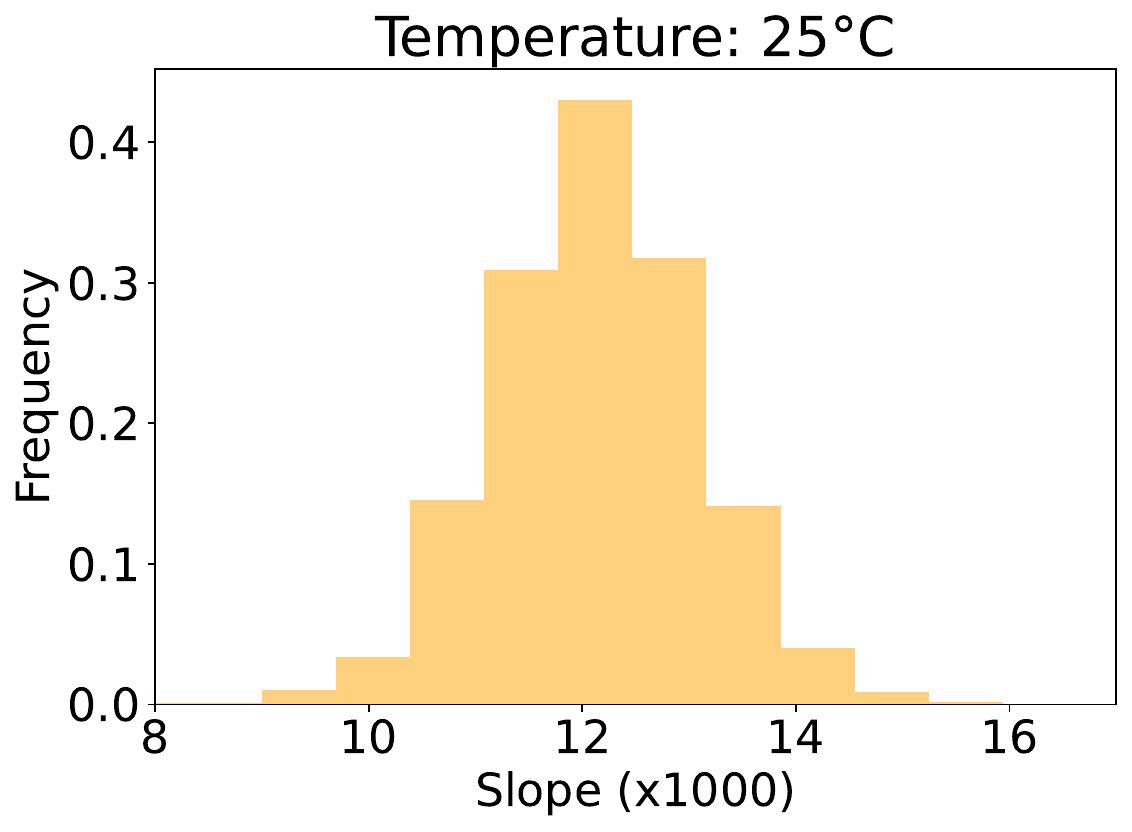} \\
        \includegraphics[height=2.2cm]{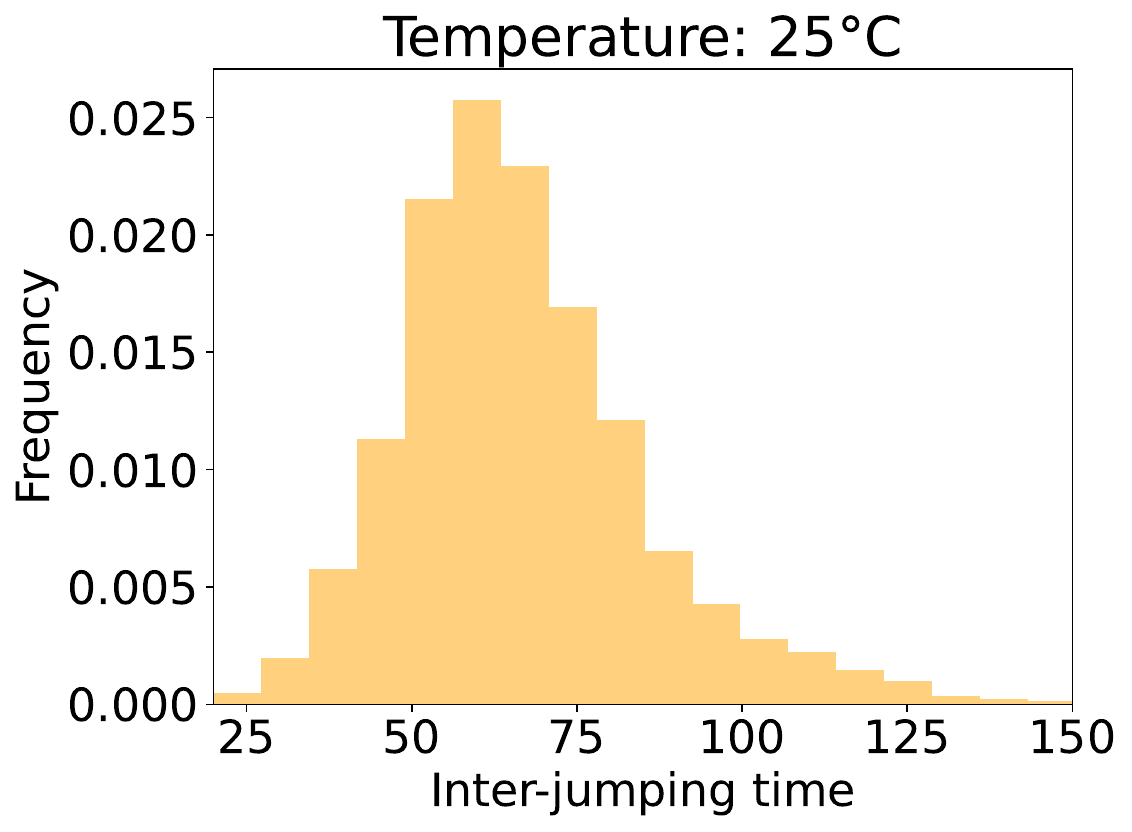}
    \end{minipage}
    \begin{minipage}{0.2\textwidth}
        \centering
        \includegraphics[height=2.2cm]{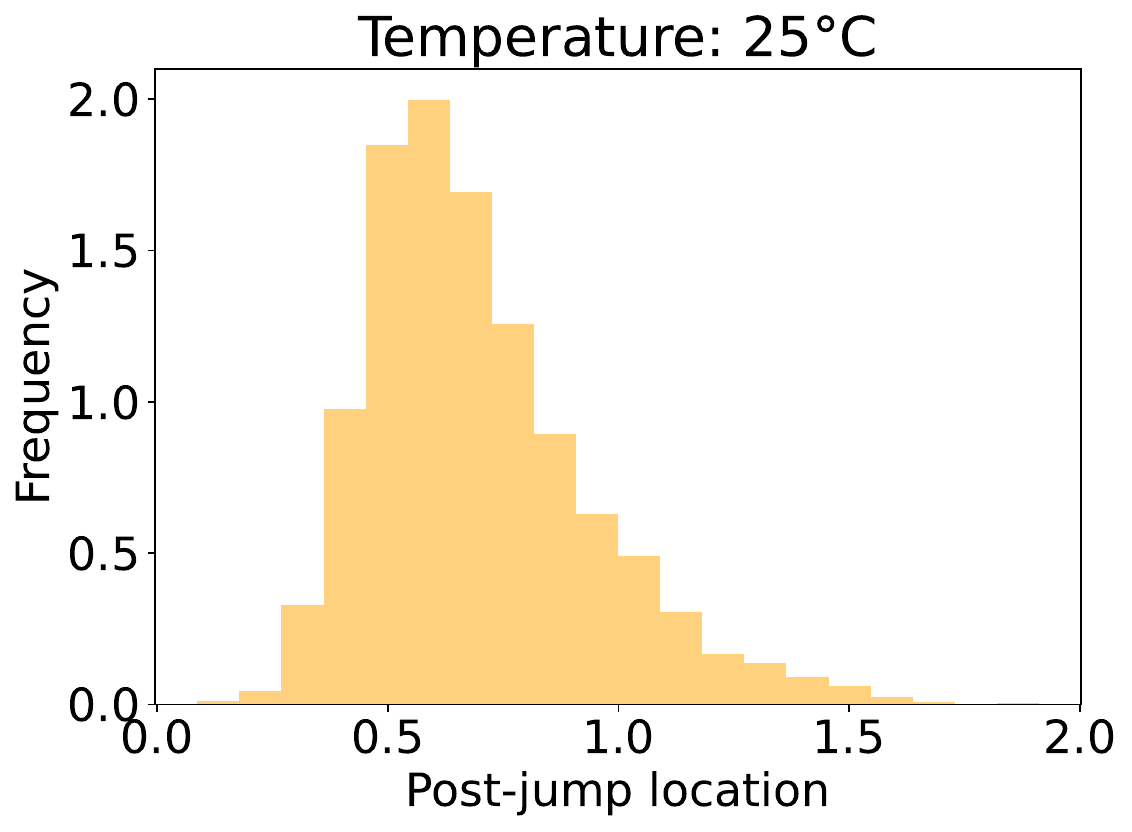} \\
        \includegraphics[height=2.2cm]{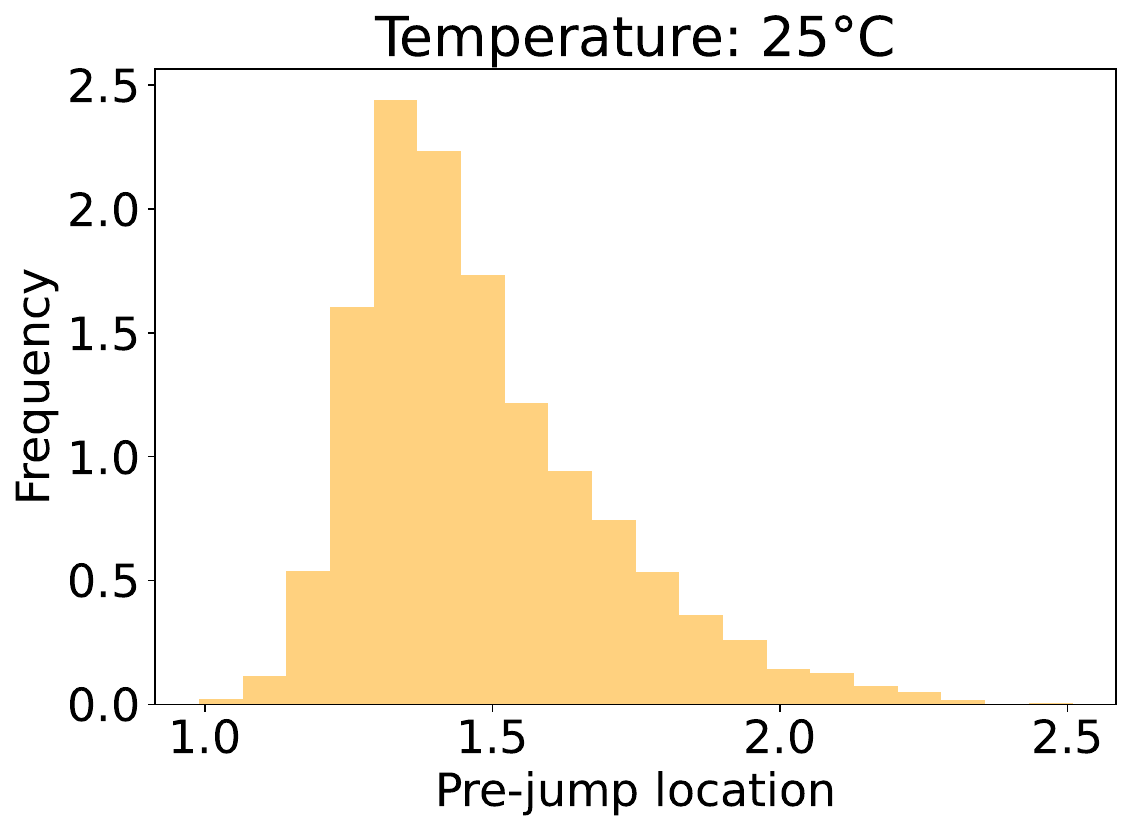}
    \end{minipage}
    
    \vspace{1cm}
    
    \begin{minipage}{0.55\textwidth}
        \centering
        \includegraphics[height=4.4cm]{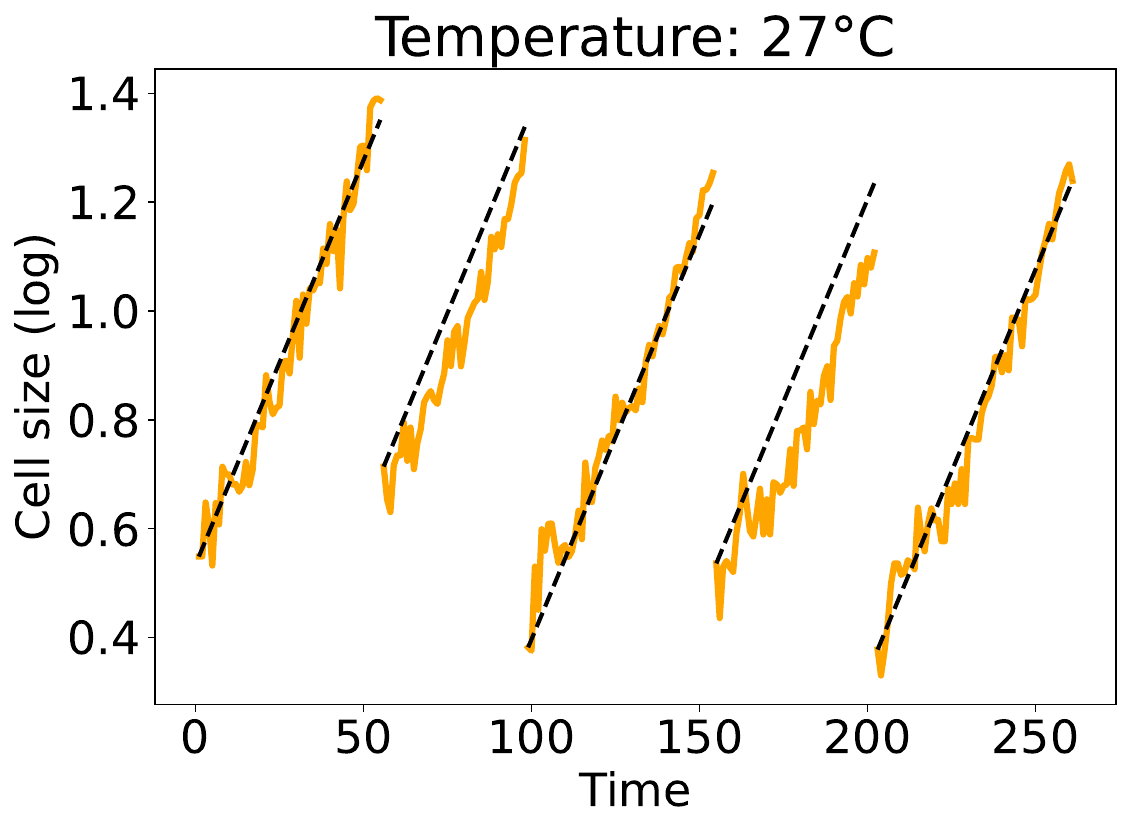}
    \end{minipage}%
    \begin{minipage}{0.2\textwidth}
        \centering
        \includegraphics[height=2.2cm]{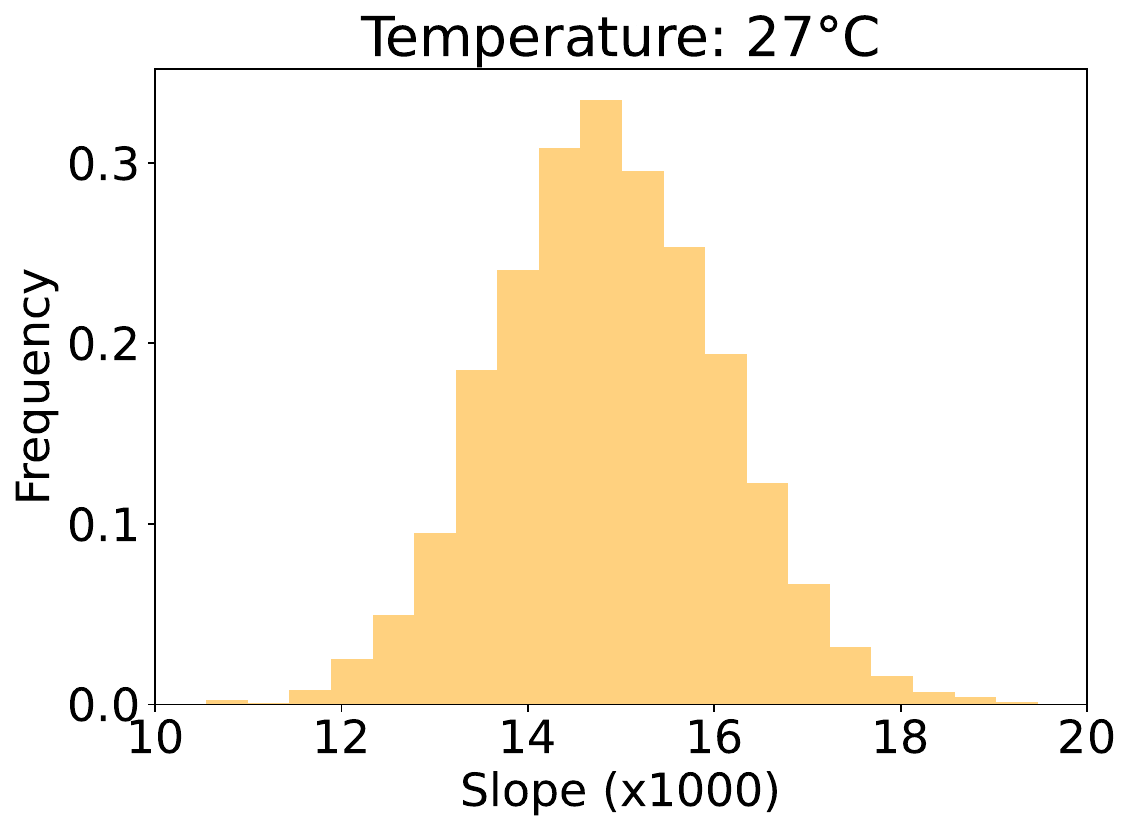} \\
        \includegraphics[height=2.2cm]{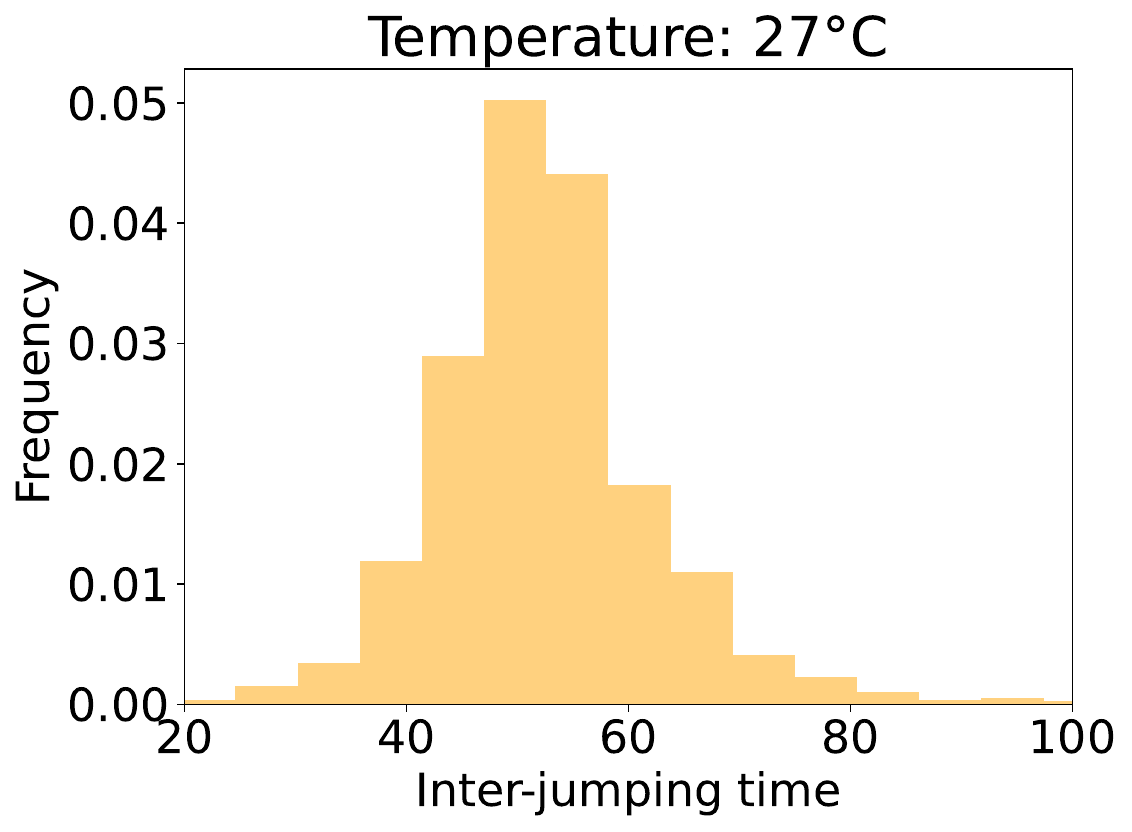}
    \end{minipage}
    \begin{minipage}{0.2\textwidth}
        \centering
        \includegraphics[height=2.2cm]{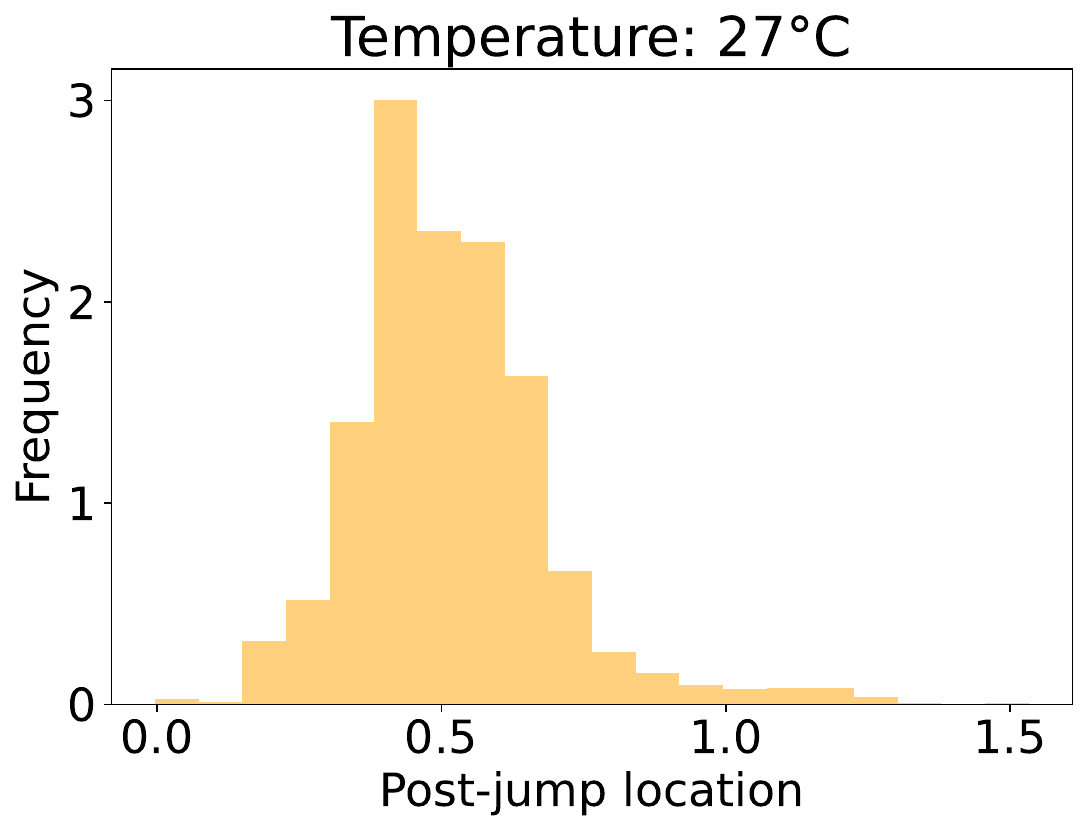} \\
        \includegraphics[height=2.2cm]{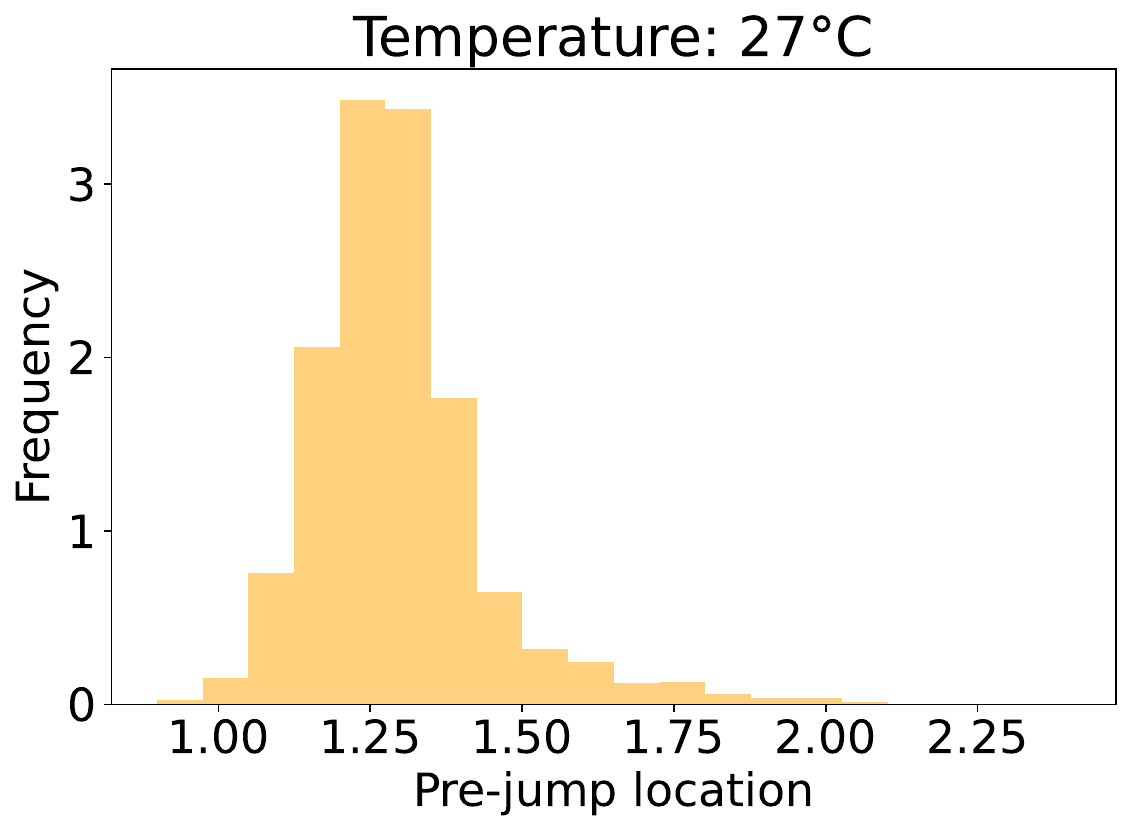}
    \end{minipage}
    
    \caption{Under temperature conditions of 25°C (upper panel) and of 27°C (lower panel), logarithm of cell size measurements before the fifth division event and fitted linear growth (left), distribution of estimated slope (top center), distribution of time between two consecutive division events (bottom center), distribution of logarithm of cell size at division time (top right), and distribution of logarithm of cell size just before division (bottom right).}
    \label{fig:model:25:27}
\end{figure}


\end{document}